\documentclass[12pt]{article}%
\usepackage{amssymb}
\usepackage{amsfonts}
\usepackage{amsmath}
\usepackage{amsthm}
\usepackage{bbm}
\usepackage[nohead]{geometry}
\usepackage[singlespacing]{setspace}
\usepackage[bottom]{footmisc}
\usepackage{indentfirst}
\usepackage{mathtools} % For \coloneqq
\usepackage{minitoc}  % for appendix-specific TOC
\usepackage{graphicx}%
\usepackage{rotating}
\usepackage{placeins}
\usepackage{siunitx}   % 【核心】用于数字对齐
\usepackage{tabularx}
\usepackage{enumerate}
\usepackage{booktabs}
\usepackage{subcaption}
\usepackage[flushleft]{threeparttable}
\usepackage{float}
\usepackage{verbatim}
\usepackage{epstopdf}
\usepackage{pdflscape}
\usepackage{multirow}
\usepackage{enumitem}
\usepackage{mathrsfs}
\usepackage{pdflscape}
\usepackage{bibunits} % <-- 1. 加载核心宏包
\usepackage{siunitx} 

\defaultbibliographystyle{plainnat}
\defaultbibliography{ref} % 假设您的 .bib 文件名为 ref.bib

\newtheorem{theorem}{Theorem}[section]
\newtheorem{lemma}{Lemma}

\newtheorem{proposition}{Proposition}

\newtheorem{assumption}{Assumption}[section]

\newcounter{remark}
\newenvironment{remark}[1][Remark]{\refstepcounter{remark}\begin{trivlist}
		\item[\hskip \labelsep {\bfseries #1 \theremark.\ }]}{\end{trivlist}}
 
\newcommand{\mP}{\mathbb{P}}
\newcommand{\E}{\mathbb{E}}
\newcommand{\Var}{\mathrm{Var}}
\newcommand{\Cov}{\mathrm{Cov}}

\newcommand{\Econd}[2]{\mathbb{E}\left[#1 \mid #2\right]}
\newcommand{\Prb}{\mathbb{P}}
\newcommand{\R}{\mathbb{R}}
\newcommand{\N}{\mathcal{N}} % Parameter space
\newcommand{\ind}[1]{\mathbf{1}_{\left\{#1\right\}}} % Indicator function
\newcommand{\norm}[1]{\left\|#1\right\|} % Generic norm
\newcommand{\Fnorm}[1]{\left\|#1\right\|} % Frobenius norm
\newcommand{\Eucnorm}[1]{\left|#1\right|} % Euclidean norm
\newcommand{\supT}{\sup_{\theta \in \N^*}} % Sup over theta
\newcommand{\maxi}{\max_{1 \le i \le n}} % Max over i
\newcommand{\Nb}{\N(\theta^*)} % Neighborhood over theta^*

\providecommand{\norm}[1]{\lVert#1\rVert}

\usepackage[usenames,dvipsnames]{color}
\usepackage[round]{natbib}
\usepackage[hyperfootnotes=false]{hyperref}
\hypersetup{
	colorlinks,
	citecolor=Blue,
	linkcolor=Blue,
	urlcolor=Blue}

\numberwithin{equation}{section}

\begin{document}
	\begin{bibunit}

	\begin{spacing}{1.2}
		
		\title{\vspace{-2.0cm}%Kernel Minimum Distance: Unified Estimation and Specification Testing for Conditional Moment Restrictions
        Kernel Minimum Distance Estimation and Testing with Conditional Moment Restrictions: A Unified Framework\thanks{The authors contributed equally to this work and are listed in alphabetical order.}}
        %Kernel minimum distance estimation and testing with conditional moment restrictions
        %Unified Kernel Minimum Distance Estimation and Specification Testing for Conditional Moment Restrictions
        %A Unified Kernel Framework for Conditional Moment Restrictions
        %Kernel Minimum Distance: A Unified Approach to Conditional Moment Restrictions
        \vspace{-0.5cm}
		\author{Yuhao Li\thanks{School of Mathematics and Physics, Xi'an Jiaotong--Liverpool University, Suzhou, China. Email: \texttt{yuhao.li@xjtlu.edu.cn}.}\\Xi'an Jiaotong--Liverpool University \and Haokun Lu\thanks{Department of Business Statistics and Econometrics,
				Guanghua School of Management, Peking University, Beijing, 100871, China. Email: \texttt{hklu25@stu.pku.edu.cn}.}\\Peking University \and Xiaojun Song\thanks{Corresponding author. Department of Business Statistics and Econometrics, Guanghua School of Management, %and Center for Statistical Science, 
                Peking	University, Beijing, 100871, China. Email: \texttt{sxj@gsm.pku.edu.cn}. This work was supported by the National Natural Science Foundation of China [Grant Numbers 72373007 and 72333001]. The author also gratefully acknowledges the research support from the Center for Statistical Science of Peking University, China, and the Key Laboratory of Mathematical Economics and Quantitative Finance (Peking University) of the Ministry of Education, China.}\\Peking University}
		\maketitle

        \vspace{-0.5cm}
		\begin{abstract}

        We propose a unified Kernel Minimum Distance (KMD) framework for estimating and testing models defined by conditional moment restrictions.  By embedding conditional moments into a Reproducing Kernel Hilbert Space (RKHS), we construct a closed-form $V$-statistic objective function that quantifies the distance from the restrictions. We establish the $\sqrt{n}$-consistency and asymptotic normality of the associated minimum distance estimator. Within this framework, the minimized objective function naturally yields a consistent omnibus specification test. Unlike projection-based methods that require auxiliary nonparametric estimation for Neyman orthogonalization, our test inherently captures the estimation effect via a projected kernel structure. We derive asymptotic properties of the test statistics under the null hypothesis, the alternative hypothesis, and a sequence of local alternatives converging to the null at the parametric rate $n^{-1/2}$. The validity of a computationally simple multiplier bootstrap is established to facilitate inference. Simulation results demonstrate robust finite-sample performance, and the framework is illustrated by analyzing Engel curves using UK Family Expenditure Survey data.
            \par\vspace{0.5em} 
		
        \noindent\textbf{Keywords:} Kernel Minimum Distance, Conditional Moment Restrictions, Reproducing Kernel Hilbert Space
		
		\noindent\textbf{JEL Classification Number:} C12; C13; C14; C26.
		\end{abstract}
	\end{spacing}
    
	\thispagestyle{empty}
	
	\newpage
	
	\normalsize
	%\onehalfspacing
	
	\section{Introduction}\label{sec:intro}

Conditional moment restrictions (CMRs) serve as a fundamental framework in econometric analysis, encompassing a wide range of empirical settings, including instrumental variable estimation, nonlinear regression, and structural modeling. In these applications, economic theory typically postulates a parametric structural relationship that implies moment conditions of the form:
\begin{equation}
    \label{H_0}
    \mathbb{E}\left[\rho\left(Z,\theta_0\right)|X\right]=0 \quad \text{almost surely (a.s.)},
\end{equation}
for some unknown true parameter $\theta_0 \in \Theta$. Here, $Z \in \mathcal{Z}$ and $X \in \mathcal{X}$ denote the observable variables, and $\rho: \mathcal{Z} \times \Theta \to \mathbb{R}^q$ is a pre-specified generalized residual function known up to the parameter $\theta \in \Theta \subset \mathbb{R}^p$, which encodes the structural restrictions imposed by the model. Consequently, empirical researchers face two closely related tasks: consistently estimating the structural parameter $\theta_0$ and assessing whether the assumed structural model is correctly specified. Without the latter, the economic interpretation of the estimated parameters remains questionable.

The standard approach to estimation under conditional moment restrictions is the Generalized Method of Moments (GMM; \citealp{hansen1982large}), which approximates the conditional restriction by a finite vector of unconditional moments. While computationally convenient and widely used, this finite reduction may lead to fundamental identification problems. When the conditioning variables have support with infinite cardinality, the conditional restriction implies an infinite collection of unconditional moment conditions, and satisfying only a finite subset does not ensure global identification in nonlinear models, even when optimal instruments are employed \citep[see][Examples~1--2]{dominguez2004consistent}. As a consequence, GMM estimators may be inconsistent even when the underlying conditional restriction is correctly specified. A closely related limitation concerns specification testing. The classical $J$-test \citep{sargan1958estimation,hansen1982large} is not omnibus and is consistent only against violations of the selected moment conditions \citep{newey1985generalized}.

To restore global identification, the literature has evolved to exploit the infinite information inherent in conditional restrictions. Prominent contributions in this vein include sieve-based methods that approximate conditional moments using an expanding sequence of basis functions \citep[e.g.,][]{donald2003empirical} and kernel-based approaches that rely on local smoothing \citep[e.g.,][]{kitamura2004empirical}. Closely related to the latter is \citet{lavergne2013smooth}, who develop a smooth minimum distance estimator and uniform-in-bandwidth inference for conditional estimating equations.
 While many of these methods are theoretically capable of achieving semiparametric efficiency, their implementation often involves choosing a dimension, bandwidth, or regularization parameter to manage the bias-variance trade-off. Consequently, implementation may entail carefully selecting such objects; for instance, determining the optimal number of basis functions in sieve estimation requires complex data-driven procedures \citep[see][]{donald2009choosing}. A distinct strategy leverages the Integrated Conditional Moment (ICM) approach, pioneered by \cite{bierens1982consistent}. This framework rests on the fundamental result that the original conditional restriction is equivalent to satisfying a continuum of unconditional moment restrictions, obtained by integrating the conditional moments with respect to a weighting function. This equivalence has been extensively developed for consistent estimation \citep[e.g.,][]{carrasco2000generalization, dominguez2004consistent} and omnibus specification testing \citep[e.g.,][]{bierens1990consistent, bierens1997asymptotic, delgado2006consistent, escanciano2006consistent, dominguez2015simple}. However, despite their theoretical appeal, the practical implementation of ICM methods remains challenging. The test statistics are typically constructed by integrating empirical processes over a continuum of indices or weighting functions, a step that may lack an analytic solution. Closed-form expressions exist only for restrictive choices of weighting functions \citep[e.g.,][]{bierens1982consistent, escanciano2006consistent}, otherwise necessitating computationally intensive numerical integration. Furthermore, because the asymptotic distributions of these omnibus statistics are often non-standard and data-dependent, obtaining valid critical values for many of these procedures requires complex adjustments or computationally intensive resampling methods, which may limit their routine use in applied econometric work.

In the spirit of the ICM approach, we propose a unified Kernel Minimum Distance (KMD) framework for both the estimation and specification testing of models defined by conditional moment restrictions. Rather than working with a pre-specified indexed family of transformations or weighting functions, KMD takes the relevant test functions to be the unit ball of a vector-valued Reproducing Kernel Hilbert Space (RKHS), so that the criterion measures the largest normalized violation of the conditional moment restriction over this RKHS class. The kernel's reproducing property then turns this infinite-dimensional criterion into a closed-form pairwise $V$-statistic objective function. Thus, the kernel plays two roles: it determines the class of moment violations against which the conditional restriction is assessed, and it delivers a tractable sample objective for estimation and testing. We establish the $\sqrt n$-consistency and asymptotic normality of the resulting estimator and discuss a one-step Newton--Raphson update toward semiparametric efficiency. Building directly on the same estimation framework, we construct a consistent omnibus specification test based on the minimized objective function. Under the null hypothesis, the limiting distribution of the test statistic automatically accounts for the first-order estimation effect via a projected-kernel structure, yielding an automatic Neyman-type orthogonalization without the need to construct a separate orthogonalized testing process. We further establish consistency against fixed alternatives, examine power against sequences of local alternatives converging to the null at the parametric rate $n^{-1/2}$, and establish the validity of a simple multiplier bootstrap procedure that is computationally efficient and avoids repeated parameter re-estimation or resampling observations.

Our framework is related to two strands of work: RKHS-based measures of distributional and moment discrepancies, and minimum-distance approaches to conditional moment restrictions in econometrics. In statistics and machine learning, the use of RKHS norms to measure discrepancies is closely related to the maximum mean discrepancy of \citet{gretton2006kernel,gretton2012kernel}; for a comprehensive review of kernel mean embeddings and their applications, see \citet{muandet2017kernel}. For conditional moment restrictions, on the estimation front, \citet{muandet2020kernel} formulate the Maximum Moment Restriction (MMR) principle, and \citet{zhang2023instrumental} recently apply a related kernel maximum-moment loss to instrumental variable regression and establish basic asymptotic properties. Their framework is primarily formulated as a learning problem: the emphasis is on constructing and optimizing a kernel-based empirical risk for IV regression, with particular attention to algorithmic implementation and predictive performance. We complement this literature by tailoring the estimator to structural inference and providing valid asymptotic inference based on consistent variance estimation.

The closest econometric precursor on the estimation side is the smooth minimum distance approach of \citet{lavergne2013smooth}. Their criterion is derived from local smoothing of conditional estimating equations, and one of their key insights is that consistent estimation of the structural parameter does not require consistent nonparametric estimation of the local conditional moment itself. Our criterion is closely related in continuous Euclidean settings with translation-invariant smoothing kernels, but it starts from a different primitive object: the RKHS norm of an embedded conditional moment. This formulation makes explicit the role of the kernel in measuring moment violations over the conditioning space and allows kernels to be specified directly for the variables being conditioned on. For example, when the conditioning variables contain both continuous covariates and discrete indicators, as is common in econometric practice, one may combine a Gaussian kernel for the former with a categorical kernel for the latter.

On the testing front, existing literature similarly harnesses the geometry of RKHS to aggregate conditional moment restrictions. \cite{muandet2020kernel} primarily focus on regression settings where the structural parameter $\theta_0$ is assumed known. Relatedly, \cite{sancetta2022testing} develops RKHS score-type tests for functional subspace restrictions in the presence of high- or infinite-dimensional nuisance parameters, using estimated projections of test functions to remove nuisance effects. \cite{escanciano2024gaussian} extends the RKHS/Gaussian-process approach to the composite hypothesis setting but adopts a decoupling strategy based on an exogenous estimator. This estimator-agnostic framework allows for estimators with slower convergence rates and dispenses with the requirement of an asymptotic linear Bahadur representation. While this flexibility is advantageous for high-dimensional nuisance parameters, it can require additional orthogonalization steps in standard structural applications. Specifically, achieving such robustness may require an auxiliary estimate of the relevant conditional score or nuisance projection whenever the gradient is not measurable with respect to the $\sigma$-algebra generated by the conditioning covariates -- a common feature in models with endogeneity, such as instrumental variable regressions. In contrast, our unified KMD framework defines the test statistic directly as the minimized value of the estimation objective function. In this setting, the parameter estimation uncertainty is captured by the statistic's geometry through a projected kernel derived from the estimator's first-order conditions, thereby enabling asymptotic inference without constructing a separate orthogonalized testing process.

Conceptually, our approach aligns with the unified paradigm established by \cite{dominguez2004consistent} and \cite{dominguez2015simple}. These works elegantly link estimation and testing by interpreting the minimized value of an integrated moment objective function as an omnibus specification test. However, their reliance on indicator-weighted empirical processes renders the methodology susceptible to the ``curse of dimensionality'' and leads to rapid power loss as the number of covariates increases. We advance this unified framework by substituting indicator functions with reproducing kernels. By defining the estimator $\widehat{\theta}_n$ as the minimizer of the KMD objective, we retain the computational simplicity of their framework, particularly in nonlinear and endogenous settings, thereby yielding a statistic that can be interpreted as a generalized Sargan test. Crucially, this kernel-based formulation leverages geometric adaptability to significantly enhance the robustness of both estimation and testing in multivariate settings.

The remainder of the paper is organized as follows. Section \ref{sec:estimation} presents the identification strategy via an RKHS embedding, defines the KMD estimator, establishes its asymptotic normality, develops asymptotic inference for individual parameters and for linear and nonlinear restrictions, and discusses a two-step update that attains semiparametric efficiency. Section \ref{sec:test} develops an omnibus specification test based on the minimized objective function, derives its limiting distribution under the null hypothesis, establishes consistency against fixed alternatives, and analyzes power against local alternatives. Section \ref{sec:mb} introduces a computationally efficient multiplier bootstrap procedure and establishes its asymptotic validity. Section \ref{sec:simu} investigates the finite-sample performance of the proposed framework through extensive Monte Carlo simulations. Section \ref{sec:real} illustrates the method's empirical relevance by revisiting the structural analysis of Engel curves. Section \ref{sec:con} concludes. Proofs of the main results and additional numerical and empirical analyses are provided in the Online Supplementary Material.

\textsc{Notation}. Throughout the paper, $\|\cdot\|$ denotes the Euclidean norm for vectors and the Frobenius norm for matrices. The superscript $\prime$ denotes the transpose of a vector or matrix. $\mathbf{I}_q$ represents the $q \times q$ identity matrix. We denote convergence in probability and in distribution by $\xrightarrow{p}$ and $\xrightarrow{d}$, respectively. The stochastic order symbols $O_p(\cdot)$ and $o_p(\cdot)$ are defined in the standard sense.

\section{Estimation}
\label{sec:estimation}

\subsection{Identification Strategy via RKHS Embedding}
\label{subsec:identification}

In this section, we consider the estimation of the structural parameter $\theta_0$ defined by the conditional moment restrictions in \eqref{H_0}. In the spirit of the ICM approach, we embed the conditional moment restrictions into a vector-valued RKHS. By the Moore--Aronszajn theorem, every symmetric positive definite kernel $k$ on $\mathcal X$ induces a unique Hilbert space $\mathcal H_k$ of real-valued functions on $\mathcal X$. Its reproducing property, $h(x)=\langle h,k(\cdot,x)\rangle_{\mathcal H_k}$ turns evaluation at $x$ into an inner product. This property allows function-indexed moment restrictions to be represented through kernel evaluations. We do not require the unknown conditional mean $x\mapsto \mathbb E[\rho(Z,\theta)| X=x]$ to belong to the RKHS. Rather, the RKHS provides the class of test functions and the geometry used to measure violations of the conditional moment restriction. Identification rests on the equivalence between the conditional restriction \eqref{H_0} and the vanishing norm of the corresponding moment embedding within the RKHS. To formalize this framework, we introduce the following assumptions.
\begin{assumption}
\label{ass:cons}
\quad
\begin{itemize}
    \setlength{\itemsep}{0pt} 
    \setlength{\parskip}{0pt} 
    \setlength{\parsep}{0pt}
    \item[(i)] The observed data $\{W_i\}_{i=1}^n$ is an i.i.d. sample of a random vector $W$. The vector $W_i$ contains subvectors $Z_i$ and $X_i$, which are not necessarily disjoint.

    \item[(ii)] The function $\rho(z, \theta): \mathcal{Z} \times \Theta \to \mathbb{R}^q$ is continuous at each $\theta \in \Theta$ with probability one, and satisfies $\mathbb{E}[\sup_{\theta \in \Theta}\|\rho(Z, \theta)\|^2]<\infty$.
    
    \item[(iii)] The parameter space $\Theta \subset \mathbb{R}^p$ is compact. There exists a unique $\theta_0 \in \Theta$ such that $\mathbb{E}[\rho(Z, \theta_0)| X]=0$ a.s.
    
    \item[(iv)] The conditioning space $\mathcal X$ is a separable metric space equipped with its Borel $\sigma$-field. The kernel $k:\mathcal X\times\mathcal X\to\mathbb R$ is symmetric, continuous, bounded, and positive definite. Moreover, $k$ is strictly positive definite over finite signed Borel measures: for every nonzero finite signed Borel measure $\nu$ on $\mathcal X$,
    $$
    \iint_{\mathcal X\times\mathcal X}
    k(x,\widetilde x)\,d\nu(x)\,d\nu(\widetilde x)>0 .
    $$
    We define the $q\times q$ matrix kernel as $K(x_1,x_2)\coloneqq k(x_1,x_2)\cdot \mathbf I_q$.
\end{itemize}
\end{assumption}

Assumptions \ref{ass:cons}(i) and (ii) are standard regularity conditions. Assumption \ref{ass:cons}(iii) imposes global identification of the structural parameter. Assumption \ref{ass:cons}(iv) imposes injectivity of the RKHS embedding at the level of finite signed measures,\footnote{This requirement is stronger than the usual characteristic-kernel property, which concerns injectivity of the kernel mean embedding over probability measures. See \citet{sriperumbudur2011universality} for detailed characterizations.} which is the relevant requirement for conditional moment restrictions. Indeed, for each component of the generalized residual, the map $A\mapsto E[\rho_l(Z,\theta)\mathbf 1\{X\in A\}]$ defines a finite signed measure on $\mathcal X$. Thus, the vanishing of the RKHS embedding recovers the conditional moment restriction only if the kernel embedding is injective for such signed measures. This formulation is not tied to a Euclidean conditioning space or to the existence of a Lebesgue density for $X$. In the standard
continuous-covariate setting, the condition is satisfied by many familiar bounded translation-invariant kernels whose spectral measure has full support, including Gaussian, Laplace, and Mat\'ern kernels. In settings with discrete or mixed conditioning variables, one may instead use kernels defined directly on the relevant support, for instance, products of a continuous kernel and a strictly positive-definite categorical kernel. Hence Assumption \ref{ass:cons}(iv) accommodates the types of continuous, discrete, and mixed conditioning variables commonly encountered in econometric applications.

Let $\mathcal H_K$ denote the vector-valued RKHS associated with the matrix kernel $K$ defined in Assumption \ref{ass:cons}(iv).\footnote{Since the matrix kernel has a diagonal structure
$K(x,\widetilde x)=k(x,\widetilde x)\mathbf I_q$, the vector-valued RKHS
$\mathcal H_K$ can be identified with the Hilbert direct sum of $q$ copies
of the scalar RKHS $\mathcal H_k$ generated by $k$. Equivalently,
$\mathcal H_K
=
\left\{
h=(h_1,\ldots,h_q)^\prime:
h_l\in\mathcal H_k,\ l=1,\ldots,q
\right\}$, with inner product
$\langle h,g\rangle_{\mathcal H_K}
=
\sum_{l=1}^q \langle h_l,g_l\rangle_{\mathcal H_k}$.
Thus each element $h\in\mathcal H_K$ is a vector-valued function
$h:\mathcal X\to\mathbb R^q$.} We next relate the KMD construction to the traditional ICM approach. The ICM approach converts a conditional moment restriction into a collection of unconditional moment restrictions indexed by test functions of the conditioning variables. Specifically, if $\mathbb E[\rho(Z,\theta)\mid X]=0$ a.s., then for every admissible vector-valued test function $h$,
$$
\mathbb E[\rho(Z,\theta)^\prime h(X)]
=
\mathbb E\!\left[
\mathbb E[\rho(Z,\theta)\mid X]^\prime h(X)
\right]
=0.
$$
The converse implication requires the test-function class to be sufficiently rich. Our KMD formulation implements this idea by taking the admissible test-function class to be the unit ball of $\mathcal H_K$. For each fixed $\theta\in\Theta$, define the linear functional $\mathcal L_\theta:\mathcal H_K\to\mathbb R$ by $\mathcal L_\theta h\coloneqq \mathbb E[\rho(Z,\theta)^\prime h(X)]$. Assumptions \ref{ass:cons}(ii) and (iv) guarantee that $\mathcal L_\theta$ is bounded. Its operator norm is
$$
\|\mathcal L_\theta\|_{\mathrm{op}}
\coloneqq
\sup_{\|h\|_{\mathcal H_K}\le 1}
|\mathcal L_\theta h|
=
\sup_{\|h\|_{\mathcal H_K}\le 1}
\left|
\mathbb E[\rho(Z,\theta)^\prime h(X)]
\right|.
$$
Thus, $\|\mathcal L_\theta\|_{\mathrm{op}}$ measures the largest normalized unconditional moment violation over the RKHS class. When the conditional moment restriction holds at $\theta$, $\mathcal L_\theta$ is the zero functional and $\|\mathcal L_\theta\|_{\mathrm{op}}=0$. Conversely, the implication from $\|\mathcal L_\theta\|_{\mathrm{op}}=0$ back to the conditional moment restriction relies on the richness of the RKHS class, ensured here by the signed-measure injectivity condition in Assumption \ref{ass:cons}(iv).

By the Riesz Representation Theorem, there exists a unique element $\mu(\theta)\in\mathcal H_K$, referred to as the \textit{conditional moment embedding} (CME) \citep{muandet2017kernel,muandet2020kernel}, such that $\mathcal L_\theta h=\langle h,\mu(\theta)\rangle_{\mathcal H_K}$ for all $h\in\mathcal H_K$. Consequently, $\|\mathcal L_\theta\|_{\mathrm{op}}=\|\mu(\theta)\|_{\mathcal H_K}$. The following lemma formalizes the identification content of this RKHS representation.

\begin{lemma}
\label{lemma:identification}
Suppose Assumption \ref{ass:cons} holds. Then, for any $\theta\in\Theta$,
$$
\mathbb{E}[\rho(Z,\theta)\mid X]=0 \;\text{a.s.}
\iff
\|\mu(\theta)\|_{\mathcal H_K}=0.
$$
\end{lemma}

Based on Lemma \ref{lemma:identification}, we define the population objective function as the squared RKHS norm $Q(\theta)\coloneqq\|\mu(\theta)\|_{\mathcal H_K}^2$. To obtain a computable expression for this criterion, note that the reproducing property implies, for every $h\in\mathcal H_K$,
\begin{equation}
\mathcal L_\theta h
=
\mathbb E[\rho(Z,\theta)^\prime h(X)]
=
\mathbb E\!\left[
\left\langle h,K(\cdot,X)\rho(Z,\theta)
\right\rangle_{\mathcal H_K}
\right]
=
\left\langle
h,\mathbb E[K(\cdot,X)\rho(Z,\theta)]
\right\rangle_{\mathcal H_K}.
\label{eq:riesz_derivation}
\end{equation}
By the uniqueness of the Riesz representer, this identifies $\mu(\theta)=\mathbb E[K(\cdot,X)\rho(Z,\theta)]$. Hence, the RKHS norm admits the closed-form double expectation
\begin{align}
Q(\theta)
&=
\left\langle
\mathbb E[K(\cdot,X)\rho(Z,\theta)],
\mathbb E[K(\cdot,\widetilde X)\rho(\widetilde Z,\theta)]
\right\rangle_{\mathcal H_K}
\nonumber \\
&=
\mathbb E\!\left[
\rho(Z,\theta)^\prime K(X,\widetilde X)\rho(\widetilde Z,\theta)
\right],
\label{eq:reproducing_trick}
\end{align}
where $(Z,X)$ and $(\widetilde Z,\widetilde X)$ are independent copies from the population distribution.

Specifically, for a continuous translation-invariant scalar kernel on
\(\mathbb R^d\), \(k(x,\widetilde x)=\varphi(x-\widetilde x)\), the
connection with characteristic-function-based ICM criteria can be made
explicit. By Bochner's theorem, there exists a finite nonnegative Borel
measure \(\Lambda\) on \(\mathbb R^d\) such that
\[
    k(x,\widetilde x)
    =
    \int_{\mathbb R^d}
    \exp\!\left\{\mathrm{i}\omega^\prime(x-\widetilde x)\right\}
    d\Lambda(\omega).
\]
Substituting this representation into \eqref{eq:reproducing_trick} and
applying Fubini's theorem and the independence of the two copies yields
\[
    Q(\theta)
    =
    \int_{\mathbb R^d}
    \left\|
    \mathbb E\!\left[
        \rho(Z,\theta)\exp(\mathrm{i}\omega^\prime X)
    \right]
    \right\|^{2}
    d\Lambda(\omega).
\]
Thus, in this special case, KMD aggregates a continuum of
characteristic-function moments, with the kernel determining how different
frequencies are weighted through the spectral measure \(\Lambda\). This
provides a direct link to Bierens-type ICM criteria and to the Fourier-moment
representation underlying the SMD construction of
\citet{lavergne2013smooth}; see also \citet{muandet2020kernel}. The
representation is interpretive rather than computational:
\(Q(\theta)\) is evaluated directly through the closed-form pairwise
expression in \eqref{eq:reproducing_trick}, without numerical integration
over \(\omega\).

Regardless of this special-case representation,
\(Q(\theta)=\|\mu(\theta)\|_{\mathcal H_K}^{2}\geq 0\), and
Lemma \ref{lemma:identification} establishes that equality holds if and
only if
\(\mathbb E[\rho(Z,\theta)\mid X]=0\) a.s. By the global identification
condition in Assumption \ref{ass:cons}(iii), this occurs if and only if
\(\theta=\theta_0\). Hence, the true parameter is uniquely identified as
the global minimizer of the population objective function. This directly
motivates our estimation strategy, which minimizes the sample analog of
\eqref{eq:reproducing_trick}.
\subsection{The KMD Estimator}
\label{subsec:estimator}

To implement the estimation, we construct the sample counterpart of the population objective $Q(\theta)$ using a $V$-statistic formulation:
\begin{equation*}
    \widehat{Q}_n(\theta) = \frac{1}{n^2}\sum_{i=1}^{n}\sum_{j=1}^{n}\rho(Z_i,\theta)^{\prime} K(X_i,X_j)\rho(Z_j,\theta).
    \label{eq:sample_obj}
\end{equation*}
The Kernel Minimum Distance (KMD) estimator is then defined as:
\begin{equation*}
    \widehat{\theta}_n = \arg\min_{\theta \in \Theta} \widehat{Q}_n(\theta).
    \label{eq:estimator_def}
\end{equation*}

\begin{remark}
We use the \(V\)-statistic criterion \(\widehat Q_n(\theta)\) rather than the off-diagonal \(U\)-statistic because it preserves the minimum-distance geometry. Specifically, \(\widehat Q_n(\theta)=\|\widehat\mu_n(\theta)\|_{\mathcal H_K}^2\), where \(\widehat\mu_n(\theta):=n^{-1}\sum_{i=1}^n K(\cdot,X_i)\rho(Z_i,\theta)\), so \(\widehat Q_n(\theta)\) is the squared RKHS norm of the empirical embedding. Since the kernel matrix is positive semidefinite, \(\widehat Q_n(\theta)\ge0\) globally, preserving the metric interpretation and improving the finite-sample optimization geometry. Off-diagonal \(U\)-statistic criteria are natural in smooth minimum-distance and local-smoothing approaches, such as \citet{lavergne2013smooth}; our use of the \(V\)-statistic reflects a finite-sample geometric consideration rather than an asymptotic distinction. In particular, because \(K-\operatorname{diag}(K)\) need not be positive semidefinite, the off-diagonal criterion need not define a squared norm and may exhibit negative regions or boundary minima in finite samples. The diagonal bias is only \(O(n^{-1})\) and vanishes asymptotically, so consistency of \(\widehat\theta_n\) is unaffected; see Section~\ref{subsec:UorV-estimation} of the Online Supplementary Material for further discussion.
\end{remark}

The following lemma formally establishes this consistency result.
\begin{lemma}
\label{lemma:theta_consist}
Under Assumptions \ref{ass:cons}, $\widehat{\theta}_n \xrightarrow{p} \theta_0$.
\end{lemma}

To establish the asymptotic normality of $\widehat{\theta}_n$, we impose the following regularity conditions necessary to validate the Taylor expansion of the first-order condition and to ensure the validity of the resulting asymptotic linear representation.
\begin{assumption}
\label{ass:an}
\quad
\begin{itemize}
    \setlength{\itemsep}{0pt} 
    \setlength{\parskip}{0pt} 
    \setlength{\parsep}{0pt}
    \item[(i)] $\theta_0\in\operatorname{int}\left(\Theta\right)$.
    
    \item[(ii)] $\rho(z, \theta)$ is twice continuously differentiable with respect to $\theta$ in an open neighborhood $\mathcal{N}\left(\theta_0\right)$ of $\theta_0$, with probability one. Let $g(z, \theta) = \partial \rho(z, \theta)/\partial \theta^{\prime}$ denote the score function with dimension $q \times p$. $\mathbb{E}\left(\sup _{\theta \in \mathcal{N}\left(\theta_0\right)}\|g(Z, \theta)\|^2\right)<\infty$,  $\mathbb{E}\left[\left\| \rho\left(Z, \theta_0\right)\right\|^4\right]<\infty$, $\mathbb{E}\left[\left\| g\left(Z, \theta_0\right)\right\|^4\right]<\infty$ and the Hessian matrix $\mathbb{E}\left(\sup _{\theta \in \mathcal{N}\left(\theta_0\right)} \|\nabla_{\theta \theta^{\prime}}^2 \rho_l\left(Z, \theta\right) \|^2\right)<\infty$, for $l=1, \cdots q$, where $\rho_l\left(\cdot\right)$ denotes the $l$-th element in $\rho\left(\cdot\right)$.
    
    \item[(iii)] The $p \times p$ matrix $\Delta\left(\theta_0\right) = \mathbb{E}\left[g(Z, \theta_0)' K(X, \tilde{X}) g(\tilde{Z}, \theta_0)\right]$ is non-singular.
\end{itemize}
\end{assumption}

Assumption \ref{ass:an} outlines the regularity conditions for the asymptotic representation. Condition (i) requires $\theta_0$ to be an interior point to validate the expansion, while condition (iii) guarantees the invertibility of the population Hessian. In condition (ii), the moment bounds involving the supremum are specifically imposed to establish the uniform convergence of the sample Hessian matrix $\widehat{H}_n(\theta)$ via the Uniform Weak Law of Large Numbers (UWLLN) for $V$-statistics. The finite fourth moments for $\rho(Z, \theta_0)$ and $g(Z, \theta_0)$ serve as sufficient conditions to ensure $\mathbb{E}[\|g(Z, \theta_0)'\rho(Z, \theta_0)\|^2] < \infty$, which enables the H\'{a}jek projection of the $V$-statistic score vector $\nabla_\theta \widehat{Q}_n(\theta_0)$. We explicitly retain these separate fourth-moment conditions to maintain a unified framework with the specification testing procedure, which relies on an analogous $V$-statistic formulation (see Section \ref{subsec:UorV-testing} of the Online Supplementary Material for a detailed discussion of $U$- versus $V$-statistics for testing). Notably, these moment requirements are slightly stronger than standard second-moment conditions; this is a direct consequence of employing the $V$-statistic estimator, which necessitates control over the stochastic behavior of the diagonal terms ($i=j$). With these regularity conditions in place, the following lemma establishes the asymptotic linear representation of $\widehat{\theta}_n$.

\begin{lemma}
\label{lemma:theta_lr}
Under Assumptions \ref{ass:cons} and \ref{ass:an}, the kernel minimum distance estimator $\widehat\theta$ admits the asymptotic linear representation
$$
\sqrt{n}\left(\widehat{\theta}_n-\theta_0\right)=-\Delta\left(\theta_0\right)^{-1} \frac{1}{\sqrt{n}} \sum_{i=1}^n G\left(X_i, \theta_0\right)^\prime \rho\left(Z_i, \theta_0\right)+o_p(1),
$$
where $G\left(X_i, \theta_0\right)=\mathbb{E}\left[ K\left(X_i, X_j\right)g\left(Z_j, \theta_0\right)|X_i\right]_{q \times p}$.
\end{lemma}

Applying the Central Limit Theorem to Lemma \ref{lemma:theta_lr} yields:
\begin{theorem}
\label{thm:theta_an}
    Under Assumptions \ref{ass:cons} and \ref{ass:an}, 
    $$
    \sqrt{n}\left(\widehat{\theta}_n-\theta_0\right) \xrightarrow{d} \mathcal{N}\left(0, \Delta \left(\theta_0\right)^{-1}\Omega\left(\theta_0\right) \Delta\left(\theta_0\right)^{-1}\right),
    $$
    where $\Omega\left(\theta_0\right)=\mathbb{E}\left[G\left(X, \theta_0\right)^\prime \rho\left(Z, \theta_0\right) \rho\left(Z, \theta_0\right)^{\prime} G\left(X, \theta_0\right)\right]_{p\times p}$.
\end{theorem}

For statistical inference, a consistent estimator of the asymptotic covariance matrix $\Sigma \coloneqq \Delta(\theta_0)^{-1}\Omega(\theta_0) \Delta(\theta_0)^{-1}$ is obtained via the plug-in sample analog $\widehat{\Sigma}_n = \widehat{\Delta}_n^{-1} \widehat{\Omega}_n \widehat{\Delta}_n^{-1}$, with components
\begin{equation}
    \label{eq:hat}
    \widehat{\Delta}_n = \frac{1}{n^2} \sum_{i=1}^n\sum_{j=1}^n \widehat{g}_i' K(X_i, X_j) \widehat{g}_j,\quad \widehat{\Omega}_n = \frac{1}{n} \sum_{i=1}^n \widehat{G}_n(X_i)^\prime\widehat{\rho}_i \widehat{\rho}_i^\prime \widehat{G}_n(X_i),
\end{equation}
where $\widehat{\rho}_i \coloneqq \rho(Z_i, \widehat{\theta}_n)$, $\widehat{g}_i \coloneqq g(Z_i, \widehat{\theta}_n)$, and $\widehat{G}_n(X_i) \coloneqq n^{-1} \sum_{j=1}^n K(X_i, X_j)\widehat{g}_j$. 

\begin{proposition} 
\label{prop:variance_consistency}
Under Assumptions \ref{ass:cons} and \ref{ass:an}, $\widehat{\Sigma}_n \xrightarrow{p} \Sigma$ as $n \to \infty$.
\end{proposition}
Proposition \ref{prop:variance_consistency} provides a basis for standard asymptotic inference on smooth functionals of $\theta_0$. Pointwise confidence intervals for individual components follow immediately from the asymptotic normality of $\widehat{\theta}_n$ and the consistency of $\widehat{\Sigma}_n$. For a full-row-rank matrix $R\in\mathbb{R}^{q\times p}$ and vector $r\in\mathbb{R}^q$, the linear restriction $\mathbb{H}_0:R\theta_0=r$ may be tested using the feasible Wald statistic $W_n^{L}=(R\widehat{\theta}_n-r)^\prime (R\widehat{\Sigma}_nR^\prime/n)^{-1}(R\widehat{\theta}_n-r)$. Under $\mathbb{H}_0$, we have $W_n^{L}\xrightarrow{d}\chi_q^2$. More generally, let $h:\mathbb{R}^p\to\mathbb{R}^q$ be continuously differentiable, with Jacobian $\nabla h(\theta)$ of full row rank at $\theta_0$. For the nonlinear restriction $\mathbb{H}_0:h(\theta_0)=0$, the feasible Wald statistic is $W_n^{NL}=h(\widehat{\theta}_n)^\prime (\nabla h(\widehat{\theta}_n)\widehat{\Sigma}_n\nabla h(\widehat{\theta}_n)^\prime/n)^{-1} h(\widehat{\theta}_n)$, and $W_n^{NL}\xrightarrow{d}\chi_q^2$ under $\mathbb{H}_0$. These standard Wald statistics concern restrictions on \(\theta_0\), which should be distinguished from the omnibus specification test of the conditional moment restrictions developed in Section~\ref{sec:test}; see, e.g., \citet{lavergne2013smooth} for SMD-based tests of parameter restrictions.

Finally, although the baseline KMD estimator is generally not semiparametrically efficient, it admits a two-step efficient refinement in the spirit of \citet{dominguez2004consistent}. In particular, starting from the $\sqrt{n}$-consistent estimator $\widehat{\theta}_n$, one may take a single Newton--Raphson step toward the feasible efficient GMM objective constructed from nonparametric estimators of the relevant conditional moments; see \citet{newey1990efficient} and \citet{robinson1991best}. However, there is no free lunch in achieving semiparametric efficiency: this refinement requires additional nonparametric estimation, together with the attendant smoothing choices, and thus sacrifices some of the simplicity and robustness of the KMD procedure. We therefore defer the formal discussion to Section \ref{app:efficiency} of the Online Supplementary Material.
\section{Specification Testing}
\label{sec:test}
\subsection{Test Statistic}
Following the identification and estimation of the structural parameters discussed in the preceding section, we now turn to the specification test of the conditional moment restrictions. Specifically, we evaluate the validity of the model by testing the null hypothesis of correct specification against the omnibus alternative:
\begin{equation}
    \label{H_1}
    \mathbb{H}_0: \eqref{H_0} \text{ holds} \quad \text{v.s.}\quad\mathbb{H}_1: \mathbb{P}\left(\E\left[\rho(Z, \theta) \mid X\right]=0\right)<1 \text{ for all }\theta \in \Theta.
\end{equation}
To construct the test statistic, we utilize the minimized value of the KMD objective function. The test statistic $\widehat{T}_n$ is defined as the scaled $V$-statistic evaluated at $\widehat{\theta}_n$:
$$
\widehat{T}_n=n\widehat{Q}_n(\widehat{\theta}_n)=\frac{1}{n}\sum_{i=1}^n\sum_{j=1}^n\rho(Z_i,\widehat\theta_n)^{\prime} K(X_i,X_j)\rho(Z_j,\widehat\theta_n).
$$
In what follows, we derive the asymptotic distribution of $\widehat{T}_n$ under the null hypothesis $\mathbb{H}_0$, establish its consistency under the fixed alternative $\mathbb{H}_1$, and analyze its local power properties under the sequence of local alternatives $\mathbb{H}_{1n}$.

\subsection{Behavior Under the Null $\mathbb{H}_0$}
We first establish an asymptotic representation of $\widehat{T}_n$ under the null hypothesis, showing that the estimation effect is asymptotically equivalent to replacing the reproducing kernel with its projection.
\begin{lemma}
\label{lemma:repre_H0}
Suppose Assumptions \ref{ass:cons} and \ref{ass:an} hold. Under $\mathbb{H}_0$, the test statistic $\widehat{T}_n$ has the following asymptotic representation:
$$
\widehat{T}_n = \frac{1}{n}\sum_{i=1}^n\sum_{j=1}^n\rho\left(Z_i,\theta_0\right)^{\prime} K^p(X_i,X_j)\rho\left(Z_j,\theta_0\right)+o_p(1),
$$
where $K^p(X_i,X_j)=K(X_i,X_j)-G(X_i,\theta_0)\Delta(\theta_0)^{-1} G(X_j,\theta_0)^{\prime}$.   
\end{lemma}

\begin{remark}
\label{rem:kp_geometry}
The projected kernel $K^p$ induces a projection geometry within the vector-valued RKHS $\mathcal{H}_K$ onto the orthogonal complement of the score subspace $\mathcal{S} = \operatorname{span}\{\psi_l\}_{l=1}^p$, where $\psi_l(\cdot) = \mathbb{E}[K(\cdot, X) g_l(Z, \theta_0)]$. As detailed in the Online Supplementary Material, for any function $h(\cdot) = \sum_{i=1}^n K(\cdot, x_i)c_i$ with arbitrary coefficients $\{c_i\}_{i=1}^n \subset \mathbb{R}^q$, the quadratic form of $K^p$ satisfies the isometry:
\begin{equation}
    \sum_{i=1}^n \sum_{j=1}^n c_i^{\prime} K^p(x_i, x_j) c_j = \left\| (I - P_{\mathcal{S}}) h \right\|_K^2 \ge 0,
\end{equation}
where $P_{\mathcal{S}}$ denotes the orthogonal projection operator onto $\mathcal{S}$. This identity establishes that $K^p$ is positive semidefinite.

In particular, the asymptotic representation of $\widehat{T}_n$ in Lemma \ref{lemma:repre_H0} corresponds to the specific case where the coefficients are the normalized residuals, $c_i = n^{-1/2}\rho(Z_i, \theta_0)$. By capturing the component of the moment conditions orthogonal to the score functions $g_l$, this asymptotic representation effectively eliminates the first-order estimation effect.
\end{remark}

Building on this representation, the following theorem derives the limiting distribution of $\widehat{T}_n$ under the null, which exhibits the standard spectral form for degenerate $V$-statistics.
\begin{theorem}
\label{thm:H0}
Suppose Assumptions \ref{ass:cons} and \ref{ass:an} hold. Then, under $\mathbb{H}_0$:
    $$
 \widehat{T}_n \xrightarrow{d} \sum_{k=1}^{\infty} \lambda_kW_k^2,
$$
where $\left\{W_k\right\}_{k=1}^\infty$ are i.i.d. $\mathcal{N}(0,1)$ and $\left\{\lambda_k\right\}_{k=1}^{\infty}$ are eigenvalues of the operator $\mathcal{A}$ defined on the function space $L_2\left(\mathcal{Z}\times\mathcal{X},\mathbb{P}_{XZ}\right)$ as
\begin{equation}
    \label{eq_A}
    \mathcal{A} \psi(s)=\int_{\mathcal{Z}\times\mathcal{X}} f_p\left(s, \tilde s, \theta_0\right) \psi\left(\tilde s\right) d \mathbb{P}_{XZ}\left(\tilde s\right),\quad s\in\mathcal{Z}\times\mathcal{X},\psi\in L_2\left(\mathcal{Z}\times\mathcal{X},\mathbb{P}_{XZ}\right),
\end{equation}
where $f_p\left(s, s^{\prime}, \theta_0\right)=\rho\left(z, \theta_0\right)^{\prime}K^p\left(x, x^{\prime}\right)\rho\left(z^{\prime}, \theta_0\right)$.
\end{theorem}
\begin{remark}
\label{rem:generalized_sargan}
The spectral representation $\sum_{k=1}^{\infty} \lambda_k W_k^2$ derived in Theorem \ref{thm:H0} characterizes the test statistic $\widehat{T}_n$ as a \textit{Generalized Sargan Test} for a continuum of moment conditions. In the classical GMM framework with a fixed number of moments $m$ and parameters $p$, Hansen's $J$-statistic converges to a $\chi^2_{m-p}$ distribution, assessing the validity of the $m-p$ overidentifying restrictions. Our framework corresponds to the limiting case $m \to \infty$. Specifically, $\widehat{T}_n$ asymptotically isolates the component of the embedded moment conditions lying in the orthogonal complement of the score subspace, thereby assessing the structural validity of the continuum of overidentifying restrictions beyond the $p$ degrees of freedom consumed by parameter estimation. This interpretation parallels the unified framework of \cite{dominguez2015simple}. While they address the estimation effect by projecting the empirical process, our framework relies on the projection kernel $K^p$.
\end{remark}

The test rejects $\mathbb{H}_0$ at significance level $\alpha$ if $\widehat{T}_n > c_{1-\alpha}$, where $c_{1-\alpha}$ is the $(1-\alpha)$-quantile of the limiting distribution defined in Theorem \ref{thm:H0}. However, since the eigenvalues $\{\lambda_k\}_{k=1}^\infty$ depend on the kernel and the unknown data-generating process, this distribution is non-pivotal. In Section \ref{sec:mb}, we introduce a computationally simple multiplier bootstrap to obtain feasible critical values.

\subsection{Behavior Under the Fixed Alternative $\mathbb{H}_1$}
We now establish the test's consistency against fixed alternatives by showing that $\widehat{T}_n$ diverges under global misspecification. This requires the following modified regularity conditions, under which a pseudo-true parameter $\theta_1$ is defined as the unique minimizer of the population objective function.
\begin{assumption}
    \label{ass:theta1_con}
Suppose Assumption \ref{ass:cons} (i), (ii), and (iv) hold, and (iii) is replaced by: The parameter space $\Theta\subset \mathbb{R}^p$ is compact. $Q\left(\theta\right)$ is uniquely minimized at a pseudo true value $\theta_1\in\Theta$ such that $Q(\theta_1) > 0$.
\end{assumption}

\begin{assumption}
    \label{ass:theta1_an}
    Suppose Assumption \ref{ass:an} holds with $\theta_0$ replaced by $\theta_1$. Furthermore, the generalized Hessian matrix under misspecification,
    $$
    \widetilde{\Delta}(\theta_1) = \mathbb{E}\left\{k(X_i,X_j)\left[g(Z_i,\theta_1)^{\prime}g(Z_j,\theta_1) + \sum_{l=1}^q\rho_l(Z_j,\theta_1)\nabla_{\theta\theta^{\prime}}^2\rho_l(Z_i,\theta_1)\right]\right\},
    $$
    is non-singular.
\end{assumption}

Analogous to Lemma \ref{lemma:theta_lr}, the following result establishes the asymptotic representation of $\widehat{\theta}_n$ around the pseudo-true value $\theta_1$.
\begin{lemma}
\label{lemma:theta1_asy}
    Under $\mathbb{H}_1$, suppose Assumptions \ref{ass:theta1_con} and \ref{ass:theta1_an} hold, then 
    $$
    \sqrt{n}\left(\widehat{\theta}_n-\theta_1\right)=-\widetilde\Delta\left(\theta_1\right)^{-1} \frac{1}{\sqrt{n}} \sum_{i=1}^n \left[G\left(X_i, \theta_1\right)^\prime \rho\left(Z_i, \theta_1\right)+g(Z_i,\theta_1)'M(X_i,\theta_1)\right]+o_p(1),
    $$
where $G\left(X_i, \theta\right)$ is defined in Lemma \ref{lemma:theta_lr}, and $M\left(X_i, \theta\right)=\mathbb{E}\left[K\left(X_i, X_j\right) \rho\left(Z_j, \theta_0\right)|X_i\right]_{q \times 1}$.
\end{lemma}
In contrast to the degenerate behavior under $\mathbb{H}_0$, the objective function $\widehat{Q}_n(\widehat{\theta}_n)$ under $\mathbb{H}_1$ behaves asymptotically as a non-degenerate $V$-statistic, yielding $\sqrt{n}-$asymptotic normality centered at the positive constant $Q(\theta_1)$.
\begin{theorem}
    \label{thm:H1}
    Suppose Assumptions \ref{ass:theta1_con} and \ref{ass:theta1_an} hold. Under the fixed alternative $\mathbb{H}_1$:
    \begin{equation*}
        \sqrt{n}\left(\widehat{Q}_n(\widehat{\theta}_n) - Q(\theta_1)\right) \xrightarrow{d} \mathcal{N}\left(0, \sigma^2(\theta_1)\right),
    \end{equation*}
    where $\sigma^2(\theta_1) = 4\operatorname{Var}\left[\rho(Z,\theta_1)^{\prime}M(X,\theta_1)\right]$.
\end{theorem}

Since $Q(\theta_1)>0$, Theorem \ref{thm:H1} implies that $\widehat{T}_n=n\widehat{Q}_n(\widehat{\theta}_n)\to\infty$. Hence, the test is consistent against any fixed alternative.
\subsection{Behavior Under the Local Alternative $\mathbb{H}_{1n}$}
We further examine the asymptotic power under a sequence of Pitman local alternatives converging at the parametric rate $n^{-1/2}$. Let $\mathbb{P}$ denote the null measure satisfying $\mathbb{E}_{\mathbb{P}}[\rho(Z, \theta_0) \mid X] = 0$ a.s., and $\mathbb{P}^{(n)}$ denote the sequence of alternative measures. We specify the form of these local perturbations as follows.

\begin{assumption}
\label{ass:local_dgp}
Under $\mathbb{H}_{1n}$, the observations $(X,Z)$ follow a sequence of probability measures $\mP^{(n)}$ such that $d\mP^{(n)}/d\mP(x,z)=1+n^{-1/2}h(x,z)$, where $\E_{\mP}[h(X,Z)^4]<\infty$ and $\E_{\mP}[h(X,Z)\mid X]=0$ a.s..
\end{assumption}

Applying the change of measure formula under Assumption \ref{ass:local_dgp}, we obtain the specific Pitman alternative:
\begin{equation}
    \label{H_1n}
    \mathbb{H}_{1n}: \quad \mathbb{E}_{\mathbb{P}^{(n)}}\left[\rho(Z, \theta_0) \mid X\right] = \frac{\delta(X)}{\sqrt{n}}.
\end{equation}
Here, the drift function is defined as $\delta(X) \coloneqq \mathbb{E}_{\mathbb{P}}[\rho(Z, \theta_0)h(X, Z)|X]$, which is square-integrable due to the fourth finite moment conditions on $\rho$ and $h$. This local drift induces an asymptotic bias shift in the estimator, as formalized in the following linear representation.

\begin{comment}
\begin{equation*}
    \E_{\mP^{(n)}}[\rho(Z, \theta_0) \mid X] 
    = \frac{\E_{\mP}\left[\rho(Z, \theta_0) \frac{d\mP^{(n)}}{d\mP} \bigg| X\right]}{\E_{\mP}\left[\frac{d\mP^{(n)}}{d\mP} \bigg| X\right]}
    = \frac{\E_{\mP}\left[\rho(Z, \theta_0)\left(1 + \frac{1}{\sqrt{n}}h(X, Z)\right) \bigg| X\right]}{\E_{\mP}\left[1 + \frac{1}{\sqrt{n}}h(X, Z) \bigg| X\right]}.
\end{equation*}
Define $\delta(X)=\E_{\mP}\left[\rho(Z, \theta_0)h(X, Z) \mid X\right]$, we obtain the specific Pitman alternative:
\begin{equation}
    \label{H_1n}
    \mathbb{H}_{1n}:\quad\E_{\mP^{(n)}}[\rho(Z, \theta_0) \mid X] = \frac{\delta(X)}{\sqrt{n}}.
\end{equation}
Note that the drift function $\delta(X)$ is square-integrable with respect to the marginal probability measure of $X$ derived from $\mP$. This follows from the fourth-moment conditions on $\rho$ and $s$ via Hölder's inequality.
\end{comment}
\begin{lemma}
\label{lemma:H1n_lp}
    Under the sequence of local alternatives $\mathbb{H}_{1n}$ satisfying Assumption \ref{ass:local_dgp}, and assuming Assumptions \ref{ass:cons}(ii)--(iv) and \ref{ass:an} hold under the null measure $\mP$, the estimator satisfies the asymptotic linear representation:
    $$
    \sqrt{n}\left(\widehat{\theta}_n-\theta_0\right) = -\Delta\left(\theta_0\right)^{-1} \frac{1}{\sqrt{n}} \sum_{i=1}^n G\left(X_i, \theta_0\right)^\prime \epsilon_i - \Delta\left(\theta_0\right)^{-1} B_\delta(\theta_0) + o_p(1),
    $$
    where $\epsilon_i \coloneqq \rho(Z_i, \theta_0) - n^{-1/2}\delta(X_i)$ represents the centered error term under $\mathbb{H}_{1n}$, and $B_\delta(\theta_0) = \E_{\mP}\left[G(X_i, \theta_0)^\prime\delta(X_i)\right]$ denotes the asymptotic bias shift.
\end{lemma}

Building on this linear representation, the following theorem establishes the asymptotic distribution of $\widehat{T}_n$ under local alternatives.
\begin{theorem}
\label{thm:H_1n}
Suppose Assumptions for Lemma \ref{lemma:H1n_lp} hold. Then, under $\mathbb{H}_{1n}$, the test statistic converges in distribution as:
$$
\widehat{T}_n \xrightarrow{d} \sum_{k=1}^{\infty}\lambda_k\left(W_k+c_k\right)^2,
$$
where $\{W_k\}_{k=1}^{\infty}$ are i.i.d. $\mathcal{N}(0,1)$, and $\{\lambda_k\}_{k=1}^{\infty}$ and $\{\varphi_k(\cdot)\}_{k=1}^{\infty}$ are the eigenvalues and orthonormal eigenfunctions, respectively, of the integral operator $\mathcal{A}$ defined in Theorem \ref{thm:H0}. The non-centrality parameters are given by $c_k = \mathbb{E}_{\mathbb{P}}[\varphi_k(Z, X) h(Z, X)]$.
\end{theorem}

\begin{remark}
\label{rem:local_power}
    Theorem \ref{thm:H_1n} characterizes the limiting distribution as an infinite sum of weighted non-central $\chi^2_1$ variables. This result situates the proposed statistic within the established framework of ICM tests, exhibiting a spectral structure analogous to that derived by \cite{bierens1997asymptotic} for local power analysis.

    Its asymptotic local power is completely determined by the non-centrality parameters $\{c_k\}_{k=1}^\infty$. In particular, local power is trivial, in the sense that it converges to size $\alpha$, if and only if $c_k=0$ for all $k$ with positive eigenvalues. As shown in the proof of Theorem \ref{thm:H_1n}, this is equivalent to
$$
\sum_{k=1}^\infty \lambda_k c_k^2
= \E_{\mP}\!\left[\delta(X)^\prime K^p(X,\tilde X;\theta_0)\delta(\tilde X)\right]
=0.
$$
Since the projected kernel operator is positive semi-definite, this holds if and only if $\delta(X)$ lies in its null space $\mathbb{P}$-a.s., that is, $\E_{\mP}[K^p(X,\tilde X;\theta_0)\delta(\tilde X)\mid X]=0$. By Remark \ref{rem:kp_geometry}, this means that the local drift $\delta(X)$ belongs to the score subspace $\mathcal S$. Provided $K$ is ISPD, trivial power obtains if and only if $\delta(X)=\E_{\mathbb P}[g(Z,\theta_0)\mid X]\mathbf v$ $(\mathbb P\text{-a.s.})$ for some $\mathbf v\in\mathbb R^p$.

Intuitively, trivial power arises precisely when the local drift mimics the score function. In that case, it is asymptotically indistinguishable from a local shift in $\theta_0$ along the direction $\mathbf v$, and is therefore absorbed by the estimator $\widehat{\theta}_n$. This is a consequence of the preliminary estimation step rather than a deficiency of the test itself.
\end{remark}

\section{Multiplier Bootstrap}
\label{sec:mb}
The asymptotic null distribution of the test statistic $\widehat{T}_n$ is generally non-pivotal, as it depends on the underlying data-generating process via the operator $\mathcal{A}$. To circumvent this issue, we introduce an easy-to-implement multiplier bootstrap procedure to approximate the critical values. Motivated by the asymptotic representation established in Lemma \ref{lemma:repre_H0}, the bootstrap procedure is implemented as follows:
\begin{enumerate}
  \item Generate a sequence of i.i.d.  random variables $\left\{V_i\right\}_{i=1}^n$ with zero mean, unit variance and finite fourth moment, independent of the original samples $\left\{X_i,Z_i\right\}_{i=1}^n$; e.g., Rademacher random variable, standard normal random variable, or Bernoulli random variable with $ P(v=1-\kappa)=\kappa/\sqrt{5} $ and $ P(v=\kappa) = 1-\kappa/\sqrt{5} $, where $ \kappa = (\sqrt{5}+1)/2 $ \citep{mammen1993bootstrap}. 
  \item Compute the multiplier bootstrap statistic 
    \begin{equation*}
    \widehat{T}^{*}_n = \frac{1}{n}\sum_{i=1}^n\sum_{j=1}^nV_iV_j\rho(Z_i,\widehat{\theta}_n)^{\prime}\widehat{K}_n^p(X_i,X_j,\widehat{\theta}_n)\rho(Z_j,\widehat{\theta}_n),
    \end{equation*}
    where $\widehat{K}^{p}_n(X_i,X_j,\widehat{\theta}_n) \coloneq K(X_i,X_j) - \widehat{G}_n(X_i)\widehat{\Delta}_n^{-1}\widehat{G}_n(X_j)^{\prime}$, with $\widehat{G}_n$ and $\widehat{\Delta}_n$ defined as in \eqref{eq:hat}.    
  \item Repeat Steps 1 and 2 $B$ times to obtain a sequence of bootstrap statistics $\left\{\widehat{T}_{n,b}^*\right\}_{b=1}^B$.
  \item For a given significance level $\alpha \in (0,1)$, compute the $(1-\alpha)$-th quantile of $\left\{\widehat{T}_{n,b}^*\right\}_{b=1}^B$, denoted by $c_{n,\alpha}^*$. The null hypothesis is rejected if $\widehat{T}_n > c_{n,\alpha}^*$.    
\end{enumerate}

To establish the asymptotic validity of this procedure, the following lemma first provides an asymptotic representation of the bootstrap statistic $\widehat{T}^{*}_n$, replacing the estimated quantities with their population counterparts.
\begin{lemma}
\label{lemma:mb_asy}
Let $\theta^*$ be the (pseudo) true value under different hypotheses, i.e., $\theta^*=\theta_0$ under $\mathbb{H}_0$ or $\mathbb{H}_{1n}$ and $\theta^*=\theta_1$ under $\mathbb{H}_1$. Under Assumptions \ref{ass:cons} and \ref{ass:an} ( or Assumptions \ref{ass:theta1_con} and \ref{ass:theta1_an} for $\theta^*=\theta_1$), the multiplier bootstrap statistic has the following asymptotic representation:
    \begin{equation*}
\widehat{T}^{*}_n =\frac{1}{n}\sum_{i=1}^n\sum_{j=1}^nV_iV_j\rho(Z_i,{\theta}^*)^{\prime}{K}^p(X_i,X_j,{\theta}^*)\rho(Z_j,{\theta}^*)+o_p(1).
\end{equation*}
\end{lemma}

Denote by `` $\xrightarrow{d,*}$ in probability'' the weak convergence in probability under the bootstrap law, that is, conditional on the original sample $\left\{X_i,Z_i\right\}_{i=1}^n$ and let $\mP^*$ be the conditional probability related to $\left\{V_i\right\}_{i=1}^n$ given $\left\{X_i,Z_i\right\}_{i=1}^n$. The next theorem establishes the asymptotic validity of the proposed multiplier bootstrap procedure. 

\begin{theorem}
    \label{thm:mb}
    Suppose Assumptions for Lemma \ref{lemma:mb_asy} hold. Then,
    \begin{itemize}
    \setlength{\itemsep}{0pt} 
    \setlength{\parskip}{0pt} 
    \setlength{\parsep}{0pt}
        \item Under $\mathbb{H}_0$ in \eqref{H_0} or $\mathbb{H}_{1n}$ in \eqref{H_1n}, we have
        $$
        \widehat{T}_n^*\xrightarrow{d,*}\sum_{k=1}^{\infty}\lambda_kW_k^2 \text{ in probability},
        $$
        where $\sum_{k=1}^{\infty}\lambda_kW_k^2$ is defined in Theorem \ref{thm:H0}.
        \item Under $\mathbb{H}_1$ in \eqref{H_1}, for every $\varepsilon>0$, there exists a positive $A$ such that
        $$
        \mathbb{P}^*\left(\widehat{T}_n^*\geq A\right)\leq \varepsilon +o_{p}(1).
        $$
    \end{itemize}
\end{theorem}

On one hand, under $\mathbb{H}_0$, both $\widehat{T}_n^*$ and $\widehat{T}_n$ converge in distribution to $\sum_{k=1}^{\infty}\lambda_kW_k^2$, as defined in Theorem \ref{thm:H0}. Hence, a test based on the bootstrapped critical value would yield an asymptotically correct level for $\widehat{T}_n$. On the other hand, under $\mathbb{H}_1$, by Theorem \ref{thm:H1}, $\widehat{T}_n(\widehat{\theta}_n)$ diverges to infinity with probability approaching one, whereas $\widehat{T}_n^*$ is stochastically bounded conditional on the original sample. This implies that the test based on the bootstrapped critical value is consistent against all such fixed alternatives. Moreover, Theorem \ref{thm:mb} combined with Theorem \ref{thm:H_1n} also implies that our tests implemented through the multiplier bootstrap preserve the asymptotic local power properties of $\widehat{T}_n$ under $\mathbb{H}_{1n}$. Therefore, the multiplier bootstrap procedure is asymptotically valid.

\section{Simulation Study}  
\label{sec:simu} 

This section evaluates the finite-sample properties of the proposed KMD framework across three data-generating processes of increasing complexity: the Linear Regression Model, the Box--Cox Transformation Model, and Instrumental Variable Models. We assess estimation accuracy using bias and RMSE, with additional coverage rate results available in Section \ref{subsec:coverage} of the Online Supplementary Material. For the specification test, we examine empirical rejection rates to validate size and power; we also provide a benchmark comparison with \cite{dominguez2015simple} in Section \ref{subsec:dl-test} of the Online Supplementary Material.

Across all designs, we consider sample sizes $n \in \{100,200\}$ for estimation and extend the analysis to $n=400$ for specification testing. We vary the covariate dimension $p \in \{3,5,10\}$ to investigate performance in higher-dimensional settings. Except for the mixed IV design with discrete instruments, where the kernel is specified separately, KMD uses a Gaussian kernel with the median-heuristic bandwidth. All results are based on 1,000 Monte Carlo replications, with test critical values obtained from 999 bootstrap repetitions.

\subsection{The Linear Regression Model}
\label{subsec:lr_models}

We begin with a linear regression benchmark under homoskedastic and heteroskedastic disturbances. We compare the KMD estimator with OLS, feasible weighted least squares, and infeasible optimal GMM, and study the corresponding KMD specification test under several basic alternatives. This benchmark mainly serves as a sanity check: the KMD estimator is essentially unbiased, with moderate efficiency loss relative to the parametric benchmarks, and the KMD test shows good size control and strong finite-sample power. To conserve space, the design and full results are reported in Section \ref{app:lr_simulation} of the Online Supplementary Material. We now turn to nonlinear designs that more directly highlight the comparative advantages of the proposed method.

\subsection{The Box--Cox Transformation Model}
\label{subsec:bc_models}

We next extend the simulation to the generalized Box--Cox transformation model. Our design draws upon the setups in \cite{shin2008semiparametric} and \cite{dominguez2015simple}, governed by the structural equation:
\begin{equation*}
\label{bc_structural}
    y_i^{(\lambda_0)} = \alpha_0 + X_i^{\prime} \beta_0 + \epsilon_i, \quad i = 1, \dots, n.
\end{equation*}
To accommodate outcomes with non-positive support without error truncation, we employ the generalized transformation family of \cite{bickel1981analysis}:
\begin{equation*}
\label{bc_bickel}
    y_i^{(\lambda)} =
    \begin{cases} 
    \dfrac{|y_i|^\lambda \operatorname{sgn}(y_i) - 1}{\lambda} & \text{if } \lambda \neq 0, \\
    \ln(y_i) & \text{if } \lambda = 0.
    \end{cases}
\end{equation*}
We fix the intercept at $\alpha_0 = 0.5$ and evaluate performance across $\lambda_0 \in \{0, 0.5, 1\}$. The slope parameter is set to $\beta_0=(1,\dots,1,-1,\dots,-1)^{\prime}$, with the first $\lceil p/2\rceil$ components equal to $1$, and the regressors $X_i$ are generated from a multivariate normal distribution $\mathcal{N}(0,\Sigma_X)$ with Toeplitz covariance structure $\Sigma_{X,jk}=0.5^{|j-k|}$. We consider both homoskedastic and heteroskedastic disturbances, with $\epsilon_i \sim \mathcal{N}(0,1)$ in the former case and $\epsilon_i=\exp(0.25X_{i,1})\nu_i$, $\nu_i \sim \mathcal{N}(0,1)$, in the latter.

To assess estimation performance, we compare the KMD estimator with two feasible consistent alternatives: the nonlinear two-stage least squares (NL2S) estimator of \citet{amemiya1981comparison} and the minimum distance estimator of \citet{shin2008semiparametric}. For the log-linear case ($\lambda_0=0$), we also report the infeasible Oracle GMM (GMM$^*$), based on the optimal instrument $A^*(X_i)=\sigma^{-2}(X_i)\bigl(-1,\,-X_i^\prime,\,[(\alpha_0+X_i^\prime\beta_0)^2+\sigma^2(X_i)]/2\bigr)^\prime$.

\begin{table}[htbp]
  \centering
  \caption{Estimation Performance of Box--Cox Model Parameters ($\lambda=0$)}
  \label{tab:boxcox_estimation_1}
  
  \footnotesize 
  \renewcommand{\arraystretch}{1.15}
  
  % --- 表格定义 ---
  % 1. 移除了 \setlength{\tabcolsep}{0pt}，恢复默认间距防止过挤
  % 2. 在列定义中，开头和结尾加了 @{\hspace{1em}}，
  %    强制在最左和最右留出 1em 的空白，解决“顶着右边”的问题
  \begin{tabular*}{\textwidth}{
    @{\hspace{1em}\extracolsep{\fill}} % 左侧留空 + 开启自动填充
    c c c 
    *{4}{S[table-format=-1.3]} % Bias 四列
    *{4}{S[table-format=1.3]}  % RMSE 四列
    @{\hspace{1em}}            % 右侧留空
  }
    \toprule
    & & & \multicolumn{4}{c}{Bias} & \multicolumn{4}{c}{RMSE} \\
    \cmidrule(lr){4-7} \cmidrule(lr){8-11}
    
    {$n$} & {$p$} & {Param.} & {NL2S} & {Shin-MD} & {GMM$^*$} & {KMD} & {NL2S} & {Shin-MD} & {GMM$^*$} & {KMD} \\
    \midrule

    \multicolumn{11}{l}{\textit{Panel A: Homoskedastic Errors}} \\
    \addlinespace[0.2em]
    \multirow{9}{*}{$100$} & \multirow{3}{*}{$3$} 
         & $\alpha$  &  0.003 & -0.003 &  0.009 & -0.013 & 0.206 & 0.163 & 0.162 & 0.150 \\
       & & $\beta$   &  0.018 &  0.007 &  0.005 &  0.008 & 0.145 & 0.160 & 0.133 & 0.134 \\
       & & $\lambda$ & -0.004 & -0.007 &  0.002 & -0.013 & 0.100 & 0.072 & 0.071 & 0.066 \\
    \addlinespace[0.2em]
                           & \multirow{3}{*}{$5$} 
         & $\alpha$  & -0.007 & -0.006 &  0.005 & -0.007 & 0.189 & 0.187 & 0.141 & 0.143 \\
       & & $\beta$   &  0.009 &  0.019 &  0.011 &  0.013 & 0.134 & 0.188 & 0.130 & 0.134 \\
       & & $\lambda$ & -0.005 & -0.003 & -0.000 & -0.005 & 0.047 & 0.037 & 0.028 & 0.028 \\
    \addlinespace[0.2em]
                           & \multirow{3}{*}{$10$} 
         & $\alpha$  & -0.001 & -0.019 &  0.001 &  0.000 & 0.168 & 0.240 & 0.140 & 0.138 \\
       & & $\beta$   &  0.015 &  0.015 &  0.017 &  0.015 & 0.139 & 0.206 & 0.138 & 0.139 \\
       & & $\lambda$ & -0.000 & -0.003 &  0.000 & -0.000 & 0.014 & 0.012 & 0.009 & 0.009 \\
    \addlinespace[0.5em]
    \multirow{9}{*}{$200$} & \multirow{3}{*}{$3$} 
         & $\alpha$  & -0.004 & -0.002 &  0.002 & -0.010 & 0.141 & 0.138 & 0.107 & 0.115 \\
       & & $\beta$   &  0.008 &  0.007 &  0.008 &  0.012 & 0.095 & 0.113 & 0.090 & 0.094 \\
       & & $\lambda$ & -0.004 & -0.003 &  0.000 & -0.007 & 0.066 & 0.061 & 0.043 & 0.048 \\
    \addlinespace[0.2em]
                           & \multirow{3}{*}{$5$} 
         & $\alpha$  & -0.004 & -0.000 &  0.001 & -0.006 & 0.125 & 0.127 & 0.097 & 0.097 \\
       & & $\beta$   &  0.005 &  0.010 &  0.004 &  0.004 & 0.092 & 0.131 & 0.090 & 0.093 \\
       & & $\lambda$ & -0.002 & -0.000 & -0.000 & -0.002 & 0.030 & 0.027 & 0.019 & 0.019 \\
    \addlinespace[0.2em]
                           & \multirow{3}{*}{$10$} 
         & $\alpha$  &  0.001 & -0.016 & -0.004 & -0.000 & 0.114 & 0.238 & 0.088 & 0.090 \\
       & & $\beta$   &  0.012 &  0.013 &  0.010 &  0.012 & 0.093 & 0.174 & 0.092 & 0.095 \\
       & & $\lambda$ &  0.000 & -0.002 & -0.000 & -0.000 & 0.010 & 0.010 & 0.006 & 0.006 \\
    
    \midrule
    
    \multicolumn{11}{l}{\textit{Panel B: Heteroskedastic Errors}} \\
    \addlinespace[0.2em]
    \multirow{9}{*}{$100$} & \multirow{3}{*}{$3$} 
         & $\alpha$  & -0.020 & -0.061 &  0.004 & -0.051 & 0.210 & 0.166 & 0.152 & 0.160 \\
       & & $\beta$   &  0.012 &  0.018 &  0.006 &  0.036 & 0.148 & 0.149 & 0.122 & 0.138 \\
       & & $\lambda$ & -0.014 & -0.039 &  0.002 & -0.034 & 0.103 & 0.078 & 0.065 & 0.074 \\
    \addlinespace[0.2em]
                           & \multirow{3}{*}{$5$} 
         & $\alpha$  & -0.028 & -0.049 & -0.000 & -0.033 & 0.191 & 0.183 & 0.140 & 0.149 \\
       & & $\beta$   &  0.013 &  0.017 &  0.014 &  0.016 & 0.141 & 0.175 & 0.126 & 0.141 \\
       & & $\lambda$ & -0.008 & -0.016 &  0.001 & -0.011 & 0.048 & 0.038 & 0.027 & 0.031 \\
    \addlinespace[0.2em]
                           & \multirow{3}{*}{$10$} 
         & $\alpha$  & -0.024 & -0.060 &  0.009 & -0.013 & 0.181 & 0.250 & 0.135 & 0.143 \\
       & & $\beta$   &  0.017 &  0.028 &  0.016 &  0.018 & 0.147 & 0.197 & 0.131 & 0.148 \\
       & & $\lambda$ & -0.003 & -0.006 &  0.001 & -0.002 & 0.015 & 0.013 & 0.009 & 0.010 \\
    \addlinespace[0.5em]
    \multirow{9}{*}{$200$} & \multirow{3}{*}{$3$} 
         & $\alpha$  & -0.003 & -0.037 & -0.003 & -0.032 & 0.158 & 0.135 & 0.104 & 0.118 \\
       & & $\beta$   &  0.011 &  0.010 &  0.006 &  0.018 & 0.100 & 0.104 & 0.082 & 0.096 \\
       & & $\lambda$ & -0.003 & -0.023 & -0.001 & -0.020 & 0.075 & 0.063 & 0.042 & 0.053 \\
    \addlinespace[0.2em]
                           & \multirow{3}{*}{$5$} 
         & $\alpha$  & -0.014 & -0.030 & -0.000 & -0.022 & 0.131 & 0.124 & 0.092 & 0.100 \\
       & & $\beta$   &  0.008 &  0.016 &  0.009 &  0.015 & 0.098 & 0.119 & 0.084 & 0.098 \\
       & & $\lambda$ & -0.005 & -0.010 & -0.001 & -0.008 & 0.034 & 0.027 & 0.017 & 0.022 \\
    \addlinespace[0.2em]
                           & \multirow{3}{*}{$10$} 
         & $\alpha$  & -0.012 & -0.051 &  0.004 & -0.010 & 0.122 & 0.223 & 0.091 & 0.096 \\
       & & $\beta$   &  0.008 &  0.018 &  0.007 &  0.009 & 0.099 & 0.159 & 0.088 & 0.100 \\
       & & $\lambda$ & -0.001 & -0.005 &  0.000 & -0.001 & 0.011 & 0.010 & 0.006 & 0.006 \\
    \bottomrule
  \end{tabular*}
  
  \par
  \vspace{0.5em}
  \begin{minipage}{\textwidth}
    \footnotesize
    \textit{Note:} The bias for $\alpha$ and $\lambda$ is the raw bias, while the bias for $\beta$ is the $L_2$ norm of the bias vector.
  \end{minipage}
\end{table}

\begin{table}[htbp]
  \centering
  \caption{Estimation Performance of Box--Cox Model Parameters ($\lambda=0.5$)}
  \label{tab:boxcox_estimation_2}
  
  \footnotesize 
  \renewcommand{\arraystretch}{1.15}
  
  % --- 表格定义 ---
  % 1. \setlength{\tabcolsep}{0pt} 让 extrolsep 完全控制内部间距
  % 2. 在列定义首尾添加 @{\hspace{1em}} 为左右两侧留出呼吸空间
  \setlength{\tabcolsep}{0pt} 
  \begin{tabular*}{\textwidth}{
    @{\hspace{1em}\extracolsep{\fill}} % 左侧留白 + 自动填充
    c c c 
    *{3}{S[table-format=-1.3]} % Bias 三列
    *{3}{S[table-format=1.3]}  % RMSE 三列
    @{\hspace{1em}}            % 右侧留白
  }
    \toprule
    & & & \multicolumn{3}{c}{Bias} & \multicolumn{3}{c}{RMSE} \\
    \cmidrule(lr){4-6} \cmidrule(lr){7-9}
    
    % S列表头需用 {} 包裹
    {$n$} & {$p$} & {Param.} & {NL2S} & {Shin-MD} & {KMD} & {NL2S} & {Shin-MD} & {KMD} \\
    \midrule

    \multicolumn{9}{l}{\textit{Panel A: Homoskedastic Errors}} \\
    \addlinespace[0.2em]
    \multirow{9}{*}{$100$} & \multirow{3}{*}{$3$} 
         & $\alpha$  &  0.019 &  0.065 &  0.009 & 0.215 & 0.232 & 0.153 \\
       & & $\beta$   &  0.026 &  0.176 &  0.010 & 0.133 & 0.475 & 0.130 \\
       & & $\lambda$ &  0.014 & -0.018 &  0.005 & 0.109 & 0.210 & 0.066 \\
    \addlinespace[0.2em]
                           & \multirow{3}{*}{$5$} 
         & $\alpha$  & -0.007 &  0.052 & -0.002 & 0.224 & 0.273 & 0.162 \\
       & & $\beta$   &  0.028 &  0.074 &  0.014 & 0.134 & 0.323 & 0.133 \\
       & & $\lambda$ &  0.001 & -0.001 &  0.001 & 0.068 & 0.124 & 0.040 \\
    \addlinespace[0.2em]
                           & \multirow{3}{*}{$10$} 
         & $\alpha$  & -0.077 & -0.074 & -0.049 & 0.246 & 0.271 & 0.168 \\
       & & $\beta$   &  0.022 &  0.038 &  0.021 & 0.140 & 0.214 & 0.139 \\
       & & $\lambda$ & -0.015 & -0.016 & -0.011 & 0.044 & 0.040 & 0.026 \\
    \addlinespace[0.5em]
    \multirow{9}{*}{$200$} & \multirow{3}{*}{$3$} 
         & $\alpha$  &  0.008 &  0.066 &  0.005 & 0.154 & 0.219 & 0.118 \\
       & & $\beta$   &  0.019 &  0.198 &  0.008 & 0.093 & 0.470 & 0.095 \\
       & & $\lambda$ &  0.007 & -0.025 &  0.003 & 0.077 & 0.204 & 0.051 \\
    \addlinespace[0.2em]
                           & \multirow{3}{*}{$5$} 
         & $\alpha$  & -0.008 &  0.081 & -0.004 & 0.162 & 0.318 & 0.113 \\
       & & $\beta$   &  0.013 &  0.173 &  0.006 & 0.093 & 0.422 & 0.095 \\
       & & $\lambda$ & -0.000 & -0.019 & -0.000 & 0.050 & 0.161 & 0.029 \\
    \addlinespace[0.2em]
                           & \multirow{3}{*}{$10$} 
         & $\alpha$  & -0.056 & -0.063 & -0.037 & 0.197 & 0.250 & 0.129 \\
       & & $\beta$   &  0.007 &  0.023 &  0.007 & 0.095 & 0.187 & 0.095 \\
       & & $\lambda$ & -0.010 & -0.014 & -0.007 & 0.036 & 0.044 & 0.020 \\

    \addlinespace[0.5em]
    
    \multicolumn{9}{l}{\textit{Panel B: Heteroskedastic Errors}} \\
    \addlinespace[0.2em]
    \multirow{9}{*}{$100$} & \multirow{3}{*}{$3$} 
         & $\alpha$  & -0.001 &  0.025 & -0.021 & 0.225 & 0.227 & 0.156 \\
       & & $\beta$   &  0.037 &  0.154 &  0.010 & 0.145 & 0.459 & 0.139 \\
       & & $\lambda$ &  0.009 & -0.038 & -0.009 & 0.125 & 0.191 & 0.069 \\
    \addlinespace[0.2em]
                           & \multirow{3}{*}{$5$} 
         & $\alpha$  & -0.028 &  0.023 & -0.029 & 0.232 & 0.302 & 0.167 \\
       & & $\beta$   &  0.020 &  0.099 &  0.005 & 0.143 & 0.360 & 0.142 \\
       & & $\lambda$ & -0.006 & -0.024 & -0.010 & 0.072 & 0.135 & 0.044 \\
    \addlinespace[0.2em]
                           & \multirow{3}{*}{$10$} 
         & $\alpha$  & -0.098 & -0.090 & -0.066 & 0.281 & 0.373 & 0.187 \\
       & & $\beta$   &  0.020 &  0.061 &  0.021 & 0.145 & 0.238 & 0.145 \\
       & & $\lambda$ & -0.018 & -0.026 & -0.014 & 0.050 & 0.068 & 0.030 \\
    \addlinespace[0.5em]
    \multirow{9}{*}{$200$} & \multirow{3}{*}{$3$} 
         & $\alpha$  &  0.013 &  0.033 & -0.007 & 0.171 & 0.189 & 0.117 \\
       & & $\beta$   &  0.027 &  0.186 &  0.010 & 0.099 & 0.463 & 0.097 \\
       & & $\lambda$ &  0.011 & -0.046 & -0.005 & 0.093 & 0.205 & 0.053 \\
    \addlinespace[0.2em]
                           & \multirow{3}{*}{$5$} 
         & $\alpha$  & -0.006 &  0.057 & -0.008 & 0.188 & 0.313 & 0.117 \\
       & & $\beta$   &  0.015 &  0.155 &  0.009 & 0.098 & 0.397 & 0.098 \\
       & & $\lambda$ &  0.002 & -0.025 & -0.002 & 0.061 & 0.151 & 0.031 \\
    \addlinespace[0.2em]
                           & \multirow{3}{*}{$10$} 
         & $\alpha$  & -0.084 & -0.097 & -0.049 & 0.221 & 0.235 & 0.133 \\
       & & $\beta$   &  0.017 &  0.033 &  0.014 & 0.101 & 0.161 & 0.100 \\
       & & $\lambda$ & -0.015 & -0.018 & -0.009 & 0.041 & 0.028 & 0.022 \\
    \bottomrule
  \end{tabular*}
  
  \par
  \vspace{0.5em}
  \begin{minipage}{\textwidth}
    \footnotesize
    \textit{Note:} The bias for $\alpha$ and $\lambda$ is the raw bias, while the bias for $\beta$ is the $L_2$ norm of the bias vector.
  \end{minipage}
\end{table}

Tables \ref{tab:boxcox_estimation_1} and \ref{tab:boxcox_estimation_2} summarize the finite-sample estimation results for the Box--Cox transformation model under $\lambda_0=0$ and $\lambda_0=0.5$, respectively.\footnote{To conserve space, estimation results for $\lambda_0=1$ are relegated to Section \ref{subsec:boxcox-lambda1} of the Online Supplementary Material. These results exhibit patterns similar to the $\lambda_0=0.5$ case.} For the log-linear specification ($\lambda_0=0$, Table \ref{tab:boxcox_estimation_1}), the proposed KMD estimator demonstrates robust performance across all considered designs. In terms of consistency, KMD exhibits negligible bias, comparable to the infeasible Oracle GMM. Regarding efficiency, the KMD estimator consistently yields lower RMSEs than the feasible consistent alternatives, NL2S and Shin-MD, in most scenarios. Under homoskedasticity (Panel A), the KMD estimator performs remarkably well, exhibiting precision levels that closely approach the semiparametric efficiency bound. Furthermore, under heteroskedasticity (Panel B), where GMM$^*$ explicitly utilizes the true conditional variance, the unweighted KMD estimator remains highly competitive, delivering precision levels comparable to the theoretical efficiency bound.

The advantages of the KMD framework are most pronounced in the specification with $\lambda_0=0.5$ (Table \ref{tab:boxcox_estimation_2}). Across the considered designs, the KMD estimator consistently yields the lowest finite-sample bias and RMSE. Notably, KMD significantly outperforms the Shin-MD estimator, with the performance gap likely attributable to the covariate dimensionality. Whereas \cite{shin2008semiparametric} relies on indicator-weighted cumulative moments (e.g., $\mathbb{I}(X_i \leq X_j)$) that are vulnerable to data sparsity in higher-dimensional spaces, the KMD approach employs a continuous kernel metric that adapts effectively to the multivariate geometry. Consequently, the KMD estimator maintains high precision without substantial deterioration as the dimension $p$ increases, offering a robust alternative for generalized regression models in nonlinear settings.

We next test the validity of the Box--Cox specification:
\begin{equation*}
    \mathbb{H}_0: \mathbb{E}[Y^{(\lambda_0)} - (\alpha_0 + X^{\prime}\beta_0) \mid X] = 0 
    \quad \text{a.s. for some } \theta_0 \in \Theta.
\end{equation*}
Focusing on the log-linear case ($\lambda_0=0$), we define the residual as $\rho(Z, \theta) = \ln Y - (\alpha + X^{\prime}\beta)$.

To evaluate power, we generate data under the alternatives by introducing an additive drift  $f(X_i)$ into the latent index: $Y_i = \exp(\alpha_0 + X_i^{\prime}\beta_0 + f(X_i) + \epsilon_i)$. We consider three representative forms of misspecification:
\begin{itemize}
    \setlength{\itemsep}{0pt}
    \setlength{\parskip}{0pt}
    \setlength{\parsep}{0pt}
    \item $\mathbb{H}_{1a}$ (Quadratic): $f(X_i)=\delta X_{i,1}^{2}$, representing smooth global nonlinearity.
    \item $\mathbb{H}_{1b}$ (Discontinuous): $f(X_i)=\delta \cdot \mathbb{I}(|X_{i,1}|\le 0.5)$, representing a localized structural break or bump that is often harder to detect.
    \item $\mathbb{H}_{1c}$ (Interaction): $f(X_i)=\delta X_{i,1}X_{i,2}$, capturing omitted interaction effects.
\end{itemize}
The perturbation magnitude is set to $\delta=0.5$ for $\mathbb{H}_{1a}$ and $\mathbb{H}_{1c}$, and to $\delta=1.5$ for $\mathbb{H}_{1b}$.

\begin{table}[htbp]
  \centering
  \caption{Empirical Rejection Rates for Box--Cox Specification Tests ($\alpha=0.05$)}
  \label{tab:boxcox_test}
  
  \small % 使用小号字体
  \renewcommand{\arraystretch}{1.15} 
  
  % 定义等宽居中列 Y
  \newcolumntype{Y}{>{\centering\arraybackslash}X}
  
  % 使用 tabularx，前两列 c (自然宽度居中)，后六列 Y (等宽居中)
  \begin{tabularx}{\textwidth}{cc YYY YYY}
    \toprule
    & & \multicolumn{3}{c}{Homoskedastic Errors} & \multicolumn{3}{c}{Heteroskedastic Errors} \\
    \cmidrule(lr){3-5} \cmidrule(lr){6-8}
    $n$ & DGP & $p=3$ & $p=5$ & $p=10$ & $p=3$ & $p=5$ & $p=10$ \\
    \midrule
    
    % --- n=100 ---
    \multirow{4}{*}{$100$} 
      & Size ($\mathbb{H}_0$)             & $0.046$ & $0.058$ & $0.039$ & $0.060$ & $0.054$ & $0.047$ \\
      & Power ($\mathbb{H}_{1a}$: Quad.)  & $0.619$ & $0.651$ & $0.549$ & $0.674$ & $0.671$ & $0.489$ \\
      & Power ($\mathbb{H}_{1b}$: Disc.)  & $0.821$ & $0.545$ & $0.227$ & $0.745$ & $0.492$ & $0.231$ \\
      & Power ($\mathbb{H}_{1c}$: Inter.) & $0.640$ & $0.587$ & $0.540$ & $0.677$ & $0.587$ & $0.495$ \\
      
    \addlinespace[0.5em]
    
    % --- n=200 ---
    \multirow{4}{*}{$200$} 
      & Size ($\mathbb{H}_0$)             & $0.056$ & $0.059$ & $0.035$ & $0.052$ & $0.046$ & $0.045$ \\
      & Power ($\mathbb{H}_{1a}$: Quad.)  & $0.953$ & $0.984$ & $0.935$ & $0.968$ & $0.979$ & $0.905$ \\
      & Power ($\mathbb{H}_{1b}$: Disc.)  & $0.986$ & $0.908$ & $0.574$ & $0.928$ & $0.798$ & $0.482$ \\
      & Power ($\mathbb{H}_{1c}$: Inter.) & $0.930$ & $0.949$ & $0.931$ & $0.973$ & $0.940$ & $0.883$ \\

    \addlinespace[0.5em]
    
    % --- n=400 ---
    \multirow{4}{*}{$400$} 
      & Size ($\mathbb{H}_0$)             & $0.046$ & $0.047$ & $0.043$ & $0.064$ & $0.053$ & $0.050$ \\
      & Power ($\mathbb{H}_{1a}$: Quad.)  & $1.000$ & $1.000$ & $1.000$ & $1.000$ & $1.000$ & $1.000$ \\
      & Power ($\mathbb{H}_{1b}$: Disc.)  & $1.000$ & $0.999$ & $0.954$ & $0.996$ & $0.995$ & $0.928$ \\
      & Power ($\mathbb{H}_{1c}$: Inter.) & $1.000$ & $1.000$ & $0.999$ & $1.000$ & $0.998$ & $0.999$ \\
      
    \bottomrule
  \end{tabularx}
\end{table}

Table \ref{tab:boxcox_test} summarizes the empirical rejection rates for the Box--Cox specification test. Under the null hypothesis, the test maintains accurate size control across all considered dimensions and error structures, with rejection rates close to the nominal level. Regarding power, we observe varying sensitivity to covariate dimensionality. For the smooth alternatives (Quadratic $\mathbb{H}_{1a}$ and Interaction $\mathbb{H}_{1c}$), the empirical power remains relatively robust as the dimension increases from $p=3$ to $p=10$. In contrast, the discontinuous alternative ($\mathbb{H}_{1b}$) exhibits a more substantial reduction in power in higher-dimensional settings, reflecting the challenge of detecting local irregularities in sparse spaces. However, these distortions vanish with larger sample sizes; rejection rates for all alternatives converge to 1 as $n$ increases to 400, confirming the test's consistency.

\subsection{Linear and Nonlinear Instrumental Variable Models}
\label{subsec:iv_models}

Finally, we evaluate the finite-sample performance of the KMD estimator in the presence of endogeneity. We consider three DGPs with nonlinear reduced forms, including two continuous-instrument designs and one mixed-instrument design.

For Designs 1--2, we generate a $p$-dimensional continuous instrument vector
$Z_i=(Z_{i1},\ldots,Z_{ip})'$, where the components $Z_{ij}\sim \mathcal{U}[0,3]$ are mutually independent. For Design 3, the conditioning vector is $W_i=(C_i',D_i')'$, where $C_i=(C_{i1},\ldots,C_{ip_c})'$ contains independent $\mathcal{U}[0,3]$ continuous components and $D_i=(D_{i1},\ldots,D_{ip_d})'$ is binary, as is common in empirical applications. Let $\Lambda(a)=1/(1+\exp(-a))$. Conditional on $C_i$, we generate
$$
    D_{im}\sim \operatorname{Bernoulli}(\pi_i),
    \qquad
    \pi_i=\Lambda\{0.75(C_{i1}-1.5)\},
    \qquad m=1,\ldots,p_d,
$$
which induces dependence between the continuous and discrete components. We set $(p,p_c,p_d)=(3,2,1),(5,3,2),(10,5,5)$, with $p=p_c+p_d$. Endogeneity is induced through correlated latent errors:
$$
    \begin{pmatrix}\eta_i \\ v_i\end{pmatrix}
    \sim
    \mathcal{N}\left(
    \begin{pmatrix}0\\0\end{pmatrix},
    \begin{pmatrix}1&0.8\\0.8&1\end{pmatrix}
    \right),
    \qquad
    u_i=v_i,\qquad
    \epsilon_i=\sigma_i\eta_i .
$$
We report both homoskedastic and heteroskedastic designs. In the homoskedastic case, $\sigma_i=1$. In the heteroskedastic case, $\sigma_i=\exp\{0.25(Z_{i1}-1.5)\}$ for Designs 1--2, while for Design 3, $\sigma_i=\exp\{0.25(C_{i1}-1.5)+0.15A_i\}$, with $A_i$ introduced below. We set $\theta_0=0.5$ and specify the three designs as follows:
\begin{itemize}
    \setlength{\itemsep}{0pt}
    \setlength{\parskip}{0pt}
    \setlength{\parsep}{0pt} 
    \item \textit{Design 1: Linear Structural Equation.}
    $$
        Y_i = \theta_0 X_i + \epsilon_i, 
        \qquad
        X_i = Z_{i1} + 0.5 Z_{i2}^{2} + u_i .
    $$
    \item \textit{Design 2: Nonlinear Structural Equation.}
    $$
        Y_i = \exp(1+\theta_0 X_i) + \epsilon_i, 
        \qquad
        X_i = 0.2 Z_{i1} + \cos(Z_{i2}Z_{i3}) + u_i .
    $$
    \item \textit{Design 3: Nonlinear Structural Equation with Mixed Instruments.}
    $$
        Y_i = \exp(1+\theta_0 X_i) + \epsilon_i,
        \qquad
        X_i = 0.2C_{i1}+0.5\cos(C_{i1}C_{i2})+A_i+B_i\cos(C_{i1}C_{i2})+u_i,
    $$
    where $A_i=p_d^{-1/2}\sum_{m=1}^{p_d}(D_{im}-\pi_i)$ and $B_i=p_d^{-1/2}\sum_{m=1}^{p_d}(-1)^{m+1}(D_{im}-\pi_i)$.
\end{itemize}

For estimation, we benchmark KMD against alternative methods categorized by their handling of endogeneity and nonlinearity. A useful feature of our reproducing kernel formulation is that the kernel can be tailored to the support of the conditioning variables. We use a Gaussian kernel with the median-heuristic bandwidth in Designs 1--2. In the mixed-data Design 3, KMD uses a product kernel $K(W_i,W_j)=K_C(C_i,C_j)K_D(D_i,D_j)$, where $K_C$ is Gaussian and $K_D$ is a coordinatewise soft categorical kernel with $k(d,d')=\mathbf{1}\{d=d'\}+0.5\,\mathbf{1}\{d\neq d'\}$.

\begin{itemize}
    \setlength{\itemsep}{0pt}
    \setlength{\parskip}{0pt}
    \setlength{\parsep}{0pt}
    
    \item \textit{Plug-in 2SLS:} We consider an oracle plug-in estimator based on the true conditional mean $\mathbb{E}[X| Z]$. While valid in linear designs, plug-in two-stage procedures are generally inconsistent in nonlinear structural models because of the forbidden-regression problem \citep{wooldridge2010econometric}.
    
    \item \textit{GMM Approaches:} We report \textit{Oracle GMM}, based on DGP-specific moment functions, and \textit{Sieve-GMM}, based on polynomial sieve moments with the efficient GMM weighting matrix.
    
    \item \textit{Smooth Minimum Distance (SMD):} Following \citet{lavergne2013smooth}, we include a smooth minimum distance benchmark based on the off-diagonal $U$-statistic criterion. We report two bandwidth choices: a fixed bandwidth $h=1$, denoted SMD-h1, and a dimension-dependent smoothing bandwidth $ h_{\dim}=n^{-1/(p+4)}$, denoted SMD-hdim.
    
    \item \textit{DML-OS:} To handle higher-dimensional nuisance components, we also report a cross-fitted orthogonal-score benchmark, motivated by \citet{chernozhukov2018double}, and implemented with random forests.\footnote{Implementation details for the DML-OS are provided in Section \ref{subsec:dml-details} of the Online Supplementary Material.}
\end{itemize}

\begin{table}[htbp]
  \centering
  \vspace{-2em}
  \caption{Estimation Performance of IV Estimators (Bias and RMSE)}
  \label{tab:iv_estimation}
  
  \footnotesize 
  \renewcommand{\arraystretch}{1.15} 
  
  % --- 表格定义 ---
  % 10列：3个标签列 + 7个数据列
  % S[table-format=-2.3]: 预留两位整数位，适应 Panel B 可能出现的较大数值
  \setlength{\tabcolsep}{0pt}
  \begin{tabular*}{\textwidth}{
    @{\hspace{1em}\extracolsep{\fill}}
    c c l                  
    *{7}{S[table-format=-2.3]} 
    @{\hspace{1em}}
  }
    \toprule
    & & & \multicolumn{7}{c}{Estimators} \\
    \cmidrule(lr){4-10} 
    
    % 表头
    {$n$} & {$p$} & {Metric$^{\dagger}$} & {2SLS$^*$} & {GMM$^*$} & {Sieve-GMM} & {SMD-h1} & {SMD-hdim} & {DML-OS} & {KMD} \\
    \midrule
    
    % =======================================================
    % Panel A: Linear Structural Equation (已填充数据)
    % =======================================================
    \multicolumn{10}{l}{\textit{Panel A: Linear Structural Equation}} \\
    \addlinespace[0.2em]
    \multirow{6}{*}{$100$}
  & \multirow{2}{*}{$3$} & Bias & 0.092 & 0.053 & 0.266 & -0.081 & -0.077 & -0.105 & 0.120 \\
  &                      & RMSE & 2.976 & 3.065 & 3.221 & 3.295 & 3.363 & 3.064 & 3.272 \\
  \addlinespace[0.3em]
  & \multirow{2}{*}{$5$} & Bias & -0.064 & -0.086 & 0.443 & -0.318 & -0.377 & -0.282 & -0.079 \\
  &                      & RMSE & 2.928 & 3.083 & 3.524 & 3.381 & 3.824 & 3.057 & 3.237 \\
  \addlinespace[0.3em]
  & \multirow{2}{*}{$10$} & Bias & 0.105 & 0.148 & 1.144 & -0.124 & -0.470 & -0.077 & 0.198 \\
  &                      & RMSE & 2.932 & 3.060 & 4.121 & 3.997 & 6.149 & 3.071 & 3.317 \\
\addlinespace[0.5em]
\multirow{6}{*}{$200$}
  & \multirow{2}{*}{$3$} & Bias & 0.040 & 0.037 & 0.136 & -0.125 & -0.118 & -0.072 & -0.025 \\
  &                      & RMSE & 2.106 & 2.125 & 2.188 & 2.295 & 2.318 & 2.123 & 2.266 \\
  \addlinespace[0.3em]
  & \multirow{2}{*}{$5$} & Bias & 0.062 & 0.069 & 0.335 & -0.072 & -0.039 & -0.058 & 0.005 \\
  &                      & RMSE & 2.167 & 2.194 & 2.366 & 2.389 & 2.524 & 2.230 & 2.389 \\
  \addlinespace[0.3em]
  & \multirow{2}{*}{$10$} & Bias & -0.014 & -0.007 & 0.552 & -0.166 & -0.220 & -0.133 & -0.064 \\
  &                      & RMSE & 2.102 & 2.126 & 2.458 & 2.747 & 3.966 & 2.147 & 2.408 \\

    \midrule
    
    % =======================================================
\multicolumn{10}{l}{\textit{Panel B: Nonlinear (Exponential) Structural Equation}} \\
\addlinespace[0.2em]
\multirow{6}{*}{$100$}
  & \multirow{2}{*}{$3$} & Bias & 9.537 & -0.091 & 0.503 & -0.585 & -0.562 & -0.530 & 0.001 \\
      &                      & RMSE & 10.248 & 1.656 & 1.850 & 2.289 & 2.395 & 2.310 & 1.930 \\
      \addlinespace[0.3em]
      & \multirow{2}{*}{$5$} & Bias & 9.382 & -0.130 & 0.892 & -0.743 & -1.174 & -0.612 & -0.022 \\
      &                      & RMSE & 10.036 & 1.621 & 1.903 & 2.572 & 4.805 & 2.238 & 1.946 \\
      \addlinespace[0.3em]
      & \multirow{2}{*}{$10$} & Bias & 9.766 & -0.005 & 1.766 & -1.986 & -19.627 & -0.819 & 0.083 \\
      &                      & RMSE & 10.429 & 1.643 & 2.550 & 12.529 & 69.540 & 6.210 & 2.030 \\
\addlinespace[0.5em]
\multirow{6}{*}{$200$}
  & \multirow{2}{*}{$3$} & Bias & 9.654 & -0.034 & 0.266 & -0.271 & -0.275 & -0.248 & -0.010 \\
      &                      & RMSE & 9.990 & 1.153 & 1.297 & 1.486 & 1.482 & 1.445 & 1.395 \\
      \addlinespace[0.3em]
      & \multirow{2}{*}{$5$} & Bias & 9.689 & 0.002 & 0.514 & -0.260 & -0.349 & -0.246 & 0.044 \\
      &                      & RMSE & 10.024 & 1.078 & 1.268 & 1.557 & 1.918 & 1.389 & 1.383 \\
      \addlinespace[0.3em]
      & \multirow{2}{*}{$10$} & Bias & 9.930 & 0.056 & 1.123 & -0.249 & -3.873 & -0.189 & 0.102 \\
      &                      & RMSE & 10.239 & 1.068 & 1.649 & 1.923 & 23.944 & 1.385 & 1.348 \\
  \midrule
      \multicolumn{10}{l}{\textit{Panel C: Nonlinear (Exponential) Structural Equation with Mixed Covariates}} \\
\addlinespace[0.2em]
\multirow{6}{*}{$100$}
 & \multirow{2}{*}{$3$} & Bias & 7.281 & 0.048 & 0.261 & -0.629 & -0.621 & -0.509 & 0.094 \\
      &                      & RMSE & 8.064 & 1.516 & 1.713 & 2.293 & 2.312 & 2.069 & 1.852 \\
      \addlinespace[0.3em]
      & \multirow{2}{*}{$5$} & Bias & 9.343 & 0.221 & 0.818 & -0.759 & -1.414 & -0.943 & 0.359 \\
      &                      & RMSE & 10.405 & 1.738 & 1.968 & 2.835 & 8.190 & 6.091 & 1.926 \\
      \addlinespace[0.3em]
      & \multirow{2}{*}{$10$} & Bias & 8.372 & 0.218 & 1.501 & -3.600 & -27.061 & -1.254 & 1.196 \\
      &                      & RMSE & 10.172 & 1.773 & 2.372 & 18.851 & 85.758 & 8.697 & 1.936 \\
\addlinespace[0.5em]
\multirow{6}{*}{$200$}
  & \multirow{2}{*}{$3$} & Bias & 7.389 & -0.019 & 0.133 & -0.288 & -0.302 & -0.256 & 0.046 \\
      &                      & RMSE & 7.790 & 1.033 & 1.143 & 1.380 & 1.340 & 1.243 & 1.292 \\
      \addlinespace[0.3em]
      & \multirow{2}{*}{$5$} & Bias & 9.469 & 0.133 & 0.410 & -0.349 & -0.478 & -0.319 & 0.209 \\
      &                      & RMSE & 10.017 & 1.232 & 1.357 & 1.682 & 2.076 & 1.587 & 1.431 \\
      \addlinespace[0.3em]
      & \multirow{2}{*}{$10$} & Bias & 8.263 & 0.118 & 0.896 & -0.422 & -4.658 & -0.229 & 0.764 \\
      &                      & RMSE & 9.122 & 1.141 & 1.467 & 2.273 & 28.903 & 1.576 & 1.367 \\
\bottomrule
  \end{tabular*}
  
  \par
  \vspace{0.5em}
  \begin{minipage}{\textwidth}
    \footnotesize
    \textit{Note:} $^{\dagger}$ All values are multiplied by $100$ for readability. 
  \end{minipage}
\end{table}

Table \ref{tab:iv_estimation} reports the estimation results for the IV designs, with all values multiplied by 100 for readability. In the linear structural setting (Panel A), KMD performs comparably to the oracle IV benchmarks. Its bias is small across all sample sizes and dimensions, and its RMSE is of the same order as oracle 2SLS, oracle GMM, and DML-OS. As the dimension increases, KMD remains stable, whereas the sieve- and smoothing-based alternatives exhibit somewhat larger RMSEs across several designs.

Panels B and C show that the advantages of KMD become more visible in nonlinear IV settings. The oracle plug-in 2SLS estimator exhibits large biases in both panels, reflecting the forbidden-regression problem in nonlinear structural models, while oracle GMM serves as an infeasible benchmark based on correctly specified moment functions. Among the feasible alternatives, KMD exhibits low bias and a stable RMSE across sample sizes and dimensions. In Panel B, KMD is broadly comparable to the strongest feasible competitors at low dimensions and becomes more stable as $p$ increases. The SMD estimator with a fixed bandwidth performs reasonably well in several cases; however, the dimension-dependent smoothing bandwidth deteriorates markedly in higher-dimensional designs, illustrating the finite-sample curse of dimensionality faced by local-smoothing-based criteria. In Panel C, the mixed-data design further highlights this stability: KMD remains well behaved when the conditioning variables include both continuous and discrete instruments, whereas DML-OS and the SMD estimators exhibit larger finite-sample variability in several higher-dimensional mixed designs.

We conclude by evaluating the specification test in instrumental variable models. Let $W_i=Z_i$ in Designs 1--2 and $W_i=(C_i',D_i')'$ in Design 3. The null hypothesis asserts the validity of the conditional moment restriction:
\begin{equation*}
    \mathbb{H}_0: \exists \theta_0 \in \Theta \text{ such that } \mathbb{E}[\rho(Y_i,X_i,\theta_0) \mid W_i] = 0 \text{ a.s.,}
\end{equation*}
where $\rho(Y_i,X_i,\theta)=Y_i-\theta X_i$ in Design I and $\rho(Y_i,X_i,\theta)=Y_i-\exp(1+\theta X_i)$ in Designs II--III. To evaluate power, we introduce perturbations of magnitude $\delta=0.1$ that target distinct misspecifications across the three designs. For the linear framework, we generate data under a quadratic alternative to introduce functional-form misspecification: $Y_i=\theta_0X_i+\delta X_i^2+\epsilon_i$. For the continuous nonlinear framework, we simulate an exclusion-type violation by allowing the instrument $Z_{i1}$ to enter the structural equation: $Y_i=\exp(1+\theta_0X_i+\delta Z_{i1})+\epsilon_i$. For the mixed-data nonlinear framework, we perturb the structural equation through the mixed component $A_i$: $Y_i=\exp(1+\theta_0X_i+\delta A_i)+\epsilon_i$.

\begin{table}[htbp]
  \centering
  %\vspace{-4em}
  \caption{Empirical Rejection Rates for IV Specification Tests ($\alpha=0.05$)}
  \label{tab:iv_test}
  
  \small
  \renewcommand{\arraystretch}{1.15}
  \newcolumntype{Y}{>{\centering\arraybackslash}X}
  
  \begin{tabularx}{\textwidth}{cc YYY YYY}
    \toprule
    & & \multicolumn{3}{c}{Homoskedastic Errors} 
      & \multicolumn{3}{c}{Heteroskedastic Errors} \\
    \cmidrule(lr){3-5} \cmidrule(lr){6-8}
    $n$ & DGP 
      & $d=3$ & $d=5$ & $d=10$
      & $d=3$ & $d=5$ & $d=10$ \\
    \midrule
    
    \multicolumn{8}{l}{\textit{Panel A: Linear structural model with continuous instruments}} \\
    \multirow{2}{*}{$100$}
      & Size ($\mathbb H_0$)  
        & $0.043$ & $0.040$ & $0.037$
        & $0.040$ & $0.052$ & $0.031$ \\
      & Power ($\mathbb H_1$) 
        & $0.653$ & $0.518$ & $0.313$
        & $0.620$ & $0.456$ & $0.253$ \\
    \addlinespace[0.3em]
    
    \multirow{2}{*}{$200$}
      & Size ($\mathbb H_0$)  
        & $0.052$ & $0.047$ & $0.040$
        & $0.057$ & $0.040$ & $0.043$ \\
      & Power ($\mathbb H_1$) 
        & $0.931$ & $0.885$ & $0.723$
        & $0.901$ & $0.845$ & $0.686$ \\
    \addlinespace[0.3em]
    
    \multirow{2}{*}{$400$}
      & Size ($\mathbb H_0$)  
        & $0.056$ & $0.039$ & $0.045$
        & $0.048$ & $0.045$ & $0.044$ \\
      & Power ($\mathbb H_1$) 
        & $0.999$ & $0.997$ & $0.974$
        & $0.999$ & $0.997$ & $0.978$ \\
    
    \addlinespace[0.7em]
    \multicolumn{8}{l}{\textit{Panel B: Nonlinear structural model with continuous instruments}} \\
    \multirow{2}{*}{$100$}
      & Size ($\mathbb H_0$)  
        & $0.050$ & $0.046$ & $0.032$
        & $0.048$ & $0.031$ & $0.033$ \\
      & Power ($\mathbb H_1$) 
        & $0.817$ & $0.713$ & $0.508$
        & $0.756$ & $0.639$ & $0.429$ \\
    \addlinespace[0.3em]
    
    \multirow{2}{*}{$200$}
      & Size ($\mathbb H_0$)  
        & $0.056$ & $0.057$ & $0.040$
        & $0.054$ & $0.047$ & $0.035$ \\
      & Power ($\mathbb H_1$) 
        & $0.982$ & $0.960$ & $0.890$
        & $0.979$ & $0.951$ & $0.858$ \\
    \addlinespace[0.3em]
    
    \multirow{2}{*}{$400$}
      & Size ($\mathbb H_0$)  
        & $0.051$ & $0.058$ & $0.042$
        & $0.048$ & $0.055$ & $0.037$ \\
      & Power ($\mathbb H_1$) 
        & $0.998$ & $1.000$ & $0.997$
        & $1.000$ & $1.000$ & $0.994$ \\
    
    \addlinespace[0.7em]
    \multicolumn{8}{l}{\textit{Panel C: Mixed nonlinear structural model with continuous and discrete instruments}} \\
    \multirow{2}{*}{$100$}
      & Size ($\mathbb H_0$)  
        & $0.053$ & $0.048$ & $0.032$
        & $0.044$ & $0.041$ & $0.034$ \\
      & Power ($\mathbb H_1$) 
        & $0.486$ & $0.559$ & $0.355$
        & $0.409$ & $0.482$ & $0.296$ \\
    \addlinespace[0.3em]
    
    \multirow{2}{*}{$200$}
      & Size ($\mathbb H_0$)  
        & $0.053$ & $0.048$ & $0.042$
        & $0.055$ & $0.054$ & $0.032$ \\
      & Power ($\mathbb H_1$) 
        & $0.856$ & $0.914$ & $0.761$
        & $0.863$ & $0.883$ & $0.730$ \\
    \addlinespace[0.3em]
    
    \multirow{2}{*}{$400$}
      & Size ($\mathbb H_0$)  
        & $0.046$ & $0.054$ & $0.042$
        & $0.048$ & $0.052$ & $0.033$ \\
      & Power ($\mathbb H_1$) 
        & $0.994$ & $0.999$ & $0.991$
        & $0.999$ & $0.997$ & $0.992$ \\
    
    \bottomrule
  \end{tabularx}
\end{table}

\FloatBarrier

Table \ref{tab:iv_test} reports the empirical rejection rates for the IV specification tests under both homoskedastic and heteroskedastic errors. Under the null, the test controls size well across all three designs, with rejection rates generally close to the nominal 5\% level; the mild undersizing observed in some high-dimensional cases is small and does not worsen under heteroskedasticity. Under the alternatives, the test shows strong power against the three types of misspecification considered: functional-form misspecification in Panel A, exclusion-type violations in Panel B, and mixed-data nonlinear misspecification in Panel C. Power is lower in smaller samples and tends to decrease with the dimension of the conditioning variables, but it increases rapidly with $n$. By $n=400$, rejection rates are close to one across all panels and error designs, indicating that the test remains effective in detecting violations of the IV conditional moment restrictions.

In summary, the simulation results underscore the KMD framework's universality and robustness across both estimation and testing tasks. In terms of estimation, the proposed estimator exhibits consistent stability across diverse settings, demonstrating distinct advantages particularly in scenarios characterized by nonlinearity and endogeneity. Regarding specification testing, the KMD statistic maintains accurate size control and satisfactory power, whereas the alternative test of \cite{dominguez2004consistent} suffers significantly more from dimensionality; see Section \ref{subsec:dl-test} of the Online Supplementary Material for a detailed comparison.

\section{Empirical Application}
\label{sec:real}

To illustrate the practical relevance of the proposed KMD framework, we revisit the Engel curve application of \citet{blundell2007semi} using their pre-processed 1995 UK Family Expenditure Survey (FES) sample. The data consist of 1,655 households, restricted to married couples with an employed head aged 25 to 55, and are treated as i.i.d. following \citet{blundell2007semi}. The sample contains budget shares for seven commodity groups ($w_{il}$ for $l=1,\dots,7$: food, catering, alcohol, fuel, motor, fares, and leisure), together with log total expenditure ($\ln x_i$), log gross earnings ($\ln v_i$), and the number of children ($k_i$). The data are publicly available as the \texttt{Engel95} dataset in the R package \texttt{np}.

Following \citet{blundell2007semi}, we consider the Engel curve specification
\begin{equation}
    w_{il} = h_l(u_i) + \gamma_l k_i + \epsilon_{il}, \quad l = 1, \dots, 7,
\end{equation}
where $u_i=\ln x_i-\theta k_i$ denotes log effective total expenditure, $\theta$ is a common equivalence-scale parameter, and $\gamma_l$ captures commodity-specific demographic effects. Whereas \citet{blundell2007semi} leave $h_l(\cdot)$ unrestricted and estimate it nonparametrically, our interest lies in testing economically motivated parametric restrictions on Engel curves. This is relevant because empirical demand analysis often relies on parsimonious functional forms, even though their adequacy is ultimately an empirical question. We therefore consider two leading specifications for $h_l(\cdot)$: the linear Working--Leser form \citep{working1943statistical,leser1963forms}, $h_l(u_i)=\alpha_l+\beta_l u_i$, and the QUAIDS form \citep{banks1997quadratic}, $h_l(u_i)=\alpha_l+\beta_l u_i+\lambda_l u_i^2$.\footnote{In this cross-sectional setting, households face common prices, so the price-index terms in the original QUAIDS specification are absorbed into the intercept $\alpha_l$.}

To address the endogeneity of total expenditure, we use the log of gross earnings as an instrument. Consequently, the null hypothesis of correct model specification is formulated as the following CMR:
\begin{equation}
    \mathbb{H}_0: \mathbb{E}[{\epsilon}_{i} \mid {Z}_i] = \mathbb{E}[{w}_{i} - {h}(u_i; \theta) - {\Gamma} k_i \mid {Z}_i] = {0},
\end{equation}
where ${Z}_i$ collects the instrument ($\ln v_i$) and demographic controls ($k_i$). The vector ${w}_i$ contains budget shares, ${h}(\cdot)$ denotes the vector of shape functions, and ${\Gamma} = (\gamma_1, \dots, \gamma_7)^\top$ represents the vector of demographic coefficients. Based on this framework, we test two canonical specifications: the linear \textit{Working--Leser} form\footnote{Note that in the linear specification $h_l(u_i) = \alpha_l + \beta_l(\ln x_i - \theta k_i)$, the equivalence scale parameter $\theta$ is not separately identified as it is collinear with the direct demographic effect $\gamma_l$. In this case, $\theta$ is absorbed into the reduced-form coefficient on $k_i$.}, and the \textit{Quadratic Almost Ideal Demand System (QUAIDS)}.

We implement the KMD test using a Gaussian kernel, with the bandwidth selected via the median heuristic and critical values obtained from $B=2999$ multiplier bootstrap replications. The results decisively favor the quadratic specification: the linear Working--Leser form is rejected at the 5\% level ($p=0.046$), indicating a misspecified functional form. In contrast, the quadratic QUAIDS specification fails to reject the null hypothesis ($p=0.703$), confirming that the inclusion of the quadratic term is sufficient to capture the structural nonlinearities in the demand system.

Given that the QUAIDS specification is not rejected by the data, we report the parameter estimates from this preferred model. We first identify the common equivalence scale parameter shared across the demand system, which we estimate as $\hat{\theta}_1 = 0.049$ (standard error $0.103$). While the positive sign is consistent with theoretical expectations regarding household composition costs, the estimate is not statistically significant in this sample. In light of this, we perform a robustness check by re-estimating the model and repeating the specification tests under the restriction $\theta=0$. The results, detailed in Section \ref{sec:add-empirical} of the Online Supplementary Material, indicate that our primary findings regarding the functional form specification remain robust to the exclusion of the equivalence scale parameter. Table \ref{tab:quaids_estimates} summarizes the remaining commodity-specific coefficients and their corresponding standard errors. 

\begin{table}[htbp]
  \centering
  \caption{Parameter Estimates for the Quadratic (QUAIDS) Specification}
  \label{tab:quaids_estimates}
  
  \small
  \renewcommand{\arraystretch}{1.25} % 适当增加行高
  
  % --- 表格设置 ---
  \setlength{\tabcolsep}{0pt} 
  
  \begin{tabular*}{\textwidth}{
    @{\extracolsep{\fill}} 
    l 
    c 
    c 
    c 
    c 
    }
    \toprule
    Commodity & $\alpha_l$ & $\beta_l$ & $\lambda_l$ & $\gamma_l$ \\
    \midrule
    Food      & $-0.319$ $(0.640)$ & $\phantom{-}0.259$ $(0.232)$ & $-0.031$ $(0.021)$ & $\phantom{-}0.048$ $(0.009)$ \\
    Catering  & $-0.312$ $(0.437)$ & $\phantom{-}0.138$ $(0.160)$ & $-0.012$ $(0.015)$ & $-0.005$ $(0.003)$ \\
    Alcohol   & $-0.524$ $(0.453)$ & $\phantom{-}0.234$ $(0.164)$ & $-0.023$ $(0.015)$ & $-0.023$ $(0.004)$ \\
    Fuel      & $\phantom{-}1.055$ $(0.267)$ & $-0.330$ $(0.097)$ & $\phantom{-}0.027$ $(0.009)$ & $\phantom{-}0.009$ $(0.005)$ \\
    Motor     & $-2.041$ $(0.728)$ & $\phantom{-}0.842$ $(0.264)$ & $-0.080$ $(0.024)$ & $-0.020$ $(0.006)$ \\
    Fares     & $\phantom{-}0.822$ $(0.376)$ & $-0.310$ $(0.138)$ & $\phantom{-}0.030$ $(0.013)$ & $-0.007$ $(0.003)$ \\
    Leisure   & $\phantom{-}1.265$ $(1.030)$ & $-0.564$ $(0.380)$ & $\phantom{-}0.065$ $(0.035)$ & $-0.010$ $(0.015)$ \\
    \bottomrule
  \end{tabular*}
  
  \par
  \vspace{0.5em}
  \begin{minipage}{\textwidth}
    \footnotesize
    \textit{Notes:} The table reports the commodity-specific coefficients with standard errors in parentheses. The parameters correspond to the specification $w_{il} = \alpha_l + \beta_l u_i + \lambda_l u_i^2 + \gamma_l k_i + \epsilon_{il}$.
  \end{minipage}
\end{table}

Inspection of Table \ref{tab:quaids_estimates} indicates that the standard errors for the individual shape parameters $\beta_l$ and $\lambda_l$ are relatively large in some equations. This pattern likely reflects the structural multicollinearity between the linear regressor $u_i$ and its quadratic counterpart $u_i^2$. To assess the influence of total expenditure, we further conduct equation-specific Wald tests for the joint null hypothesis $\mathbb{H}_0: \beta_l = \lambda_l = 0$, utilizing the established asymptotic normality of the KMD estimator. The results, summarized in Table \ref{tab:wald_tests}, show a decisive rejection of the null hypothesis at the 1\% significance level for six out of seven commodities. This finding confirms that, despite inflated standard errors for individual coefficients, the expenditure terms are jointly significant in characterizing the demand system. The sole exception is catering, where the expenditure elasticity does not appear statistically distinguishable from zero.

% --- Table: Wald Test Results (Wide Format) ---
\begin{table}[htbp]
  \centering
  \caption{Wald Tests for Joint Significance of Expenditure Terms ($\mathbb{H}_0: \beta_l = \lambda_l = 0$)}
  \label{tab:wald_tests}
  \small
  \renewcommand{\arraystretch}{1.2}
  \setlength{\tabcolsep}{0pt}
  
  \begin{tabular*}{\textwidth}{@{\extracolsep{\fill}}l ccccccc}
    \toprule
    & Food & {Catering} & {Alcohol} & {Fuel} & {Motor} & {Fares} & {Leisure} \\
    \midrule
    $p$-value & $0.000^{***}$ & $0.435$ & $0.006^{**}$ & $0.000^{***}$ & $0.000^{***}$ & $0.008^{**}$ & $0.000^{***}$ \\
    \bottomrule
  \end{tabular*}
  \par
  \vspace{0.5em}
  \begin{minipage}{\textwidth}
    \footnotesize
    \textit{Note:} Significance levels: $^{*} p<0.05$, $^{**} p<0.01$, $^{***} p<0.001$.
  \end{minipage}
\end{table}

\begin{figure*}[htbp]
    \centering
    \includegraphics[width=0.9\textwidth]{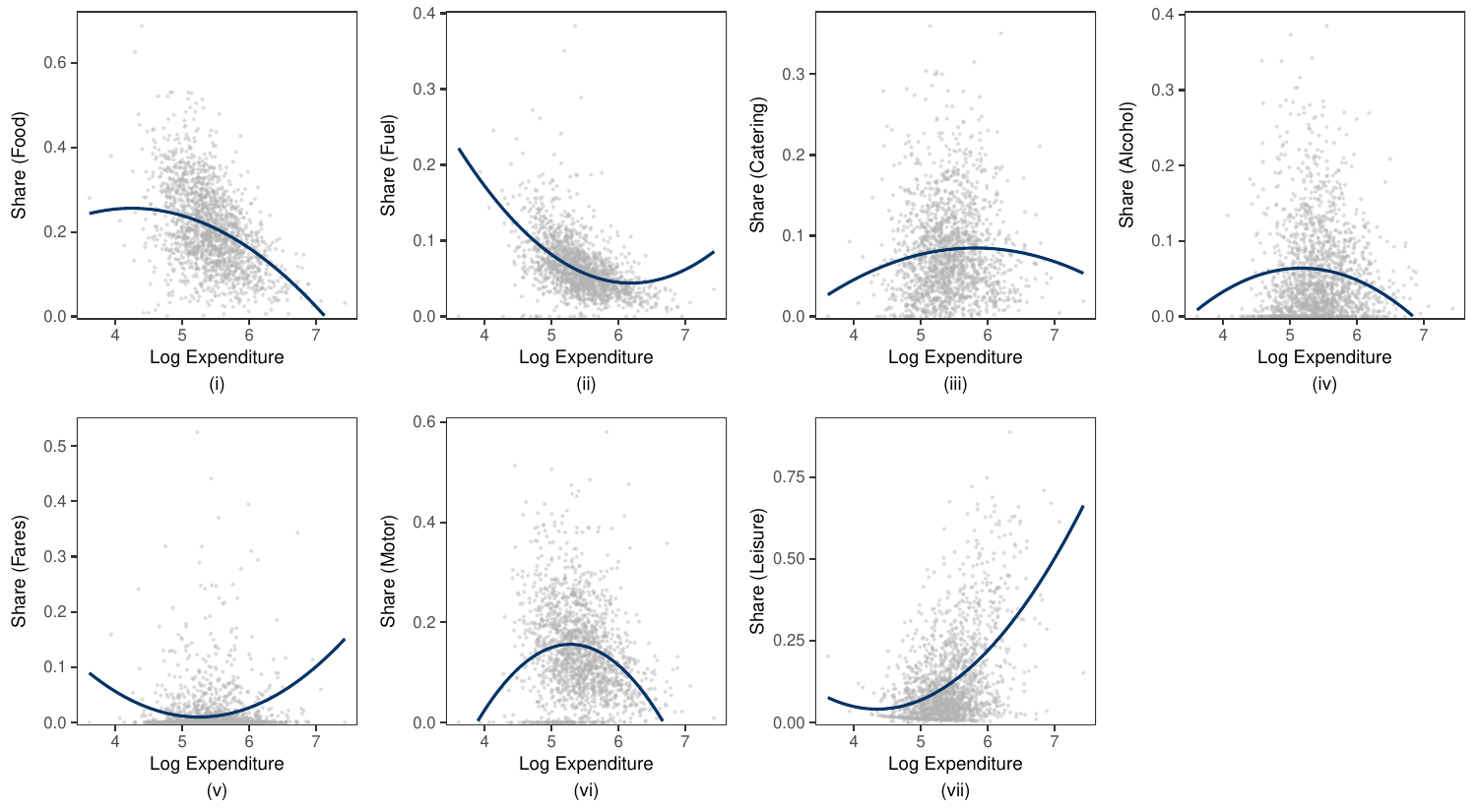}
    \caption{Estimated QUAIDS Engel curves for seven commodity groups. Grey dots represent observed household shares, while solid blue-black lines indicate the fitted quadratic specification. Sub-labels (i) through (vii) refer to food, fuel, catering, alcohol, fares, motor, and leisure, respectively.}
    \label{fig:quaid_results}
\end{figure*}
\FloatBarrier

Finally, Figure \ref{fig:quaid_results} illustrates the economic implications of the estimated model by plotting the predicted Engel curves over the support of the data. The recovered shapes are generally consistent with established demand theory. The curves for food and fuel exhibit a predominant downward trend consistent with Engel's Law, although the fuel curve suggests a mild reversal at the upper tail of the expenditure distribution. In contrast, the curves for motor and alcohol exhibit a distinct inverted-U-shaped trajectory, capturing the ``rank-3'' flexibility inherent in the QUAIDS specification. Notably, this non-monotonic pattern for alcohol is consistent with the empirical findings of \cite{banks1997quadratic}. Furthermore, leisure exhibits a convex, accelerating upward trend, indicative of a luxury good. Collectively, these heterogeneous patterns highlight the limitations of the linear form and support adopting the quadratic specification.

\section{Conclusion}
\label{sec:con}
This paper establishes a unified Kernel Minimum Distance (KMD) framework for estimating and testing models defined by conditional moment restrictions. By embedding conditional moments into a Reproducing Kernel Hilbert Space, we construct a tractable, closed-form $V$-statistic objective function. The proposed framework delivers a $\sqrt{n}$-consistent, asymptotically normal estimator and yields a naturally associated omnibus specification test based directly on the minimized objective value. We establish the asymptotic properties of our proposed tests under the null hypothesis, the fixed alternative, and a sequence of local alternatives converging to the null at the parametric rate $n^{-1/2}$. A distinguishing feature of our approach is that the estimation effect is inherently captured through a projected kernel structure, thereby obviating the need for auxiliary orthogonalization. Supported by a computationally efficient multiplier bootstrap, the proposed framework is practically implementable; extensive simulation results demonstrate robust performance in both estimation and testing across diverse settings. %including nonlinear and endogenous models, and are applicable to multivariate scenarios.

%Several extensions are possible. First, the framework can be extended to settings with dependent data, including time series models, in which asymptotic analysis requires establishing limit theorems for $V$-statistics under dependence. Second, the choice of the reproducing kernel is expected to play an important role in finite-sample performance. Future research may explore data-driven kernel selection procedures or adaptive weighting schemes to improve both the efficiency of parameter estimation and the power of the associated specification test. We leave these interesting topics for future research.

%\clearpage

%\newpage
%\bibliographystyle{plainnat}
%\bibliography{ref}
\putbib % <-- 2. 放置当前单元的参考文献

\end{bibunit}
% 附录开始
\clearpage

\begin{bibunit}
\appendix
\renewcommand{\thesection}{ \Alph{section}}
\renewcommand{\thesubsection}{\thesection.\arabic{subsection}}
\counterwithin{equation}{section}
\counterwithin{figure}{section}
\counterwithin{table}{section}
\counterwithin{lemma}{section}

% 页码从1重新开始
\setcounter{page}{1}
\pagenumbering{arabic}

\begin{center}
    {\LARGE Online Supplementary Material for}\\[1ex]
    {\Large ``Kernel Minimum Distance Estimation and Testing with
Conditional Moment Restrictions: A Unified Framework''}\\[3ex]
    Yuhao Li \quad Haokun Lu \quad Xiaojun Song\\[1ex]
    %\textsuperscript{1}School of Statistics, Renmin University of China\\
    %\textsuperscript{2}Guanghua School of Management, Peking University\$$2ex]
\end{center}

\vspace{1em}

The Online Supplementary Material is organized into four sections. Section \ref{sec:proofs} provides the proofs of all main results. Section \ref{sec:discussion} contains additional methodological discussions, including the two-step efficient refinement and the choice between $U$- and $V$-test statistics. Section \ref{sec:add-simu} presents additional simulation results, and Section \ref{sec:add-empirical} reports additional empirical results.
\section{Proofs}
\label{sec:proofs}
\subsection{Proofs for Section \ref{sec:estimation}}
\paragraph{Proof of Lemma \ref{lemma:identification}}
\begin{proof}
For each component $l=1,\ldots,q$, define the finite signed Borel measure
$$
\nu_{\theta,l}(A)
\coloneqq
\mathbb E\!\left[\rho_l(Z,\theta)\mathbf 1\{X\in A\}\right],
\qquad A\in\mathcal B(\mathcal X).
$$
This measure is finite because Assumption \ref{ass:cons}(ii) implies
$\mathbb E|\rho_l(Z,\theta)|<\infty$. Moreover, by the law of iterated
expectations, $\nu_{\theta,l}$ is absolutely continuous with respect to
$P_X$, with Radon--Nikodym derivative
$m_l(X,\theta)\coloneqq \mathbb E[\rho_l(Z,\theta)\mid X]$.

By the diagonal structure $K(x,\widetilde x)=k(x,\widetilde x)\mathbf I_q$
and the reproducing property,
$$
\|\mu(\theta)\|_{\mathcal H_K}^2
=
\sum_{l=1}^q
\iint_{\mathcal X\times\mathcal X}
k(x,\widetilde x)\,
d\nu_{\theta,l}(x)d\nu_{\theta,l}(\widetilde x).
$$
Equivalently,
$$
\|\mu(\theta)\|_{\mathcal H_K}^2
=
\sum_{l=1}^q
\mathbb E\!\left[
m_l(X,\theta)k(X,\widetilde X)m_l(\widetilde X,\theta)
\right],
$$
where $\widetilde X$ is an independent copy of $X$.

If $\mathbb E[\rho(Z,\theta)\mid X]=0$ a.s., then
$m_l(X,\theta)=0$ a.s. for all $l$, and hence
$\nu_{\theta,l}=0$ for all $l$. Therefore
$\|\mu(\theta)\|_{\mathcal H_K}=0$.

Conversely, suppose $\|\mu(\theta)\|_{\mathcal H_K}=0$. Since $k$ is positive
definite, each term in the preceding sum is nonnegative. Hence every component
satisfies
$$
\iint_{\mathcal X\times\mathcal X}
k(x,\widetilde x)\,
d\nu_{\theta,l}(x)d\nu_{\theta,l}(\widetilde x)
=0,
\qquad l=1,\ldots,q.
$$
By the signed-measure injectivity condition in Assumption \ref{ass:cons}(iv),
this implies $\nu_{\theta,l}=0$ for every $l$. Since
$\nu_{\theta,l}\ll P_X$ with Radon--Nikodym derivative
$m_l(\cdot,\theta)$, it follows that
$m_l(X,\theta)=0$ $P_X$-a.s. for all $l$. Therefore
$\mathbb E[\rho(Z,\theta)\mid X]=0$ a.s., completing the proof.
\end{proof}
\paragraph{Proof of Lemma \ref{lemma:theta_consist}}
\begin{proof}
Let $S_i=(X_i,Z_i)$ and define 
$$
f(S_i,S_j,\theta)
=\rho(Z_i,\theta)'K(X_i,X_j)\rho(Z_j,\theta).
$$
By Assumption \ref{ass:cons}(ii) and (iv), $f(S_i,S_j,\theta)$ is almost surely continuous in $\theta\in\Theta$. Moreover,
\begin{align*}
\mathbb E\!\left[\sup_{\theta\in\Theta}\|f(S_i,S_j,\theta)\|\right]
&=\mathbb E\!\left[\sup_{\theta\in\Theta}\|\rho(Z_i,\theta)'K(X_i,X_j)\rho(Z_j,\theta)\|\right] \\
&\le \sup_{x\in\mathcal X}k(x,x)\,\Big(\mathbb E[\sup_{\theta\in\Theta}\|\rho(Z,\theta)\|]\Big)^2 <\infty, \\
\mathbb E\!\left[\sup_{\theta\in\Theta}\|f(S_i,S_i,\theta)\|\right]
&\le \sup_{x\in\mathcal X}k(x,x)\,\mathbb E[\sup_{\theta\in\Theta}\|\rho(Z,\theta)\|^2] <\infty.
\end{align*}
Hence, by Lemma 8.5 of \cite{newey1994large}, the population criterion
$$
Q(\theta)\coloneqq\mathbb E[f(S_i,S_j,\theta)]
$$
is continuous on $\Theta$ and satisfies
$$
\sup_{\theta\in\Theta}\|\widehat Q_n(\theta)-Q(\theta)\|\xrightarrow{p}0.
$$

By Assumption \ref{ass:cons}, the kernel $K$ is ISPD, implying $Q(\theta)\ge 0$ for all $\theta\in\Theta$. By the identification condition in Assumption \ref{ass:cons}(iii) and the Law of Iterated Expectations, $-Q(\theta)$ attains a unique maximizer at $\theta_0$. Since $Q(\theta)$ is continuous and the uniform convergence above holds, Theorem 2.1 of \cite{newey1994large} yields
$$
\widehat\theta_n\xrightarrow{p}\theta_0.
$$
This completes the proof.
\end{proof}

\paragraph{Proof of Lemma \ref{lemma:theta_lr}}
\begin{proof}

By Assumption \ref{ass:an}, $\nabla_\theta \widehat Q_n(\widehat\theta_n)$ admits a mean‐value expansion around $\theta_0$:
$$
\nabla_\theta\widehat Q_n(\widehat\theta_n)
=\nabla_\theta\widehat Q_n(\theta_0)
+\widehat H_n(\bar\theta_n)(\widehat\theta_n-\theta_0),
$$
where $\bar\theta_n$ lies between $\theta_0$ and $\widehat\theta_n$.  

Since
$$
\nabla_\theta f(S_i,S_j,\theta)
=k(X_i,X_j)\big[g(Z_i,\theta)'\rho(Z_j,\theta)
+g(Z_j,\theta)'\rho(Z_i,\theta)\big],
$$
we obtain
$$
\nabla_\theta\widehat Q_n(\theta)
=\frac{2}{n^2}\sum_{i,j} k(X_i,X_j) g(Z_i,\theta)'\rho(Z_j,\theta),
$$
and the sample Hessian is
$$
\widehat H_n(\theta)
=\frac{2}{n^2}\sum_{i,j}k(X_i,X_j)
\!\left[
\sum_{l=1}^q \rho_l(Z_j,\theta)\nabla^2_{\theta\theta'}\rho_l(Z_i,\theta)
+g(Z_i,\theta)'g(Z_j,\theta)
\right].
$$
First, we show that $\widehat H_n(\bar\theta_n)\xrightarrow{p}2\Delta(\theta)$. To see this, note that $H_n(\theta)$ is a $V-$statistic defined by the kernel
$$\phi(S_i,S_j,\theta)=2\cdot k\left(X_i, X_j\right)\left[\sum_{l=1}^q \rho_l\left(Z_j, \theta\right) \nabla_{\theta \theta^{\prime}}^2 \rho_l\left(Z_i, \theta\right)+g\left(Z_i, \theta\right)' g\left(Z_j, \theta\right)\right],$$
with $\phi(S_i,S_j,\theta)$ almost surely continuous in $\theta\in\mathcal{N}(\theta_0)$. Moreover, direct bound by triangular inequalities yields that $\mathbb E\!\left[\sup_{\theta\in\mathcal{N}(\theta_0)}\|\phi(S_i,S_j,\theta)\|\right]<\infty$ and $\mathbb E\!\left[\sup_{\theta\in\mathcal{N}(\theta_0)}\|\phi(S_i,S_i,\theta)\|\right]<\infty$. By Lemma 8.5 of \cite{newey1994large} again, the population criterion
$$
H(\theta)\coloneqq\mathbb E\left[\phi(S_i,S_j,\theta)\right]
$$
is continuous on $\Theta$ and satisfies
$$
\sup_{\theta\in\mathcal{N}(\theta_0)}\|\widehat H_n(\theta)-H(\theta)\|\xrightarrow{p}0.
$$
Furthermore, note that when $\theta=\theta_0$, $\E\left[\rho_l(Z,\theta_0)|X\right]\equiv0$ for $l=1,\cdots,q$. Thus, 
\begin{align*}
    \E\left[\phi(S_i,S_j,\theta_0)\right]&=2\Delta(\theta_0)+2\E\left\{k\left(X_i, X_j\right)\left[\sum_{l=1}^q \rho_l\left(Z_j, \theta\right) \nabla_{\theta \theta^{\prime}}^2 \rho_l\left(Z_i, \theta\right)\right]\right\}\\
    &=2\Delta(\theta_0)+2\E\left\{k(X_i,X_j)\sum_{l=1}^q\E\left[\rho_l\left(Z_j,\theta_0\right)|X_j\right]\nabla_{\theta \theta^{\prime}}^2 \rho_l\left(Z_i, \theta\right)\right\}=2\Delta(\theta_0).
\end{align*}
Since $\widehat{\theta}_n\xrightarrow{p}\theta_0$ and $\bar\theta$ is between $\theta$ and $\widehat{\theta}_n$, by triangular inequality,
\begin{align*}
    \norm{\widehat H_n(\bar\theta_n)-2\Delta(\theta_0)}&\le\norm{\widehat{H}_n(\bar\theta_n)-H(\bar\theta_n)}+\norm{H(\bar\theta_n)-H(\theta_0)}\\
    &\le \sup_{\theta\in\mathcal{N}(\theta_0)}\norm{\widehat{H}_n(\theta)-H(\theta)}+\norm{H(\bar\theta_n)-H(\theta_0)}\xrightarrow{p}0,
\end{align*}
where the second term follows from the Continuous Mapping Theorem. Thus, $\widehat H_n(\bar\theta_n)\xrightarrow{p}2\Delta(\theta)$.

Next, note that the score $\nabla_\theta\widehat{Q}_n\left(\theta_0\right)$. is also a $V$-statistic with symmetrized kernel
$$
\varrho(S_i,S_j,\theta)=k(X_i,X_j)\left[g(Z_i,\theta)'\rho(Z_j,\theta)+g(Z_j,\theta)'\rho(Z_i,\theta)\right].
$$
Assumption \ref{ass:an}(ii) ensures
$$
\mathbb E[\varrho(S_i,S_i,\theta_0)^2]
\le \left(\sup_{x\in\mathcal{X}}k(x,x)\right)^2\,
\E\left[\|g(Z_i,\theta_0)'\rho(Z_i,\theta_0)\|^2\right]\le\E\left[\|\rho(Z_i,\theta_0)\|^4\right]\E\left[\|g(Z_i,\theta_0)\|^4\right]<\infty.
$$
Thus, by the Hájek projection result for $V$–statistics (e.g. \cite{serfling1980approximation}, §5.7.3),
$$
    \label{eq:nabla_Q}
    \nabla_\theta\widehat Q_n(\theta_0)
=\frac{2}{n}\sum_{i=1}^n G(X_i,\theta_0)^\prime\rho(Z_i,\theta_0)+o_p(n^{-1/2}),
$$
because
\begin{align*}
    \E\left[\varrho(S_i,S_j,\theta_0)|S_i\right]&=k(X_i,X_j)g(Z_i,\theta_0)\E\left[\rho(Z_j,\theta_0)|X_j\right]+\E\left[g(Z_j,\theta_0)'K(X_i,X_j)|X_i\right]\rho(Z_i,\theta_0)\\
    &=G(X_i,\theta_0)^\prime\rho(Z_i,\theta_0).
\end{align*}

The first-order condition implies  $\nabla_\theta\left(\widehat{Q}_n(\widehat{\theta}_n)\right)=\mathbf{0}$. Thus,
$$
\frac{2}{n}\sum_{i=1}^n G(X_i,\theta_0)^\prime\rho(Z_i,\theta_0)+o_p\left(\frac{1}{\sqrt{n}}\right)+\left(2\Delta(\theta_0)+o_p(1)\right)\left(\widehat{\theta}_n-\theta_0\right)=\mathbf{0},
$$
Multiplying both side by $\sqrt{n}$, by Slutsky's Theorem, Continuous Mapping Theorem and Assumption \ref{ass:an} (iii),
$$
    \sqrt{n}\left(\widehat{\theta}_n-\theta_0\right)=-\Delta\left(\theta_0\right)^{-1} \frac{1}{\sqrt{n}} \sum_{i=1}^n G\left(X_i, \theta_0\right)^\prime \rho\left(Z_i, \theta_0\right)+o_p(1).
$$
This completes the proof.

\end{proof}

\paragraph{Proof of Theorem \ref{thm:theta_an}}

\begin{proof}
By Lemma \ref{lemma:theta_lr} we have
$$
\sqrt n(\widehat\theta_n-\theta_0)
=-\Delta(\theta_0)^{-1}\frac{1}{\sqrt n}\sum_{i=1}^n\psi_i+o_p(1),
\qquad
\psi_i\coloneq G(X_i,\theta_0)^\prime\rho(Z_i,\theta_0).
$$
Under Assumption \ref{ass:cons}(iii), $\E[\rho(Z,\theta_0)\mid X]=0$, hence $\E[\psi_i]=0$. By Assumption \ref{ass:an}(ii), the components of $\psi_i$ have finite second moments and
\begin{align*}
    \Omega(\theta_0)&=\E[\psi_i\psi_i']=\E\big[G(X,\theta_0)'\rho(Z,\theta_0)\rho(Z,\theta_0)'G(X,\theta_0)\big]\\
    &=\E_i\left\{\E_j\left[g(Z_j,\theta_0)'K(X_i,X_j)|X_i\right]\rho(Z_i,\theta_0)\rho(Z_i,\theta_0)'\E_j\left[K(X_i,X_j)g(Z_j,\theta_0)|X_i\right]\right\}\\
    &\le \left(\sup_{x\in\mathcal{X}}k(x,x)\right)^2\E\left[\|g(Z_i,\theta_0)'\rho(Z_i,\theta_0)\|^2\right]<\infty,
\end{align*}
is well defined and finite. Thus, the multivariate CLT therefore yields
$$
\frac{1}{\sqrt n}\sum_{i=1}^n\psi_i \ \xrightarrow{d}\ \mathcal N(0,\Omega(\theta_0)).
$$
Combining this with the linear representation and applying Slutsky's theorem gives
$$
\sqrt n(\widehat\theta_n-\theta_0)
\ \xrightarrow{d}\ -\Delta(\theta_0)^{-1}\mathcal N(0,\Omega(\theta_0))
=\mathcal N\big(0,\Delta(\theta_0)^{-1}\Omega(\theta_0)\Delta(\theta_0)^{-1}\big),
$$
as required.
\end{proof}
\paragraph{Proof of Proposition \ref{prop:variance_consistency}}
\begin{proof}
First, by Lemma \ref{lemma:unif_delta}\footnote{Lemma \ref{lemma:unif_G} and \ref{lemma:unif_delta} are established in Section 4 for the multiplier bootstrap. They provide uniform convergence over a neighborhood covering $\theta^*$, which implies the consistency at $\theta_0$ required here.}, we have $\sup_{\theta \in \mathcal{N}^*} \|\widehat{\Delta}_n(\theta)^{-1} - \Delta(\theta)^{-1}\| = o_p(1)$ for some compact neighborhood of $\theta_0$. Combined with the consistency of $\widehat{\theta}_n$ and the continuity of $\Delta(\theta)^{-1}$, it follows directly that $\widehat{\Delta}_n^{-1} \xrightarrow{p} \Delta(\theta_0)^{-1}$.

Next, we establish the consistency of $\widehat{\Omega}_n$. Define the infeasible sample analog 
$$\widetilde{\Omega}_n \coloneqq n^{-1} \sum_{i=1}^n G(X_i, \theta_0)' \rho(Z_i, \theta_0) \rho(Z_i, \theta_0)' G(X_i, \theta_0).$$
By the Law of Large Numbers and the moment conditions in Assumptions \ref{ass:cons} and \ref{ass:an}, $\widetilde{\Omega}_n \xrightarrow{p} \Omega(\theta_0)$. It remains to show $\|\widehat{\Omega}_n - \widetilde{\Omega}_n\| = o_p(1)$. For notation simplicity, let $G_i\coloneq G(X_i,\theta_0)$, $\rho_i\coloneq\rho(Z_i,\theta_0)$, $\Delta G_i \coloneqq \widehat{G}_n(X_i, \widehat{\theta}_n)-G_i$, $\Delta \rho_i \coloneqq \rho(Z_i, \widehat{\theta}_n)-\rho_i$.

For $\Delta G_i$, the triangle inequality yields $\|\Delta G_i\| \le \|\widehat{G}_n(X_i, \widehat{\theta}_n) - G(X_i, \widehat{\theta}_n)\| + \|G(X_i, \widehat{\theta}_n) - G(X_i, \theta_0)\|$. Invoking Lemma \ref{lemma:unif_G} and the Lipschitz continuity of $G(\cdot, \theta)$, we obtain a uniform bound:
\begin{equation*}
    \max_{1 \le i \le n} \|\Delta G_i\| \le \max_{i}\sup_{ \theta}\|\widehat{G}_n(X_i, \theta) - G(X_i, \theta)\| + L_G \|\widehat{\theta}_n - \theta_0\| = o_p(1).
    \label{eq:bound_G}
\end{equation*}
For $\Delta \rho_i$, the Mean Value Theorem implies $\|\Delta \rho_i\| \le \sup_{\theta}\|g(Z_i, \theta)\| \|\widehat{\theta}_n - \theta_0\|$. We obtain:
\begin{equation*}
    \frac{1}{n} \sum_{i=1}^n \|\Delta \rho_i\|^2 \le \|\widehat{\theta}_n - \theta_0\|^2\left(\frac{1}{n}\sum_{i=1}^n \sup_{\theta\in\mathcal{N}(\theta_0)}\|g(Z_i,\theta)\|^2\right) = o_p(1)\cdot O_p(1)=o_p(1).
    \label{eq:bound_rho}
\end{equation*}
Furthermore, $\max_i \|G_i\| \le \sup_{x,x'}k(x,x')\mathbb{E}[\|g\|]<\infty$.

Let $U_i \coloneqq G_i' \rho_i$ and $\widehat{U}_i \coloneqq (G_i + \Delta G_i)' (\rho_i + \Delta \rho_i)$. The difference is $\Delta U_i = \Delta G_i' \rho_i + G_i' \Delta \rho_i + \Delta G_i' \Delta \rho_i$. We bound the mean squared norm $n^{-1} \sum \|\Delta U_i\|^2$ by analyzing each term using the inequality $(a+b+c)^2 \le 3(a^2+b^2+c^2)$:
\begin{enumerate}
    \item $\frac{1}{n} \sum \|\Delta G_i' \rho_i\|^2 \le (\max_i \|\Delta G_i\|)^2 (\frac{1}{n} \sum \|\rho_i\|^2) = o_p(1) \cdot O_p(1) = o_p(1)$.
    \item $\frac{1}{n} \sum \|G_i' \Delta \rho_i\|^2 \le (\max_i \|G_i\|)^2 (\frac{1}{n} \sum \|\Delta \rho_i\|^2) = O(1) \cdot o_p(1) = o_p(1)$.
    \item $\frac{1}{n} \sum \|\Delta G_i' \Delta \rho_i\|^2 \le (\max_i \|\Delta G_i\|)^2 (\frac{1}{n} \sum \|\Delta \rho_i\|^2) = o_p(1) \cdot o_p(1) = o_p(1)$.
\end{enumerate}
Thus, $n^{-1} \sum \|\Delta U_i\|^2 = o_p(1)$. 
\begin{align*}
    \norm{\widehat{\Omega}_n-\widetilde\Omega_n}&=\left\|\frac{1}{n} \sum_{i=1}^n\left(\widehat{U}_i \widehat{U}_i^{\prime}-U_i U_i^{\prime}\right)\right\|\\
    &\le \frac{1}{n} \sum_{i=1}^n\left\|\widehat{U}_i \widehat{U}_i^{\prime}-U_i U_i^{\prime}\right\|\\
    &\le \frac{1}{n} \sum_{i=1}^n\left\|\Delta U_i\right\|^2+2 \frac{1}{n} \sum_{i=1}^n\left\|U_i\right\|\left\|\Delta U_i\right\|\\
    &\le \frac{1}{n} \sum_{i=1}^n\left\|\Delta U_i\right\|^2+2\left(\frac{1}{n} \sum_{i=1}^n\left\|U_i\right\|^2\right)^{1 / 2}\left(\frac{1}{n} \sum_{i=1}^n\left\|\Delta U_i\right\|^2\right)^{1 / 2},
\end{align*}
where the inequalities follow, respectively, from the triangle inequality, the matrix norm bound $\|aa' - bb'\| \le \|a-b\|^2 + 2\|b\|\|a-b\|$, and the Cauchy-Schwarz inequality. Noting $n^{-1} \sum \|U_i\|^2 = O_p(1)$, $\norm{\widehat{\Omega}_n-\widetilde\Omega_n}=o_p(1)$. Thus, $\widehat{\Omega}_n\xrightarrow{p}\Omega(\theta_0)$. Consequently, by Slutsky's theorem, $\widehat{\Sigma}_n = \widehat{\Delta}_n^{-1} \widehat{\Omega}_n \widehat{\Delta}_n^{-1} \xrightarrow{p} \Sigma$.
\end{proof}
\subsection{Proofs for Section \ref{sec:test}}
\label{subsec:proof-for-test}
\paragraph{Proof of Lemma \ref{lemma:repre_H0}}
\begin{proof}
    By Assumption \ref{ass:an}, $ \widehat T_n=n\widehat Q_n(\widehat\theta_n)$ admits a mean‐value expansion around $\theta_0$:
    \begin{equation}
        \label{eq:mv-T}
        \widehat T_n=n\widehat{Q}_n(\theta_0)+\left(\sqrt{n}\nabla_\theta\widehat{Q}_n(\theta_0)\right)'\left[\sqrt{n}\left(\widehat{\theta}_n-\theta_0\right)\right]+\frac{1}{2}\left[\sqrt{n}\left(\widehat{\theta}_n-\theta_0\right)\right]'\widehat{H}_n(\bar{\theta}_n)\left[\sqrt{n}\left(\widehat{\theta}_n-\theta_0\right)\right],
    \end{equation} 
    where $\bar\theta_n$ lies between $\widehat\theta_n$ and $\theta_0$. From the proof of Lemma \ref{lemma:theta_lr}:
    \begin{align*}
        \widehat{H}_n(\bar{\theta}_n)&=2\Delta(\theta_0)+o_p(1),\\
        \sqrt{n}\nabla_\theta\widehat{Q}_n(\theta_0)&=2S_n(\theta_0)+o_p(1),\\
        \sqrt{n}\left(\widehat{\theta}_n-\theta_0\right)&=-\Delta(\theta_0)^{-1}S_n(\theta_0)+o_p(1),
    \end{align*}
    where $S_n(\theta_0)\coloneq 1/\sqrt{n}\sum_{i=1}^nG(X_i,\theta_0)^\prime\rho(Z_i,\theta_0)$.
    Substituting these expansions and applying Slutsky’s theorem yields
    \begin{align*}
        \widehat T_n&=n\widehat{Q}_n(\theta_0)-S_n(\theta_0)'\Delta(\theta_0)^{-1}S_n(\theta_0)+o_p(1)\\
        &=\frac{1}{n}\sum_{i=1}^n\sum_{j=1}^n\rho\left(Z_i,\theta_0\right)^{\prime} \left[K(X_i,X_j)-G(X_i,\theta_0)\Delta(\theta_0)^{-1} G(X_j,\theta_0)^{\prime}\right]\rho\left(Z_j,\theta_0\right)+o_p(1)\\
        &\coloneq \frac{1}{n}\sum_{i=1}^n\sum_{j=1}^n\rho\left(Z_i,\theta_0\right)^{\prime} K^p(X_i,X_j)\rho\left(Z_j,\theta_0\right)+o_p(1).
    \end{align*}
This completes the proof.
\end{proof}

\paragraph{Proof of the PSD of $K^p$}
\begin{proof}
    Let $\mathcal{H}_k$ be the scalar Reproducing Kernel Hilbert Space (RKHS) associated with the scalar kernel $k(\cdot, \cdot)$. Since the matrix kernel is defined as a diagonal matrix $K(x, y)=k(x, y) \mathbf{I}_q$, the associated vector-valued RKHS, denoted by $\mathcal{H}_K$, is isomorphic to the direct sum of $q$ copies of $\mathcal{H}_k$. Specifically, $\mathcal{H}_K$ consists of vectorvalued functions $f=\left(f_1, \ldots, f_q\right)^{\prime}: \mathcal{X} \rightarrow \mathbb{R}^q$ such that each component $f_l \in \mathcal{H}_k$ . This space is equipped with the inner product defined by the sum of the component-wise inner products:
    $$
    \langle f, h\rangle_K=\sum_{l=1}^q\left\langle f_l, h_l\right\rangle_{\mathcal{H}_k}, \quad \forall f, h \in \mathcal{H}_K .
    $$
Recall that $G(x, \theta_0) = \mathbb{E}[K(x, X) g(Z, \theta_0)]$ is a $q \times p$ matrix. Let $g_l(Z, \theta_0)$ denote the $l$-th column of the Jacobian $g(Z, \theta_0)$. We define the component functions $\psi_l \in \mathcal{H}_K$ via the Bochner integrals:
\begin{equation*}
    \psi_l(\cdot) = \mathbb{E}[K(\cdot, X) g_l(Z, \theta_0)], \quad l = 1, \dots, p.
\end{equation*}
By the properties of Bochner integrals, point evaluation yields $\psi_l(x) = G_l(x, \theta_0)$, where $G_l(x, \theta_0)$ is the $l$-th column of $G(x, \theta_0)$. Furthermore, the entries of the matrix $\Delta(\theta_0)$ correspond precisely to the inner products of these functions in $\mathcal{H}_K$. Specifically, applying the reproducing property $\langle f, k(\cdot, X)v \rangle_K = v^{\prime}f(X)$ for $v \in \mathbb{R}^q$ and the independence of $Z$ and its copy $Z'$, we obtain:
\begin{align*}
    \langle \psi_l, \psi_m \rangle_K &= \mathbb{E}_{Z} \left[ \langle k(\cdot, X) g_l(Z, \theta_0), \psi_m \rangle_K \right] \\
    &= \mathbb{E}_{Z} \left[ g_l(Z, \theta_0)^{\prime} \psi_m(X) \right] \\
    &= \mathbb{E}_{Z, Z'} \left[ g_l(Z, \theta_0)^{\prime} k(X, X') g_m(Z', \theta_0) \right].
\end{align*}
This equality confirms that $\Delta(\theta_0)$ is the $p \times p$ Gram matrix of the functions $\{\psi_1, \dots, \psi_p\}$. Let $\mathcal{S} = \text{span}\{\psi_1, \dots, \psi_p\} \subset \mathcal{H}_K$ denote the $p$-dimensional subspace spanned by these functions. Since $\Delta(\theta_0)$ is assumed to be non-singular (Assumption \ref{ass:an}), the set $\{\psi_l\}_{l=1}^p$ is linearly independent.

    To establish the positive semi-definiteness of the $q \times q$ matrix kernel $K^p$, consider an arbitrary set of points $\{x_i\}_{i=1}^n \subset \mathcal{X}$ and arbitrary vectors $\{c_i\}_{i=1}^n \subset \mathbb{R}^q$. Define the function $h = \sum_{i=1}^n K(\cdot, x_i)c_i \in \mathcal{H}_K$. We analyze the quadratic form associated with $K^p$:
    \begin{equation*}
        Q_n(c) \equiv \sum_{i=1}^n \sum_{j=1}^n c_i^{\prime} K^p(x_i, x_j) c_j = \sum_{i,j} c_i^{\prime} K(x_i, x_j) c_j - \sum_{i,j} c_i^{\prime} G(x_i, \theta_0) \Delta(\theta_0)^{-1} G(x_j, \theta_0)^{\prime} c_j.
    \end{equation*}
    The first term is exactly the squared norm $\|h\|_K^2$. For the second term, let $v \in \mathbb{R}^p$ be a vector whose $l$-th component is given by
    \begin{equation*}
        v_l = \sum_{i=1}^n c_i^{\prime} \psi_l(x_i) = \sum_{i=1}^n \langle K(\cdot, x_i)c_i, \psi_l \rangle_K = \langle h, \psi_l \rangle_K.
    \end{equation*}
    In matrix notation, this can be written as $v = \sum_{i=1}^n G(x_i, \theta_0)^\prime c_i$, and the second term can be rewritten as $v^{\prime} \Delta(\theta_0)^{-1} v$. We now show that this term corresponds exactly to the squared norm of the orthogonal projection of $h$ onto $\mathcal{S}$.

    Let $P_{\mathcal{S}} h$ denote the orthogonal projection of $h$ onto $\mathcal{S}$. Since $P_{\mathcal{S}} h \in \mathcal{S}$, it can be uniquely expressed as a linear combination of the basis functions: $P_{\mathcal{S}} h = \sum_{k=1}^p \alpha_k \psi_k$. The coefficient vector $\alpha = (\alpha_1, \dots, \alpha_p)^{\prime}$ is determined by the orthogonality condition, which requires that the residual $h - P_{\mathcal{S}} h$ be orthogonal to each basis function $\psi_l$:
    \begin{equation*}
        \langle h - P_{\mathcal{S}} h, \psi_l \rangle_K = 0 \quad \implies \quad \langle h, \psi_l \rangle_K = \langle P_{\mathcal{S}} h, \psi_l \rangle_K, \quad \forall l=1,\dots,p.
    \end{equation*}
    Substituting the expansion of $P_{\mathcal{S}} h$ into the right-hand side yields
    \begin{equation*}
        v_l = \left\langle \sum_{k=1}^p \alpha_k \psi_k, \psi_l \right\rangle_K = \sum_{k=1}^p \alpha_k \langle \psi_k, \psi_l \rangle_K = \sum_{k=1}^p \Delta(\theta_0)_{lk} \alpha_k.
    \end{equation*}
    In matrix notation, this system implies $v = \Delta(\theta_0) \alpha$. Since $\Delta(\theta_0)$ is non-singular, the unique coefficients are given by $\alpha = \Delta(\theta_0)^{-1} v$.

    Finally, we compute the squared norm of the projection:
    \begin{equation*}
        \begin{aligned}
            \|P_{\mathcal{S}} h\|_K^2 &= \langle P_{\mathcal{S}} h, P_{\mathcal{S}} h \rangle_K = \left\langle \sum_{k=1}^p \alpha_k \psi_k, \sum_{m=1}^p \alpha_m \psi_m \right\rangle_K \\
            &= \sum_{k=1}^p \sum_{m=1}^p \alpha_k \alpha_m \Delta(\theta_0)_{km} = \alpha^{\prime} \Delta(\theta_0) \alpha.
        \end{aligned}
    \end{equation*}
    Substituting $\alpha = \Delta(\theta_0)^{-1} v$, we obtain $\|P_{\mathcal{S}} h\|_K^2 = v^{\prime} \Delta(\theta_0)^{-1} v$.
   
    Therefore, the quadratic form becomes
    \begin{equation*}
        Q_n(c) = \|h\|_K^2 - \|P_{\mathcal{S}} h\|_K^2 = \| (I - P_{\mathcal{S}}) h \|_K^2 \ge 0,
    \end{equation*}
    which proves that $K^p$ is positive semi-definite.
\end{proof}

\paragraph{Proof of Theorem \ref{thm:H0}}
\begin{proof}
By Lemma \ref{lemma:repre_H0}, it suffices to analyze the asymptotic distribution of
$$
V_n \coloneq \frac{1}{n}\sum_{i=1}^n\sum_{j=1}^n 
f^p(S_i,S_j,\theta_0),
\qquad
f^p(s,s',\theta_0)
= \rho(z,\theta_0)'K^p(x,x')\rho(z',\theta_0).
$$
Assumption~\ref{ass:an}(ii) implies
\begin{equation}
\label{L2_fp}
    \E\!\left[ f^p(S_i,S_j,\theta_0)^2 \right] < \infty,
\qquad
\E\!\left[ f^p(S_i,S_i,\theta_0)^2 \right] < \infty,
\end{equation}
so the kernel is square–integrable. Under $\mathbb{H}_0$,
$$
\E\!\left[ f^p(S_i,S_j,\theta_0)\mid S_i \right] = 0,
$$
hence $V_n$ is a degenerate $V$–statistic. Therefore, by Theorem~B in §6.4.1 of \citet{serfling1980approximation} (see also Theorem~3.16 of \citet{shao1999mathematical}),
$$
V_n \xrightarrow{d} \sum_{k=1}^{\infty} \lambda_k W_k^2,
$$
where $\{W_k\}$ are i.i.d.\ $\mathcal N(0,1)$ and $\{\lambda_k\}$ are the eigenvalues of the Hilbert–Schmidt operator 
$$
\mathcal{A}\psi(s)
=\int_{\mathcal{Z}\times\mathcal{X}} f^p(s,s',\theta_0)\psi(s')\,d\mathbb{P}_{XZ}(s').
$$
Slutsky's theorem then yields the stated claim for $\widehat T_n(\widehat\theta_n)$.
\end{proof}
\begin{lemma}
    Under $\mathbb{H}_1$, suppose Assumptions \ref{ass:theta1_con} hold, then $\widehat{\theta}_n\xrightarrow{p}\theta_1$.
\end{lemma}
\begin{proof}
    The proof is identical to that of Lemma~\ref{lemma:theta_consist}, with $\theta_0$ replaced by $\theta_1$, and is therefore omitted.
\end{proof}

\paragraph{Proof of Lemma \ref{lemma:theta1_asy}}
\begin{proof}
The proof proceeds as in the null case. By Assumption~\ref{ass:theta1_an} the map
$\theta\mapsto\nabla_\theta\widehat Q_n(\theta)$ admits a mean-value expansion about $\theta_1$:
$$
\nabla_\theta\widehat Q_n(\widehat\theta_n)
=\nabla_\theta\widehat Q_n(\theta_1)
+\widehat H_n(\bar\theta_n)\,(\widehat\theta_n-\theta_1),
$$
with $\bar\theta_n$ is between $\theta_1$ and $\widehat\theta_n$. As in the null case one verifies that $\widehat H_n(\bar\theta_n)\xrightarrow{p}2\widetilde\Delta(\theta_1)$.
Next, we analyze the score at $\theta_1$. As before, the score is a $V$–statistic with symmetrized kernel
$$
\varrho(S_i,S_j,\theta)
=k(X_i,X_j)\big[g(Z_i,\theta)'\rho(Z_j,\theta)+g(Z_j,\theta)'\rho(Z_i,\theta)\big].
$$
Applying the Hájek projection and Assumption~\ref{ass:theta1_an} yields
$$
\nabla_\theta\widehat Q_n(\theta_1)
=\frac{2}{n}\sum_{i=1}^n\widetilde\psi_i(\theta_1)+o_p\left(\frac{1}{\sqrt{n}}\right),
$$
where the projected summand is
$$
\widetilde\psi_i(\theta_1)
=\E[\varrho(S_i,S_j,\theta_1)\mid S_i]
=G(X_i,\theta_1)^\prime\rho(Z_i,\theta_1)+g(Z_i,\theta_1)'M(X_i,\theta_1).
$$
(Indeed, the second term appears because under $\mathbb H_1$ the conditional moment $\E[\rho(Z,\theta_1)\mid X]$ need not vanish.)

The first-order condition $\nabla_\theta\widehat Q_n(\widehat\theta_n)=\mathbf{0}$ thus gives
$$
\frac{2}{n}\sum_{i=1}^n\widetilde\psi_i(\theta_1)
+\big(2\widetilde\Delta(\theta_1)+o_p(1)\big)(\widehat\theta_n-\theta_1)
+o_p\left(\frac{1}{\sqrt{n}}\right)
=0.
$$
Multiplying by $\sqrt n$ and rearranging yields
$$
\sqrt n(\widehat\theta_n-\theta_1)
=-\widetilde\Delta(\theta_1)^{-1}\frac{1}{\sqrt n}\sum_{i=1}^n\widetilde\psi_i(\theta_1)+o_p(1).
$$
Substituting the expression for $\widetilde\psi_i(\theta_1)$ gives the asserted representation
$$
\sqrt{n}\big(\widehat{\theta}_n-\theta_1\big)
=-\widetilde\Delta(\theta_1)^{-1}\frac{1}{\sqrt{n}}\sum_{i=1}^n\big[
G(X_i,\theta_1)^\prime\rho(Z_i,\theta_1)+g(Z_i,\theta_1)'M(X_i,\theta_1)
\big]+o_p(1).
$$

This completes the proof.
\end{proof}
\paragraph{Proof of Theorem \ref{thm:H1}}
\begin{proof}
    Consider a Taylor expansion of the objective function around the pseudo-true value $\theta_1$:
    $$
    \widehat{Q}_n(\widehat{\theta}_n) = \widehat{Q}_n(\theta_1) + \nabla_\theta \widehat{Q}_n(\theta_1)'(\widehat{\theta}_n - \theta_1) + \frac{1}{2}(\widehat{\theta}_n - \theta_1)' \nabla^2_{\theta\theta} \widehat{Q}_n(\bar{\theta}_n) (\widehat{\theta}_n - \theta_1),
    $$
    where $\bar{\theta}_n$ lies between $\widehat{\theta}_n$ and $\theta_1$.
    
    From the proof of Lemma \ref{lemma:theta1_asy}, we have $\|\nabla_\theta \widehat{Q}_n(\theta_1)\| = O_p(n^{-1/2})$. $\|\widehat{\theta}_n - \theta_1\| = O_p(n^{-1/2})$. Furthermore, the Hessian term is bounded in probability. Rearranging terms and scaling by $\sqrt{n}$ yields that
    \begin{align*}
        \sqrt{n}\left( \widehat{Q}_n(\widehat{\theta}_n) - \widehat{Q}_n(\theta_1) \right) &= \sqrt{n} \nabla_\theta \widehat{Q}_n(\theta_1)' (\widehat{\theta}_n - \theta_1) + \sqrt{n} O_p(\|\widehat{\theta}_n - \theta_1\|^2) \\
        &= \sqrt{n} \cdot O_p(n^{-1/2}) \cdot O_p(n^{-1/2}) + \sqrt{n} \cdot O_p(n^{-1}) \\
        &= O_p(n^{-1/2}) + O_p(n^{-1/2}) = o_p(1).
    \end{align*}
    Therefore,
    $$
    \sqrt{n}\left[ \widehat{Q}_n(\widehat{\theta}_n) - Q(\theta_1) \right] = \sqrt{n}\left[ \widehat{Q}_n(\theta_1) - Q(\theta_1) \right] + o_p(1).
    $$
    
    Under $\mathbb{H}_1$, note that $\widehat{Q}_n(\theta_1)$ is a standard non-degenerate V-statistic with kernel $h(S_i, S_j, \theta_1) = \rho(Z_i, \theta_1)' K(X_i, X_j) \rho(Z_j, \theta_1)$. The first-order projection is $h_1(S_i) = \mathbb{E}[h(S_i, S_j, \theta_1) | S_i] = \rho(Z_i, \theta_1)' M(X_i, \theta_1)$. Therefore, by Theorem~A in §6.4.1 of \citet{serfling1980approximation} (see also Theorem~3.16 of \citet{shao1999mathematical}),
    $$
    \sqrt{n}\left[ \widehat{Q}_n(\theta_1) - Q(\theta_1) \right] \xrightarrow{d} \mathcal{N}\left(0, \sigma^2(\theta_1)\right),
    $$
    where $\sigma^2(\theta_1) = 4 \operatorname{Var}(h_1(S_i))$.
    The result follows from Slutsky's Theorem.
\end{proof}

\begin{lemma}
\label{lemma:consistency_local}
Under $\mathbb{H}_{1n}$, suppose Assumption \ref{ass:local_dgp} holds. Further, assume that Assumptions \ref{ass:cons}(ii)--(iv) and \ref{ass:an} are satisfied with respect to the null measure $\mP$. Then $\widehat{\theta}_n \xrightarrow{p} \theta_0$.
\end{lemma}

\begin{proof}
The result is established by invoking Le Cam's First Lemma to extend the consistency property from the null hypothesis to the sequence of local alternatives. First, observe that under the null measure $\mP$, the estimator $\widehat{\theta}_n$ is consistent for $\theta_0$. To show that this convergence in probability holds under the local alternatives $\mathbb{H}_{1n}$, it suffices to prove that the sequence of measures $\mP^{(n)}$ is contiguous with respect to $\mP$. Consider the log-likelihood ratio $\Lambda_n = \log(d\mP^{(n)}/d\mP)$. Under Assumption \ref{ass:local_dgp}, a second-order Taylor expansion yields
\begin{equation}
\label{eq:loglikelihood}
    \Lambda_n = \sum_{i=1}^n \log \left( 1 + \frac{1}{\sqrt{n}}h(X_i, Z_i) \right) = \frac{1}{\sqrt{n}}\sum_{i=1}^n h(X_i, Z_i) - \frac{1}{2n}\sum_{i=1}^n h(X_i, Z_i)^2 + o_p(1).
\end{equation}
By Assumption \ref{ass:local_dgp}, the perturbation function $s(X, Z)$ has zero mean and finite variance $\sigma_h^2 \equiv \E_{\mP}[h(X, Z)^2]$ under the null. Applying the Central Limit Theorem and the Weak Law of Large Numbers under $\mP$, we obtain
\begin{equation*}
    \Lambda_n \xrightarrow{d} \mathcal{N}\left(-\frac{1}{2}\sigma_h^2, \sigma_h^2\right).
\end{equation*}
According to Le Cam's First Lemma (e.g., Lemma 6.4 in \cite{van2000asymptotic}), the asymptotic normality of the log-likelihood ratio with mean $-\sigma_h^2/2$ implies that $\mP^{(n)}$ is contiguous to $\mP$ (denoted $\mP^{(n)} \triangleleft \mP$). By the definition of contiguity, convergence in probability to zero under $\mP$ implies convergence in probability to zero under $\mP^{(n)}$. Consequently, $\mP^{(n)}(\| \widehat{\theta}_n - \theta_0 \| > \epsilon) \to 0$, establishing the consistency of $\widehat{\theta}_n$ under the local alternatives.
\end{proof}
    
\paragraph{Proof of Lemma \ref{lemma:H1n_lp}}
\begin{proof}
The derivation proceeds by expanding the first-order condition for the estimator around the true parameter $\theta_0$. Since $\widehat{\theta}_n$ minimizes the objective function, it satisfies $\nabla_\theta \widehat{Q}_n(\widehat{\theta}_n) = \mathbf{0}$. By a mean-value expansion, we obtain
\begin{equation}
    \label{eq:mean_value_expansion}
    \mathbf{0} = \nabla_\theta \widehat{Q}_n(\theta_0) + \widehat{H}_n(\bar{\theta}_n)(\widehat{\theta}_n - \theta_0),
\end{equation}
where $\widehat{H}_n(\theta) = \nabla_{\theta\theta'}^2 \widehat{Q}_n(\theta)$ is the Hessian matrix and $\bar{\theta}_n$ lies on the line segment joining $\widehat{\theta}_n$ and $\theta_0$.

We first establish the probability limit of the Hessian matrix under the sequence of local alternatives. The sample Hessian is given by a V-statistic:
\begin{equation*}
    \widehat{H}_n(\theta) = \frac{2}{n^2} \sum_{i=1}^n \sum_{j=1}^n k(X_i, X_j) \left[ \sum_{l=1}^q \rho_l(Z_j, \theta) \nabla_{\theta\theta'}^2 \rho_l(Z_i, \theta) + g(Z_i, \theta)' g(Z_j, \theta) \right].
\end{equation*}
Under the null measure $\mP$, Lemma \ref{lemma:theta_lr} establishes that $\widehat{H}_n(\theta_0) \xrightarrow{p} 2\Delta(\theta_0)$. Under the local alternative $\mP^{(n)}$, recall that $\rho(Z, \theta_0) = \epsilon + n^{-1/2}\delta(X)$. The term involving $\rho$ in the Hessian is thus linear in the perturbation $n^{-1/2}\delta(X)$ plus a centered error term. Consequently, its contribution is asymptotically negligible ($O_p(n^{-1/2})$). The dominant term $g(Z_i, \theta_0)' g(Z_j, \theta_0)$ remains unchanged in the limit because the second moments of $s(X,Z)$ are finite. Since $\mP^{(n)}$ is contiguous to $\mP$, the probability limit derived under $\mP$ is preserved. Combined with the consistency $\widehat{\theta}_n \xrightarrow{p} \theta_0$ (Lemma \ref{lemma:consistency_local}) and the uniform convergence of the Hessian, we conclude that $\widehat{H}_n(\bar{\theta}_n) \xrightarrow{p} 2\Delta(\theta_0)$ under $\mP^{(n)}$.

Next, we analyze the asymptotic behavior of the score vector $\nabla_\theta \widehat{Q}_n(\theta_0)$. Recall that
\begin{equation*}
    \nabla_\theta \widehat{Q}_n(\theta_0) = \frac{2}{n^2} \sum_{i=1}^n \sum_{j=1}^n k(X_i, X_j) g(Z_i, \theta_0)' \rho(Z_j, \theta_0).
\end{equation*}
To explicitly characterize the asymptotic bias, we decompose the residual vector using the specific form of the local alternative from \eqref{H_1n}: $\rho(Z_j, \theta_0) = \epsilon_j + n^{-1/2}\delta(X_j)$, where $\epsilon_j$ is the centered error satisfying $\mathbb{E}_{\mP^{(n)}}[\epsilon_j \mid X_j] = 0$ almost surely. Substituting this decomposition into the gradient expression yields
\begin{equation*}
    \nabla_\theta \widehat{Q}_n(\theta_0) = \underbrace{\frac{2}{n^2} \sum_{i,j} k(X_i, X_j) g(Z_i, \theta_0)' \epsilon_j}_{V_{1n}} + \underbrace{\frac{1}{\sqrt{n}} \frac{2}{n^2} \sum_{i,j} k(X_i, X_j) g(Z_i, \theta_0)' \delta(X_j)}_{V_{2n}}.
\end{equation*}
The first term, $V_{1n}$, is a $V$-statistic with a kernel that is degenerate in the $j$-argument under $\mP^{(n)}$ (since $\mathbb{E}_{\mP^{(n)}}[\epsilon_j \mid X_j] = 0$). By the Hájek projection, we get:
\begin{equation*}
    \sqrt{n} V_{1n} = \frac{2}{\sqrt{n}} \sum_{i=1}^n G(X_i, \theta_0)^\prime \epsilon_i + o_p(1),
\end{equation*}
where $G(X_i, \theta_0) = \mathbb{E}_{\mP}[k(X_i, X_j) g(Z_j, \theta_0)' \mid X_i]$. Note that the projection is calculated with respect to the limit measure $\mP$ due to the convergence of $\mP^{(n)}$ to $\mP$.

The second term, $V_{2n}$, represents the drift-induced bias. $\frac{1}{n^2} \sum_{i,j} k(X_i, X_j) g(Z_i, \theta_0)' \delta(X_j)$ converges in probability (under $\mP^{(n)}$) to the expectation $\mathbb{E}_{\mP}[k(X_i, X_j) g(Z_i, \theta_0)' \delta(X_j)]= B_\delta(\theta_0)$. Consequently, the scaled term behaves as
\begin{equation*}
    \sqrt{n} V_{2n} = 2 B_\delta(\theta_0) + o_p(1).
\end{equation*}

Substituting the asymptotic expressions for $\widehat{H}_n(\bar{\theta}_n)$ and $\sqrt{n}\nabla_\theta \widehat{Q}_n(\theta_0)$ back into \eqref{eq:mean_value_expansion} and rearranging terms, we arrive at
\begin{equation*}
    \sqrt{n}(\widehat{\theta}_n - \theta_0) = -(2\Delta(\theta_0))^{-1} \left( \frac{2}{\sqrt{n}} \sum_{i=1}^n G(X_i, \theta_0)^\prime \epsilon_i + 2 B_\delta(\theta_0) \right) + o_p(1).
\end{equation*}
Simplifying the factor of 2 yields the stated asymptotic linear representation
\begin{equation*}
    \sqrt{n}(\widehat{\theta}_n - \theta_0) = -\Delta(\theta_0)^{-1} \frac{1}{\sqrt{n}} \sum_{i=1}^n G(X_i, \theta_0)^\prime \epsilon_i - \Delta(\theta_0)^{-1} B_\delta(\theta_0) + o_p(1).
\end{equation*}
\end{proof}

\paragraph{Proof of Theorem \ref{thm:H_1n}}
\begin{proof}
Since the sequence of local alternatives $\mP^{(n)}$ is contiguous to the null measure $\mP$ (denoted $\mP^{(n)} \triangleleft \mP$), the asymptotic representation derived in Lemma \ref{lemma:repre_H0} remains valid under $\mathbb{H}_{1n}$ by Le Cam's First Lemma (e.g., Lemma 6.4 in \cite{van2000asymptotic}). Specifically,
$$
\widehat{T}_n = \frac{1}{n}\sum_{i=1}^n\sum_{j=1}^n f_p(S_i, S_j, \theta_0) + o_p(1),
$$
where $f_p(s, s', \theta_0) = \rho(z, \theta_0)' K^p(x, x') \rho(z', \theta_0)$.

Since the kernel $f_p$ is square-integrable with respect to $\mathbb{P} \otimes \mathbb{P}$, $\mathcal{A}$ is a Hilbert-Schmidt operator and thus compact. Moreover, the construction of $K^p$ as a projected kernel ensures that $\mathcal{A}$ is self-adjoint and positive semi-definite. By the Spectral Theorem for compact self-adjoint operators, the kernel admits the following expansion (in the $L^2$ sense):
$$
f_p(s, s', \theta_0) = \sum_{k=1}^{\infty} \lambda_k \varphi_k(s) \varphi_k(s'),
$$
where $\{\lambda_k\}_{k=1}^{\infty}$ is the sequence of non-negative eigenvalues and $\{\varphi_k\}_{k=1}^{\infty}$ are the corresponding orthonormal eigenfunctions satisfying $\E_{\mP}[\varphi_k(S)\varphi_j(S)] = \delta_{kj}$. 

Crucially, under the null hypothesis $\mathbb{H}_0$, the conditional moment restriction $\E_{\mP}[\rho(Z, \theta_0)|X] = 0$ implies that the kernel is degenerate, i.e., $\E_{\mP}[f_p(S, S', \theta_0) | S] = 0$ a. This degeneracy implies that all eigenfunctions corresponding to non-zero eigenvalues have zero mean under the null measure:
$$
\E_{\mP}[\varphi_k(S)] = 0, \quad \forall k \text{ such that } \lambda_k > 0.
$$

Substituting this spectral decomposition into the V-statistic yields the spectral representation:
$$
\begin{aligned}
\widehat{T}_n &= \sum_{k=1}^{\infty} \lambda_k \left( \frac{1}{\sqrt{n}} \sum_{i=1}^n \varphi_k(S_i) \right)^2 + o_p(1).
\end{aligned}
$$
Defining the projection component for the $k$-th eigenfunction as $\xi_{k,n} \coloneqq n^{-1/2} \sum_{i=1}^n \varphi_k(S_i)$, the statistic can be rewritten as $\widehat{T}_n = \sum_{k=1}^{\infty} \lambda_k \xi_{k,n}^2 + o_p(1)$.

Next, fix an arbitrary integer $M \ge 1$. We analyze the joint asymptotic distribution of the first $M$ projection components $\boldsymbol{\xi}_n^{(M)} = (\xi_{1,n}, \dots, \xi_{M,n})'$ and the log-likelihood ratio $\Lambda_n$ defined in \eqref{eq:loglikelihood}.

Construct the random vector $\mathbf{Y}_i \in \mathbb{R}^{M+1}$ for the $i$-th observation:
$$
\mathbf{Y}_i = \left( \varphi_1(S_i), \dots, \varphi_M(S_i), h(S_i) \right)'.
$$
Under the null hypothesis $\mathbb{H}_0$, $\{\mathbf{Y}_i\}_{i=1}^n$ constitutes a sequence of independent and identically distributed random vectors. We verify the first two moments of $\mathbf{Y}_i$ under the null measure $\mathbb{P}$. The mean vector is zero, $\E_{\mP}[\mathbf{Y}_i] = \mathbf{0}$, since $\E_{\mP}[\varphi_k(S_i)] = 0$ and $\E_{\mP}[h(S_i)] = 0$. Regarding the covariance structure, the orthonormality of the eigenfunctions implies $\Var_{\mP}(\varphi_k) = 1$ and $\Cov_{\mP}(\varphi_k, \varphi_l) = 0$ for $k \ne l$, while $\Var_{\mP}(h) = \sigma_h^2 < \infty$ holds by Assumption \ref{ass:local_dgp}. The covariance between the $k$-th eigenfunction and the score is given by:
$$
\text{Cov}_{\mP}(\varphi_k(S_i), h(S_i)) = \E_{\mP}[\varphi_k(S_i) h(S_i)] \coloneqq c_k.
$$
Let $\Sigma_{M+1}$ be the covariance matrix of $\mathbf{Y}_i$, which is structured as:
$$
\Sigma_{M+1} = \begin{pmatrix} I_M & \mathbf{c} \\ \mathbf{c}' & \sigma_h^2 \end{pmatrix}, \quad \text{where } \mathbf{c} = (c_1, \dots, c_M)'.
$$
By the Multivariate Central Limit Theorem, the scaled partial sum converges in distribution under $\mathbb{P}$:
$$
\frac{1}{\sqrt{n}} \sum_{i=1}^n \mathbf{Y}_i = \begin{pmatrix} \boldsymbol{\xi}_n^{(M)} \\ \frac{1}{\sqrt{n}} \sum_{i=1}^n h(S_i) \end{pmatrix} \xrightarrow{d} \mathcal{N}\left( \mathbf{0}, \Sigma_{M+1} \right).
$$
Recall the expansion of the log-likelihood ratio derived in \eqref{eq:loglikelihood}:
$$
\Lambda_n = \frac{1}{\sqrt{n}} \sum_{i=1}^n h(S_i) - \frac{1}{2}\sigma_h^2 + o_p(1).
$$
By Slutsky's Theorem, the joint vector of projection components and the log-likelihood ratio converges as:
$$
\begin{pmatrix} \boldsymbol{\xi}_n^{(M)} \\ \Lambda_n \end{pmatrix} \xrightarrow{d} \mathcal{N}\left( \begin{pmatrix} \mathbf{0} \\ -\frac{1}{2}\sigma_h^2 \end{pmatrix}, \begin{pmatrix} I_M & \mathbf{c} \\ \mathbf{c}' & \sigma_h^2 \end{pmatrix} \right) \quad \text{under } \mathbb{P}.
$$
Thus, by Le Cam's Third Lemma (e.g., Example 6.7 in \cite{van2000asymptotic}),  under the contiguous alternative sequence $\mathbb{P}^{(n)}$, 
$$
\boldsymbol{\xi}_n^{(M)} \xrightarrow{d} \mathcal{N}(\mathbf{c}, I_M),$$
which implies that, for any finite $M$, the components $\xi_{k,n}$ are asymptotically independent under $\mathbb{H}_{1n}$ and converge to $Z_k \sim \mathcal{N}(c_k, 1)$.

Let $T = \sum_{k=1}^{\infty} \lambda_k (W_k + c_k)^2$ denote the target limit variable. From above, for any fixed integer $M \ge 1$, the truncated statistic converges in distribution:
$$
\sum_{k=1}^M \lambda_k \xi_{k,n}^2 \xrightarrow{d} \sum_{k=1}^M \lambda_k (W_k + c_k)^2 \quad \text{under } \mathbb{P}^{(n)}.
$$
To establish the convergence of the full statistic $\widehat{T}_n(\widehat{\theta}_n)$ to $T$, by Theorem 3.2 in \cite{billingsley1999convergence}, it suffices to verify the uniform asymptotic negligibility of the tail remainder:
\begin{equation}
\label{eq:tail_cond}
\lim_{M \to \infty} \limsup_{n \to \infty} \mP^{(n)} \left( \sum_{k=M+1}^{\infty} \lambda_k \xi_{k,n}^2 > \epsilon \right) = 0, \quad \forall \epsilon > 0.
\end{equation}
By Markov's inequality, the probability in \eqref{eq:tail_cond} is bounded by $\epsilon^{-1} \sum_{k=M+1}^{\infty} \lambda_k \E_{\mP^{(n)}}[\xi_{k,n}^2]$. We analyze the limit of the expectation of the total energy. Under the contiguous alternative $\mathbb{P}^{(n)}$, a direct calculation of the $V$-statistic's expectation yields:
$$
\lim_{n \to \infty} \E_{\mP^{(n)}}[\widehat{T}_n] = \sum_{k=1}^{\infty} \lambda_k + \sum_{k=1}^{\infty} \lambda_k c_k^2.
$$
By Assumptions \ref{ass:cons} and \ref{ass:an},  $\sum_{k=1}^{\infty} \lambda_k = \E_{\mP}[\rho(Z)'K^p(X,X)\rho(Z)]<\infty$, implying that the integral operator $\mathcal{A}$ is of trace class. Furthermore, by the definition $c_k = \E_{\mP}[h(S) \varphi_k(S)]$ and the spectral decomposition $\mathcal{A} = \sum_{k=1}^{\infty} \lambda_k \varphi_k \otimes \varphi_k$, we verify that:
\begin{align*}
    \sum_{k=1}^{\infty} \lambda_k c_k^2 &= \sum_{k=1}^{\infty} \lambda_k \langle h, \varphi_k \rangle_{L^2(\mP)}^2 = \langle h, \mathcal{A}h \rangle_{L^2(\mP)}\\
    &=\int_{\mathcal{Z} \times \mathcal{X}} h(z, x) \left( \int_{\mathcal{Z} \times \mathcal{X}} f_p((z, x), (z', x'), \theta_0) h(z', x') d\mP(z', x') \right) d\mP(z, x) \\
    &=\E_{S, S'} \left[ h(S) \rho(Z, \theta_0)' K^p(X, X') \rho(Z', \theta_0) h(S') \right]\\
    &=\E_{X, X'} \left[ \delta(X)' K^p(X, X') \delta(X') \right]<\infty.
\end{align*}

Since the total expectation converges to the sum of the limits of individual components (i.e., $\lim_{n} \E_{\mP^{(n)}}[\xi_{k,n}^2] = 1 + c_k^2$), the expectation of the tail sum must vanish asymptotically:
$$
\lim_{n \to \infty} \sum_{k=M+1}^{\infty} \lambda_k \E_{\mP^{(n)}}[\xi_{k,n}^2] = \sum_{k=M+1}^{\infty} \lambda_k (1 + c_k^2).
$$
Since the series $\sum \lambda_k$ and $\sum \lambda_k c_k^2$ are convergent, the right-hand side tends to zero as $M \to \infty$. This verifies condition \eqref{eq:tail_cond}. Thus, we conclude that
$$
\widehat{T}_n \xrightarrow{d} \sum_{k=1}^{\infty} \lambda_k (W_k + c_k)^2 \quad \text{under } \mathbb{H}_{1n}.
$$
\end{proof}

\subsection{Proofs for Section \ref{sec:mb}}
\begin{lemma}
\label{lemma:unif_G}
Let $\theta^*$ be the (pseudo) true value under different hypotheses, i.e., $\theta^*=\theta_0$ under $\mathbb{H}_0$ or $\mathbb{H}_{1n}$ and $\theta^*=\theta_1$ under $\mathbb{H}_1$. Under Assumptions \ref{ass:cons} and \ref{ass:an} ( or Assumptions \ref{ass:theta1_con} and \ref{ass:theta1_an} for $\theta^*=\theta_1$), we have
$$
    T_n\coloneqq \maxi \supT \norm{ \widehat{G}_n(X_i, \theta) - G(X_i, \theta) }=o_p(1),
$$
where $\mathcal{N}^*$ is any compact neighborhood of $\theta^*$ such that $\mathcal{N}^* \subset \mathcal{N}(\theta^*)$.
\end{lemma}
\begin{proof}
    The proof employs a truncation strategy based on a uniform envelope, a grid-covering argument, and component-wise Bernstein inequalities. We begin by introducing some notations. Throughout the proof, let $S(z) \coloneqq \supT \norm{g(z,\theta)}$ be the uniform envelope. By Assumption \ref{ass:an} (ii) (or \ref{ass:theta1_an} (ii) for $\theta_1$), we have $\E\left[S(Z)^2\right] < \infty$ and $\E\left[S(Z)\right] < \infty$. Furthermore, by the boundedness of the Hessian tensor of $\rho$, i.e., $\mathbb{E}\left(\sup _{\theta \in \mathcal{N}\left(\theta^*\right)} \Fnorm{\nabla_{\theta \theta^{\prime}}^2 \rho_l\left(Z, \theta\right)}^2\right)<\infty$, for $l=1, \cdots q$. Define 
$L(z) \coloneqq\sqrt{\sum_{l=1}^q \left( \sup_{\theta \in \N^*} \norm{\nabla_{\theta\theta'}^2 \rho_l(z, \theta)}_{\mathrm{op}}\right)^2}$. By the Mean Value Theorem, $g(z,\theta)$ is Lipschitz continuous in $\mathcal{N}(\theta^*)$ with $L(z)$, i.e.
    $$
        \forall \theta_1, \theta_2 \in \N^*, \quad \norm{g(z, \theta_1) - g(z, \theta_2)} \le L(z) \norm{\theta_1 - \theta_2}.
    $$
    Furthermore, $\E\left[L(Z)^2\right]<\infty$ and $\E\left[L(Z)\right]<\infty$.

    First, we use a $\theta$-independent truncation based on the envelope $S(z)$. Let $M_n$ be a truncation sequence such that $M_n \to \infty$ as $n \to \infty$. We decompose $g(z, \theta)$ into a truncated part $\tilde g_n$ and a tail part $\tilde r_n$:
    $$
    \tilde{g}_n(z, \theta)\coloneqq g(z, \theta) \cdot \ind{S(z) \le M_n},\quad \tilde{r}_n(z, \theta)\coloneqq g(z, \theta) \cdot \ind{S(z) > M_n}.
    $$
 This induces a decomposition of the $q \times p$ matrices $\widehat{G}_n$ and $G$:
\begin{align*}
    S_n(x, \theta) &\coloneqq \frac{1}{n} \sum_{j=1}^n k(x, X_j) \tilde{g}_n(Z_j, \theta) \\
    R_n(x, \theta) &\coloneqq \frac{1}{n} \sum_{j=1}^n k(x, X_j) \tilde{r}_n(Z_j, \theta)
\end{align*}
So, $\widehat{G}_n(x, \theta) = S_n(x, \theta) + R_n(x, \theta)$, and $G(x, \theta) = \Econd{S_n(x, \theta)}{X=x} + \Econd{R_n(x, \theta)}{X=x}$.

By the triangle inequality,
\begin{align*}
    T_n &= \maxi \supT \norm{ \left(S_n(X_i, \theta) - \Econd{S_n(X_i, \theta)}{X_i}\right) + \left(R_n(X_i, \theta) - \Econd{R_n(X_i, \theta)}{X_i}\right) } \\
    &\le \underbrace{\maxi \supT \norm{ S_n(X_i, \theta) - \Econd{S_n(X_i, \theta)}{X_i} }}_{\text{I}} \\
    &\quad + \underbrace{\maxi \supT \norm{ R_n(X_i, \theta) - \Econd{R_n(X_i, \theta)}{X_i} }}_{\text{II}}
\end{align*}
We will prove that $\text{I} = o_p(1)$ and $\text{II} = o_p(1)$. 

We first show the tail part $\text{II}=o_p(1)$. By the triangle inequality, we have
$$
    \text{II} \le \maxi \supT \norm{R_n(X_i, \theta)} + \maxi \supT \norm{\Econd{R_n(X_i, \theta)}{X_i}}.
$$
For any $i\in\left\{1,\cdots n\right\}$ and $\theta\in\mathcal{N}^*$:
\begin{align*}
    \norm{R_n(X_i, \theta)} &= \norm{ \frac{1}{n} \sum_{j=1}^n k(X_i, X_j) \tilde{r}_n(Z_j, \theta) } \\
    &\le K \cdot \frac{1}{n} \sum_{j=1}^n \norm{g(Z_j, \theta) \ind{S(Z_j) > M_n}}  \\
    &\le K \cdot \frac{1}{n} \sum_{j=1}^n \norm{g(Z_j, \theta)} \ind{S(Z_j) > M_n} \\
    &\le K \cdot \left( \frac{1}{n} \sum_{j=1}^n S(Z_j) \ind{S(Z_j) > M_n} \right),
\end{align*}
where the last inequality holds by the definition of $S\left(\cdot\right)$. This bound is uniform in $i$ and $\theta$. Similarly, for any $i$ and $\theta$:
\begin{align*}
    \norm{\Econd{R_n(X_i, \theta)}{X_i}} &= \norm{\E\left[k(X_i, X_j) \tilde{r}_n(Z_j, \theta) \mid X_i\right]} \\
    &\le \E\left[ |k(X_i, X_j)| \cdot \norm{\tilde{r}_n(Z_j, \theta)} \mid X_i \right] \\
    &\le K \cdot \E\left[ S(Z_j) \ind{S(Z_j) > M_n} \right].
\end{align*}
This bound is also uniform in $i$ and $\theta$. Since $\E[S(Z)] < \infty$, by the Dominated Convergence Theorem, 
$K \cdot \E\left[ S(Z) \ind{S(Z) > M_n} \right] = o(1)$ as $M_n \to \infty$. 
For the sample part, $K \cdot \left( \frac{1}{n} \sum_{j=1}^n S(Z_j) \ind{S(Z_j) > M_n} \right)$, by Markov's inequality. For any $\epsilon > 0$,
$$
\Prb\left(K \cdot \left[ \frac{1}{n} \sum_{j=1}^n S(Z_j) \ind{S(Z_j) > M_n} \right] > \epsilon\right) \le  \frac{K}{\epsilon} \E\left[ \frac{1}{n} \sum_{j=1}^n S(Z_j) \ind{S(Z_j) > M_n} \right]= o(1).
$$    
Therefore, $\text{II} \le o_p(1) + o(1) = o_p(1)$, for any $M_n \to \infty$.

We now prove that the truncated term I is also $o_p(1)$. We use a grid-covering argument. Let $\N_n$ be a $\delta_n$-net of the compact set $\N^* \subset \R^p$. The size of this net $N_n \coloneqq |\N_n|$ satisfies $N_n \le C \delta_n^{-p}$ for some constant $C$ and $\delta_n \to 0$. Using the triangle inequality, we add and subtract the value at the nearest grid point $\theta_k$:
\begin{align*}
    \text{I} &\le \underbrace{ \maxi \max_{\theta_k \in \N_n} \norm{ S_n(X_i, \theta_k) - \Econd{S_n(X_i, \theta_k)}{X_i} } }_{\text{I}_A} \\
    &\quad + \underbrace{ \maxi \max_{\theta_k \in \N_n} \sup_{\theta: \Eucnorm{\theta - \theta_k} \le \delta_n} \norm{(S_n(X_i, \theta) - S_n(X_i, \theta_k)) - \Econd{S_n(X_i, \theta) - S_n(X_i, \theta_k)}{X_i}} }_{\text{I}_B}
\end{align*}

Analogously to the proof of Term II, by the triangle inequality, $\text{I}_B$ is bounded by the sum of its sample and expectation parts. We bound the sample part (the expectation part is identical). For any $i\in\left\{1,\cdots n\right\}$, $\theta_k\in\mathcal{N}_n$ and $\theta\in\mathcal{N}^*$:
\begin{align*}
    &\norm{ S_n(X_i, \theta) - S_n(X_i, \theta_k) } \\
    &= \norm{ \frac{1}{n} \sum_{j=1}^n k(X_i, X_j) \left[ \tilde{g}_n(Z_j, \theta) - \tilde{g}_n(Z_j, \theta_k) \right] } \\
    &\le K \cdot \frac{1}{n} \sum_{j=1}^n \norm{ \left( g(Z_j, \theta) - g(Z_j, \theta_k) \right) \cdot \ind{S(Z_j) \le M_n} } \\
    &\le K \cdot \frac{1}{n} \sum_{j=1}^n \norm{ g(Z_j, \theta) - g(Z_j, \theta_k) } \\
    &\le K \cdot \frac{1}{n} \sum_{j=1}^n L(Z_j) \norm{\theta - \theta_k}   \\
    &\le K \delta_n \left( \frac{1}{n} \sum_{j=1}^n L(Z_j) \right), 
\end{align*}
where the second-to-last inequality holds from the Lipschitz condition of $g(\cdot)$. This bound is uniform in $i, \theta, \theta_k$.
Since $\E[L(Z)] < \infty$. By LLN, $\frac{1}{n} \sum L(Z_j) = O_p(1)$, and similarly, the expectation part $\norm{\Econd{S_n(X_i, \theta) - S_n(X_i, \theta_k)}{X_i}}$ has the uniform bound of $O(\delta_n)$.
Thus, as long as we choose $\delta_n \to 0$, we have $\text{I}_B = o_p(1)$.

For Term $\text{I}_A$, we apply Bernstein's inequality and finally use a union bound. Since $S_n(\cdot)$ is a $q\times p$ matrix, we bound the Frobenius norm by the $L_1$ norm of the components. Let $S_n^{lm}(\cdot)$ be the $(l,m)$-th component of the matrix $S_n(\cdot)$.
\begin{align*}
    \text{I}_A &\le \maxi \max_k \sum_{l=1}^q \sum_{m=1}^p \left| S_n^{lm}(X_i, \theta_k) - \Econd{S_n^{lm}(X_i, \theta_k)}{X_i} \right| \\
    &\le \sum_{l=1}^q \sum_{m=1}^p \left( \maxi \max_k \left| S_n^{lm}(X_i, \theta_k) - \Econd{S_n^{lm}(X_i, \theta_k)}{X_i} \right| \right),
\end{align*}
which is a summation over a fixed number ($q \times p$) of terms. To show that $\text{I}_A=o_p(1)$, it is sufficient to show that $\text{I}_A^{(l,m)} \coloneqq\maxi \max_k \left| S_n^{lm}(X_i, \theta_k) - \Econd{S_n^{lm}(X_i, \theta_k)}{X_i} \right|= o_p(1)$ for every component $(l,m)$.

Fix $(l,m)$. Let $W_{j, (i,k,l,m)} \coloneqq k(X_i, X_j) \left[\tilde{g}_n(Z_j, \theta_k)\right]_{lm}$. Given $X_i$, the scalar variables $\{W_{j, (i,k,l,m)}\}_{j=1}^n$ are independent.
We apply Bernstein's inequality. Let $Y_j = W_j - \Econd{W_j}{X_i}$. First, it is easy to verify that $Y_j$ is bounded, since
\begin{align*}
    |W_{j, (i,k,l,m)}| &\le |k(X_i, X_j)| \cdot \left| \left[\tilde{g}_n(Z_j, \theta_k)\right]_{lm} \right| \\
    &\le K \cdot \norm{\tilde{g}_n(Z_j, \theta_k)}  \\
    &= K \cdot \norm{g(Z_j, \theta_k) \ind{S(Z_j) \le M_n}} \\
    &\le K M_n.
\end{align*}
So, $|Y_j| \le |W_j| + |\E[W_j \mid X_i]| \le 2 K M_n$. Let $B_n \coloneqq 2 K M_n$. Furthermore, the conditional variance $V_n \coloneqq \sum_{j=1}^n \operatorname{Var}(Y_j \mid X_i) \le \sum_{j=1}^n \Econd{W_j^2}{X_i} \le  n (K M_n)^2$. For $\varepsilon>0$, by Bernstein's Inequality:
\begin{align*}
    \Prb\left(\left| \frac{1}{n} \sum_{j=1}^n Y_j \right| > \epsilon \mid X_i\right) &\le 2 \exp\left( - \frac{(n\epsilon)^2}{2V_n + \frac{2}{3} B_n (n\epsilon)} \right) \\
    &\le 2 \exp\left( - \frac{n \epsilon^2}{2 (K M_n)^2 + \frac{4}{3} K M_n \epsilon} \right).
\end{align*}
Thus, by the tower property, $\Prb\left(\left| \frac{1}{n} \sum_{j=1}^n Y_j \right| > \epsilon\right)\le 2 \exp\left( - \frac{n \epsilon^2}{2 (K M_n)^2 + \frac{4}{3} K M_n \epsilon} \right)$. The total number of terms in the $\max_{i,k}$ is $n \times N_n$. By a union bound over $i$ and $k$:
\begin{align*}
    \Prb\left( \text{I}_A^{(l,m)} > \epsilon \right) 
    &\le 2n \cdot N_n \cdot  \exp\left( - \frac{n \epsilon^2}{2 (K M_n)^2 + \frac{4}{3} K M_n \epsilon} \right).
\end{align*}
Finally, we choose our sequences $M_n$ and $\delta_n$. Let $M_n = \log n$, $\delta_n = 1/n$. Then, $N_n = O(n^p)$. For any given $\varepsilon>0$, $\Prb\left( \text{I}_A^{(l,m)} > \epsilon \right)\le O(n^{p+1}) \cdot \exp\left( - \frac{n \epsilon^2}{O((\log n)^2)} \right)\to 0$. Thus, $\text{I}_A^{(l,m)}=o_p(1)$ for any given $(l,m)$. Thus, $\text{I}_A$ is $o_p(1)$. This completes the proof.
\end{proof}

\begin{lemma}
\label{lemma:unif_delta}
Let $\theta^*$ be the (pseudo) true value under different hypotheses, i.e., $\theta^*=\theta_0$ under $\mathbb{H}_0$ or $\mathbb{H}_{1n}$ and $\theta^*=\theta_1$ under $\mathbb{H}_1$. Under Assumptions \ref{ass:cons} and \ref{ass:an} ( or Assumptions \ref{ass:theta1_con} and \ref{ass:theta1_an} for $\theta^*=\theta_1$), for some compact neighborhood $\mathcal{N}^* \subset \mathcal{N}(\theta^*)$ of $\theta^*$, we have
$$
    \supT \norm{ \widehat{\Delta}_n(\theta)^{-1} - \Delta(\theta)^{-1} }=o_p(1),
$$
\end{lemma}
\begin{proof}
    Let $S_i=\left\{X_i, Z_i\right\}$, $g\left(S_i, S_j, \theta\right)=g\left(Z_i, \theta\right)^{\prime} K\left(X_i, X_j\right) g\left(Z_j, \theta\right)$. By Assumption \ref{ass:an}(ii)( or Assumption \ref{ass:theta1_an} (ii) for $\theta^*=\theta_1$), $g\left(S_i, S_j, \theta\right)$ is continuous at each $\theta \in \mathcal{N}(\theta^*)$ with probability one. Furthermore,
    \begin{align*}
& \mathbb{E}\left(\sup _{\theta \in \Nb}\left\|g\left(S_i, S_j, \theta\right)\right\|\right) \leqslant \sup _{x \in \mathcal{X}} k(x, x) \cdot\left[\mathbb{E}\left(\sup _{\theta \in \Nb}\|g(z, \theta)\|\right)\right]^2<\infty,\\
& \mathbb{E}\left(\sup _{\theta \in \Nb}\left\|g\left(S_i, S_i, \theta\right)\right\|\right) \leqslant \sup _{x \in \mathcal{X}} k(x, x) \cdot\mathbb{E}\left(\sup _{\theta \in \Nb}\|g(z, \theta)\|^2\right)<\infty .
\end{align*}
    By Lemma 8.5 of \cite{newey1994large}, $\Delta(\theta)$ is continuous in $\theta\in\Nb$, and
    \begin{equation}
        \label{eq:delta_unif}
        \sup_{\theta\in\Nb}\norm{\widehat{\Delta}(\theta)-\Delta(\theta)}=o_p(1).
    \end{equation}
    By Assumption \ref{ass:an}(iii) (or \ref{ass:theta1_an}(iii)), $\Delta(\theta^*)$ is non-singular. By the continuity of $\Delta(\theta)$ (and thus $\det(\Delta(\theta))$) on the open set $\Nb$, there exists a compact sub-neighborhood $\mathcal{N}^* \subset \Nb$ of $\theta^*$ on which $\Delta(\theta)$ is uniformly non-singular, i.e.:
    \begin{equation}
        \label{eq:inv_bound_concise}
        \sup_{\theta \in \mathcal{N}^*} \norm{\Delta(\theta)^{-1}} = O(1).
    \end{equation}
    Now, consider $\widehat{\Delta}(\theta)$ on $\mathcal{N}^*$. Combining \eqref{eq:delta_unif} and \eqref{eq:inv_bound_concise}, it follows by standard arguments on matrix inversion that 
    \begin{equation} \label{eq:inv_hat_bound}
        \sup_{\theta \in \mathcal{N}^*} \norm{\widehat{\Delta}(\theta)^{-1}} = O_p(1).
    \end{equation}
    Using the identity $A^{-1} - B^{-1} = A^{-1}(B - A)B^{-1}$, 
    \begin{align*}
        \sup_{\theta \in \mathcal{N}^*} \norm{ \widehat{\Delta}(\theta)^{-1} - \Delta(\theta)^{-1} } 
        &\le \sup_{\theta \in \mathcal{N}^*} \norm{ \widehat{\Delta}(\theta)^{-1} (\Delta(\theta) - \widehat{\Delta}(\theta)) \Delta(\theta)^{-1} } \\
        &\le \left( \sup_{\theta \in \mathcal{N}^*} \norm{\widehat{\Delta}(\theta)^{-1}} \right) \cdot \left( \sup_{\theta \in \mathcal{N}^*} \norm{\widehat{\Delta}(\theta) - \Delta(\theta)} \right) \cdot \left( \sup_{\theta \in \mathcal{N}^*} \norm{\Delta(\theta)^{-1}} \right) \\
        &= O_p(1) \cdot o_p(1) \cdot O(1) = o_p(1).
    \end{align*}
    This completes the proof.
\end{proof}

\begin{lemma}
\label{lemma:unif_Kp}
Let $\theta^*$ be the (pseudo) true value under different hypotheses, i.e., $\theta^*=\theta_0$ under $\mathbb{H}_0$ or $\mathbb{H}_{1n}$ and $\theta^*=\theta_1$ under $\mathbb{H}_1$. Under Assumptions \ref{ass:cons} and \ref{ass:an} ( or Assumptions \ref{ass:theta1_con} and \ref{ass:theta1_an} for $\theta^*=\theta_1$), 
$$
\max_{i,j}\norm{\widehat{K}_n^p(X_i,X_j,\widehat{\theta}_n)-K^p(X_i,X_j,\theta^*)}=o_p(1).
$$
\end{lemma}

\begin{proof}
    Applying the triangle inequality, we decompose the difference of interest into two components:
    \begin{align*}
        &\quad \max_{1 \le i,j \le n} \norm{\widehat{K}_n^p(X_i,X_j,\widehat{\theta}_n)-K^p(X_i,X_j,\theta^*)}\\
        &\le \max_{i,j} \norm{\widehat{G}_n(X_i,\widehat{\theta}_n)\widehat{\Delta}_n(\widehat{\theta}_n)^{-1}\widehat{G}_n(X_j,\widehat{\theta}_n)' - G(X_i,\widehat{\theta}_n)\Delta(\widehat{\theta}_n)^{-1}G(X_j,\widehat{\theta}_n)'} \\
        &\quad + \max_{i,j} \norm{G(X_i,\widehat{\theta}_n)\Delta(\widehat{\theta}_n)^{-1}G(X_j,\widehat{\theta}_n)' - G(X_i,{\theta}^*)\Delta({\theta}^*)^{-1}G(X_j,{\theta}^*)'},
    \end{align*}
    denoted as $\mathrm{I}$ and $\mathrm{II}$ respectively. 
    
    To establish the convergence of term $\mathrm{I}$, observe that consistency $\widehat{\theta}_n \xrightarrow{p} \theta^*$ implies that $\widehat{\theta}_n$ lies within a compact neighborhood $\mathcal{N}^*$ of $\theta^*$ with probability approaching one. Therefore, it suffices to bound the supremum of the difference over $\theta \in \mathcal{N}^*$. Adding and subtracting cross-terms yields
    \begin{align*}
        \mathrm{I} &\le \max_{i,j}\sup_{\theta\in\mathcal{N}^*}\norm{(\widehat{G}_n(X_i,\theta) - G(X_i,\theta)) \widehat{\Delta}_n(\theta)^{-1} \widehat{G}_n(X_j,\theta)'} \\
        &\quad + \max_{i,j}\sup_{\theta\in\mathcal{N}^*}\norm{G(X_i,\theta) (\widehat{\Delta}_n(\theta)^{-1} - \Delta(\theta)^{-1}) \widehat{G}_n(X_j,\theta)'} \\
        &\quad + \max_{i,j}\sup_{\theta\in\mathcal{N}^*}\norm{G(X_i,\theta) \Delta(\theta)^{-1} (\widehat{G}_n(X_j,\theta) - G(X_j,\theta))'}.
    \end{align*}
    By Lemmas \ref{lemma:unif_G} and \ref{lemma:unif_delta}, uniformly in $i,\theta$, we have $\norm{G(X_i,\theta)} = O(1)$, $\norm{\widehat{G}_n(X_i,\theta)-G(X_i,\theta)} = o_p(1)$, $\norm{\widehat G_n(X_i,\theta)} = O_p(1)$, $\norm{\Delta(\theta)^{-1}} = O(1)$, $\norm{\widehat\Delta_n(\theta)^{-1}-\Delta(\theta)^{-1}} = o_p(1)$,  $\norm{\widehat\Delta_n(\theta)^{-1}} = O_p(1)$. Consequently, each term in the expansion above is the product of an $o_p(1)$ factor and $O_p(1)$ factors, proving that $\mathrm{I} = o_p(1)$.

    For term $\mathrm{II}$, define $\Psi_{ij}(\theta) \coloneqq G(X_i,\theta)\Delta(\theta)^{-1}G(X_j,\theta)'$. By the Mean Value inequality, $\mathrm{II}$ is bounded by $\|\widehat{\theta}_n - \theta^*\| \cdot \max_{i,j} \sup_{\theta \in \mathcal{N}^*} \|\nabla_{\theta}\Psi_{ij}(\theta)\|$.
    
Invoking the Dominated Convergence Theorem under Assumptions \ref{ass:cons} and \ref{ass:an} (or Assumptions \ref{ass:theta1_con} and \ref{ass:an} for $\theta^*=\theta_1$), we establish the uniform bounds:
\begin{align*}
    \sup_{x \in \mathcal{X}} \sup_{\theta \in \mathcal{N}^*} \|G(x, \theta)\| 
    &\le \sup_{x, \tilde{x}} \|K(x, \tilde{x})\| \cdot \mathbb{E}\left[\sup_{\theta \in \mathcal{N}^*} \|g(Z, \theta)\|\right] \le C_1 < \infty, \\
    \sup_{x \in \mathcal{X}} \sup_{\theta \in \mathcal{N}^*} \|\nabla_\theta G(x, \theta)\| 
    &\le \sup_{x, \tilde{x}} \|K(x, \tilde{x})\| \cdot \mathbb{E}\left[\sup_{\theta \in \mathcal{N}^*} \|\nabla_\theta g(Z, \theta)\|\right] \le C_2 < \infty.
\end{align*}
Furthermore, $\sup_{\theta \in \mathcal{N}^*} \|\Delta(\theta)^{-1}\| \le C_\Delta < \infty$ and $\sup_{\theta \in \mathcal{N}^*} \|\nabla_\theta \Delta(\theta)\| \le C_{\nabla\Delta} < \infty$. By the matrix product rule, the derivative of $\Psi_{ij}(\theta)$ is:
\begin{equation*}
    \nabla_\theta \Psi_{ij}(\theta) = (\nabla_\theta G_i) \Delta^{-1} G_j' - G_i \Delta^{-1} (\nabla_\theta \Delta) \Delta^{-1} G_j' + G_i \Delta^{-1} (\nabla_\theta G_j)',
\end{equation*}
where dependence on $\theta$ and $X$ is suppressed for brevity. Applying the triangle inequality and the sub-multiplicative property of the matrix norm implies:
\begin{align*}
    \|\nabla_\theta \Psi_{ij}(\theta)\| &\le \|\nabla_\theta G_i\| \|\Delta^{-1}\| \|G_j\| + \|G_i\| \|\Delta^{-1}\|^2 \|\nabla_\theta \Delta\| \|G_j\| + \|G_i\| \|\Delta^{-1}\| \|\nabla_\theta G_j\| \\
    &\le C_2 C_\Delta C_1 + C_1^2 C_\Delta^2 C_{\nabla\Delta} + C_1 C_\Delta C_2.
\end{align*}
Since the upper bound depends only on finite constants which are independent of indices $i,j$, we conclude that $\max_{1 \le i,j \le n} \sup_{\theta \in \mathcal{N}^*} \left\| \nabla_\theta \Psi_{ij}(\theta) \right\| = O(1)$. Since $\|\widehat{\theta}_n - \theta^*\| = O_p(n^{-1/2})$, it follows that $\mathrm{II} = O_p(n^{-1/2}) = o_p(1)$. Combining the results for $\mathrm{I}$ and $\mathrm{II}$ completes the proof.
\end{proof}

\paragraph{Proof of Lemma \ref{lemma:mb_asy}}
\begin{proof}
    A second-order Taylor expansion of $\rho(Z_i, \widehat{\theta}_n)$ around $\theta^*$ yields:
    \begin{equation}
        \rho(Z_i, \widehat{\theta}_n) = \rho(Z_i, \theta^*) + g(Z_i, \theta^*) (\widehat{\theta}_n - \theta^*) + \frac{1}{2} \mathcal{R}_i(\bar{\theta}),
        \label{eq:rho_taylor_exp}
    \end{equation}
    where the remainder term $\mathcal{R}_i(\bar{\theta})$ is a $q \times 1$ vector with its $l$-th component given by:
    \begin{equation*}
        \mathcal{R}_{il}(\bar{\theta}) = (\widehat{\theta}_n - \theta^*)' \nabla_{\theta\theta}^2 \rho_l(Z_i, \bar{\theta}_{il}) (\widehat{\theta}_n - \theta^*).
    \end{equation*}
    Here, $\bar{\theta}_{il}$ denotes an intermediate point on the line segment joining $\widehat{\theta}_n$ and $\theta^*$, which may vary across components $l$.
    Additionally, Lemma \ref{lemma:unif_Kp} establishes the uniform convergence of the kernel estimator:
    \begin{equation}
        \label{eq:Kp_asy}
        \max_{1 \le i,j \le n} \| \widehat{K}_n^p(X_i,X_j,\widehat{\theta}_n) - K^p(X_i,X_j,\theta^*) \| = o_p(1).
    \end{equation}

    To analyze the asymptotic difference between $\widehat{T}^{*}_n$ and ${T}^{*}_n$, we introduce the following compact notation. For $u \in \{1, 2, 3\}$, define the vector expansion terms $A_{u,i}$:
    $$
    A_{1,i} \coloneqq \rho(Z_i, \theta^*),\quad A_{2,i} \coloneqq g(Z_i, \theta^*) (\widehat{\theta}_n - \theta^*),\quad A_{3,i} \coloneqq \frac{1}{2} \mathcal{R}_i(\bar{\theta}).
    $$
    Similarly, for $v \in \{1, 2\}$, define the kernel matrix components $M_{v,ij}$:
    $$
    M_{1,ij} \coloneqq K^p(X_i, X_j, \theta^*),\quad M_{2,ij} \coloneqq \widehat{K}_n^p(X_i, X_j, \widehat{\theta}_n) - K^p(X_i, X_j, \theta^*).
    $$
    Substituting these definitions into $\widehat{T}_n^*$, we obtain the decomposition:
    \begin{equation}
        \label{eq:expansion_18_terms}
        \widehat{T}_n^* = \frac{1}{n} \sum_{i=1}^n \sum_{j=1}^n V_i V_j \left( \sum_{u=1}^3 A_{u,i} \right)^\prime \left( \sum_{v=1}^2 M_{v,ij} \right) \left( \sum_{w=1}^3 A_{w,j} \right) = \sum_{u,v,w} Q_{uvw},
    \end{equation}
    where $Q_{uvw} \coloneqq \frac{1}{n} \sum_{i,j} V_i V_j A_{u,i}^\prime M_{v,ij} A_{w,j}$ for indices $u,w \in \{1,2,3\}$ and $v \in \{1,2\}$. We analyze representative terms from this sum:

    \begin{enumerate}
        \item Leading Term ($Q_{111}$): By definition, $Q_{111} = T_n^*$.

        \item Linear Estimation Error ($Q_{112}$): This term captures the interaction between the score function and the true kernel:
        \begin{equation*}
            Q_{112} = \left[\frac{1}{n}\sum_{i,j}\rho(Z_i,\theta^*)^\prime K^p(\theta^*)g(Z_j,\theta^*)V_iV_j\right] (\widehat{\theta}_n-\theta^*).
        \end{equation*}
        The bracketed term is a $n$-scaled degenerate $V$-statistic (due to centered multipliers $\mathbb{E}[V_i]=0$) and is thus $O_p(1)$. Since $(\widehat{\theta}_n-\theta^*) = O_p(n^{-1/2})$, it follows that $Q_{112} = O_p(1) \cdot O_p(n^{-1/2}) = o_p(1)$. 

        \item Second-Order Remainder ($Q_{113}$): This term involves the Hessian via $\mathcal{R}_i(\bar{\theta})$. Rearranging terms to isolate the convergence rate yields:
        \begin{equation}
            Q_{113} = \frac{1}{2} \sqrt{n}(\widehat{\theta}_n - \theta^*)' \left[ \frac{1}{n^2} \sum_{i,j=1}^n \mathbf{\Psi}_n(Z_i, Z_j, \bar{\theta}) \right] \sqrt{n}(\widehat{\theta}_n - \theta^*),
            \label{eq:Q113_form}
        \end{equation}
        where $\mathbf{\Psi}_n(Z_i, Z_j, \theta) \coloneqq V_i V_j \sum_{l=1}^q \left[ \rho(Z_i, \theta^*)' K^p(\theta^*) \right]_l \nabla_{\theta\theta}^2 \rho_l(Z_j, \theta)$.
        Since $\|\sqrt{n}(\widehat{\theta}_n - \theta^*)\| = O_p(1)$, the asymptotic behavior is determined by the bracketed term. Observe that $\frac{1}{n^2} \sum_{i,j} \mathbf{\Psi}_n(\cdot)$ constitutes a $V$-statistic with a zero-mean kernel function. Applying the Uniform Weak Law of Large Numbers (UWLLN) for $V$-statistics (e.g., Lemma 8.5 of \cite{newey1994large}), we obtain:
        \begin{equation*}
            \sup_{\theta \in \mathcal{N}^*} \left\| \frac{1}{n^2} \sum_{i,j=1}^n \mathbf{\Psi}_n(Z_i, Z_j, \theta) \right\| = o_p(1).
        \end{equation*}
        As $\bar{\theta} \xrightarrow{p} \theta^*$, the continuous mapping theorem implies the bracketed term is $o_p(1)$, and thus $Q_{113} = o_p(1)$.

        \item Kernel Estimation Error ($Q_{121}$): This term accounts for the error in $\widehat{K}_n^p$:
        \begin{equation*}
            Q_{121} = \frac{1}{n}\sum_{i,j} V_i V_j \rho(Z_i,\theta^*)^\prime M_{2,ij} \rho(Z_j,\theta^*).
        \end{equation*}
        Using the uniform bound $\max_{i,j} \|M_{2,ij}\| = o_p(1)$ from Lemma \ref{lemma:unif_Kp}, combined with the $O_p(1)$ boundedness of the degenerate $V$-statistic formed by the multipliers and $\rho$, we have $Q_{121} = o_p(1)$.
    \end{enumerate}

    The remaining 14 terms are of equal or smaller stochastic order and vanish asymptotically. We categorize them based on their dependence on the intermediate value $\bar{\theta}$. For terms involving the second-order remainder $A_3$ (where $u=3$ or $w=3$), the quadratic factor $(\widehat{\theta}_n - \theta^*)$ allows for the extraction of an $n^{-1}$ scaling factor. This effectively results in an $n^{-2}$-scaled sum, activating the UWLLN for $V$-statistics to control the dependence on $\bar{\theta}$ and rendering these terms $o_p(1)$. Conversely, terms independent of $\bar{\theta}$ (involving only $A_1, A_2$, and $M_{1,2}$) do not require uniform convergence laws; their negligibility follows directly from the boundedness of degenerate $V$-statistics (scaled by $n$), combined with the $\sqrt{n}$-consistency of $\widehat{\theta}_n$ and the uniform convergence of $\widehat{K}_n^p$. 
    
    Consequently, $\widehat{T}_n^* = T_n^* + o_p(1)$.
\end{proof}

\begin{lemma}
\label{lemma:max_row_bound}
    Suppose Assumptions \ref{ass:cons} and
     \ref{ass:an} hold. Let $f_p(s,s',\theta_0)=\rho(z,\theta_0)'K^p(x,x')\rho(z',\theta_0)$ as defined in Theorem \ref{thm:H0}. Then, we have:
    \begin{equation*}
        M_n \coloneqq \frac{1}{\sqrt{n}} \max_{1 \le i \le n} \left( \frac{1}{n} \sum_{j=1}^n \left(f_p(S_i,S_j,\theta_0) \right)^2 \right)^{1/2} = o_p(1).
    \end{equation*}
\end{lemma}

\begin{proof}
    We first establish the uniform boundedness of the projected kernel $K^p$. By Assumptions \ref{ass:cons}, \ref{ass:an} and Jensen's inequality, observe that $\sup_{x \in \mathcal{X}} \|G(x, \theta_0)\| \le C_K \mathbb{E}[\|g(Z, \theta_0)\|]<\infty$. Using the definition of $K^p$, the triangle inequality, and the sub-multiplicative property of matrix norms, we obtain
    $\sup_{x,y \in \mathcal{X}} |K^p(x, y)| \le C_K + C_G^2 \|\Delta(\theta_0)^{-1}\|$. Since $\Delta(\theta_0)$ is non-singular, $\|\Delta(\theta_0)^{-1}\| < \infty$. Thus, there exists a finite constant $C_{K^p}$ such that $\sup_{x,y} |K^p(x, y)| \le C_{K^p}$. 
    
    By the Cauchy--Schwarz inequality, we have $|f_p(S_i,S_j,\theta_0)| \le C_{K^p} \|\rho(Z_i,\theta_0)\| \|\rho(Z_j,\theta_0)\|$. Substituting this bound into the definition of $M_n$ yields
    \begin{equation*}
        M_n \le \frac{1}{\sqrt{n}} \max_{1 \le i \le n} \left( \frac{1}{n} \sum_{j=1}^n C_{K^p}^2 \|\rho(Z_i)\|^2 \|\rho(Z_j)\|^2 \right)^{1/2} = C_{K^p} \left( \frac{\max_{1 \le i \le n} \|\rho(Z_i)\|}{\sqrt{n}} \right) \left( \frac{1}{n} \sum_{j=1}^n \|\rho(Z_j)\|^2 \right)^{1/2}.
    \end{equation*}
     Weak Law of Large Numbers imply that $n^{-1} \sum_{j=1}^n \|\rho(Z_j)\|^2 \xrightarrow{p} \mathbb{E}[\|\rho(Z)\|^2]$, so the square root of this sum is $O_p(1)$. Furthermore, $\max_{1 \le i \le n} \|\rho(Z_i)\| = o_p(n^{1/2})$. Combining these, we conclude that $M_n \le C_{K^p} \cdot o_p(1) \cdot O_p(1) = o_p(1)$.
\end{proof}

\paragraph{Proof of Theorem \ref{thm:mb}}
\begin{proof}
    First, by Lemma \ref{lemma:mb_asy}, we have
    $$
    \widehat{T}_n^*=T_n^*+o_p(1),
    $$
    where $T_n^*=\frac{1}{n} \sum_{i=1}^n \sum_{j=1}^n V_i V_j \rho(Z_i, \theta^*)^{\prime} K^p(X_i, X_j, \theta^*) \rho(Z_j, \theta^*)$ is defined in the proof of Lemma \ref{lemma:mb_asy}, and $\theta^*$ is the (pseudo) true value under different hypotheses, i.e. $\theta^*=\theta_0$ under $\mathbb{H}_0$ or $\mathbb{H}_{1n}$ and $\theta^*=\theta_1$ under $\mathbb{H}_1$. Thus, by Slutsky's Theorem, it is sufficient to show that conditional on the original samples, $T_n^*$ is stochastically bounded and furthermore, under $\mathbb{H}_0$ or $\mathbb{H}_{1n}$ (i.e., when $\theta^*=\theta_0$), $T_n^*(\theta_0)\xrightarrow{d,*}\sum_{k=1}^{\infty}\lambda_kW_k^2$, as defined in Theorem \ref{thm:H0}.

Under $\mathbb{H}_0$ or $\mathbb{H}_{1n}$, we have $\theta^* = \theta_0$. Let $f_p(s, s^{\prime}, \theta_0)$ be defined as in Theorem \ref{thm:H0}, specifically $f_p(s, s^{\prime}, \theta_0)=\rho(z, \theta_0)^{\prime}K^p(x, x^{\prime})\rho(z^{\prime}, \theta_0)$. Consider the $n \times n$ symmetric Gram matrix $\mathbf{H}_n$ with entries $H_{ij}=f_p(S_i, S_j, \theta_0)$. By the spectral theorem, we decompose $\mathbf{H}_n$ as
\begin{equation*}
    \mathbf{H}_n = \sum_{k=1}^n \widehat{\gamma}_{k,n} \mathbf{u}_k \mathbf{u}_k^{\prime},
\end{equation*}
where $\widehat{\gamma}_{1,n} \ge \widehat{\gamma}_{2,n} \ge \dots$ are the real eigenvalues and $\mathbf{u}_1, \dots, \mathbf{u}_n$ are the orthonormal eigenvectors ($\mathbf{u}_k \in \mathbb{R}^n, \|\mathbf{u}_k\|^2=1$).

Define the normalized eigenvalues by $\widehat{\lambda}_{k,n} = \widehat{\gamma}_{k,n} / n$. These correspond to the non-zero eigenvalues of the empirical operator $\widehat{\mathcal{A}}_n$, which is the finite-sample counterpart to $\mathcal{A}$, defined by
\begin{equation}
\label{eq_An}
    \widehat{\mathcal{A}}_n \psi(s) = \int f_p(s, s^{\prime}, \theta_0) \psi(s^{\prime}) \, d \mathbb{P}_{n}(s^{\prime}) = \frac{1}{n} \sum_{i=1}^n f_p(s, S_i, \theta_0) \psi(S_i).
\end{equation}

It is a standard duality result that the non-zero eigenvalues of the empirical operator $\widehat{\mathcal{A}}_n$ are exactly $\widehat{\lambda}_{k,n}$. Substituting the decomposition into the expression for $T_n^*$ yields:
\begin{equation}
    T_n^* = \frac{1}{n} \boldsymbol{V}^{\prime} \mathbf{H}_n \boldsymbol{V} = \sum_{k=1}^n \widehat{\lambda}_{k,n} (\mathbf{u}_k^{\prime} \boldsymbol{V})^2 = \sum_{k=1}^n \widehat{\lambda}_{k,n} Z_{k,n}^2,
    \label{eq:spec_rep}
\end{equation}
where $Z_{k,n} \coloneqq \mathbf{u}_k^{\prime} \boldsymbol{V} = \sum_{i=1}^n u_{k,i} V_i$. 
Conditional on the original sample $\mathcal{S}_n$, the variables $\{Z_{k,n}\}_{k=1}^n$ are uncorrelated with zero mean and unit variance, since $\mathbb{E}^*[\boldsymbol{V}\boldsymbol{V}^{\prime}] = \mathbf{I}_n$ implies $\text{Cov}^*(Z_{k,n}, Z_{l,n}) = \mathbf{u}_k^{\prime} \mathbf{I}_n \mathbf{u}_l = \delta_{kl}$.

Let $\phi_{n}^*(t) = \mathbb{E}^*[\exp(i t T_n^*)]$ denote the conditional characteristic function of $T_n^*$. Let $Y = \sum_{k=1}^\infty \lambda_k W_k^2$ be the target random variable with characteristic function $\phi_Y(t)$, where $W_k \stackrel{\text{i.i.d.}}{\sim} \mathcal{N}(0,1)$.
Fix a truncation integer $M \ge 1$. We decompose the statistic and the target variable as:
$$T_n^* = T_{n,M}^* + R_{n,M}^*, \quad \text{where } T_{n,M}^* = \sum_{k=1}^M \widehat{\lambda}_{k,n} Z_{k,n}^2,$$
$$Y = Y_M + Y_{tail, M}, \quad \text{where } Y_M = \sum_{k=1}^M \lambda_k W_k^2.$$
Using the inequality $|e^{ix} - e^{iy}| \le |x-y|$, we bound the distance between the characteristic functions for any fixed $t$:
\begin{align}
    |\phi_{n}^*(t) - \phi_Y(t)| &\le |\mathbb{E}^*[e^{i t T_n^*}] - \mathbb{E}^*[e^{i t T_{n,M}^*}]| + |\mathbb{E}^*[e^{i t T_{n,M}^*}] - \mathbb{E}[e^{i t Y_M}]| + |\mathbb{E}[e^{i t Y_M}] - \mathbb{E}[e^{i t Y}]| \notag \\
    &\le \mathbb{E}^*|e^{i t T_n^*} - e^{i t T_{n,M}^*}| + |\phi_{T_{n,M}^*}(t) - \phi_{Y_M}(t)| + \mathbb{E}|e^{i t Y_M} - e^{i t Y}| \notag \\
    &\le |t| \underbrace{\mathbb{E}^*|R_{n,M}^*|}_{\Delta_1} + \underbrace{|\phi_{T_{n,M}^*}(t) - \phi_{Y_M}(t)|}_{\Delta_2} + |t| \underbrace{\mathbb{E}|Y_{tail, M}|}_{\Delta_3}. \label{eq:triangle}
\end{align}
We show that for any $\varepsilon > 0$ and fixed $t$, the RHS converges to zero in probability by first choosing $M$ large enough and then letting $n \to \infty$.

First, we analyze the empirical tail term $\Delta_1$. Since the projected kernel $K^p$ is positive semi-definite (PSD), the kernel matrix $\mathbf{H}_n$ is PSD, implying $\widehat{\lambda}_{k,n} \ge 0$. Consequently, the residual sum is non-negative, $R_{n,M}^* \ge 0$. Given that the multipliers satisfy $\mathbb{E}^*[Z_{k,n}^2] = 1$, the conditional expectation is:
$$
\Delta_1 = \mathbb{E}^* \left[ \sum_{k=M+1}^n \widehat{\lambda}_{k,n} Z_{k,n}^2 \right] = \sum_{k=M+1}^n \widehat{\lambda}_{k,n}.
$$
Using the trace identity, $\sum_{k=1}^n \widehat{\lambda}_{k,n} = \frac{1}{n} \text{tr}(\mathbf{H}_n)$. By the Law of Large Numbers (LLN), the empirical trace converges in probability to the population trace:
$$ 
\sum_{k=1}^n \widehat{\lambda}_{k,n} = \frac{1}{n} \sum_{i=1}^n \rho(Z_i)^{\prime} K^p(X_i, X_i) \rho(Z_i) \xrightarrow{p} \mathbb{E}\left[\rho(Z)^{\prime} K^p(X, X) \rho(Z)\right] = \sum_{k=1}^\infty \lambda_k. 
$$
The finiteness of this limit is guaranteed by Assumptions \ref{ass:cons} and \ref{ass:an}, which imply the operator is trace class.
Let $\mathcal{H}_{f_p}$ be the RKHS associated with the kernel $f_p(s, s', \theta_0) = \rho(z, \theta_0)^{\prime}K^p(x, x')\rho(z', \theta_0)$, which is valid by the Schur product theorem. We define the random element $\chi_i = f_p(\cdot, S_i, \theta_0) \in \mathcal{H}_{f_p}$.

Using the reproducing property $\langle f_p(\cdot, s), h \rangle_{\mathcal{H}_{f_p}} = h(s)$, we observe that the population operator $\mathcal{A}$ defined in \eqref{eq_A} is equivalent to the expected rank-one tensor product:
\begin{equation}
    \mathcal{A} h = \int f_p(\cdot, s) h(s) d\mP(s) = \mathbb{E}\left[ \chi_i \langle \chi_i, h \rangle_{\mathcal{H}_{f_p}} \right] = \mathbb{E}[\chi_i \otimes \chi_i] h.
\end{equation}
Similarly, the empirical counterpart $\widehat{\mathcal{A}}_n$ defined in \eqref{eq_An} admits the equivalent tensor product representation $\widehat{\mathcal{A}}_n = \frac{1}{n} \sum_{i=1}^n \chi_i \otimes \chi_i$. Since $\mathbb{E}[\|\chi_1\|_{f_p}^2]=\mathbb{E}[f_p(S_1, S_1, \theta_0)] < \infty$, by Theorem 8.1.2 of \cite{hsing2015theoretical}, $\|\widehat{\mathcal{A}}_n - \mathcal{A}\|_{HS} \xrightarrow{a.s.} 0$, and consequently $\|\widehat{\mathcal{A}}_n - \mathcal{A}\|_{op} \xrightarrow{a.s.} 0$. By Weyl's Inequality, the eigenvalues converge uniformly, so for any fixed $M$:
$$ 
\sum_{k=1}^M \widehat{\lambda}_{k,n} \xrightarrow{p} \sum_{k=1}^M \lambda_k. 
$$
Combining these results yields the convergence of the tail sum in probability:
$$ 
\Delta_1 = \sum_{k=1}^n \widehat{\lambda}_{k,n} - \sum_{k=1}^M \widehat{\lambda}_{k,n} \xrightarrow{p} \sum_{k=1}^\infty \lambda_k - \sum_{k=1}^M \lambda_k = \sum_{k=M+1}^\infty \lambda_k. 
$$
Since $\sum \lambda_k < \infty$, this limit vanishes as $M \to \infty$. Thus, for sufficiently large $M$, $\Delta_1$ is negligible with high probability.   

Second, we analyze the term $\Delta_2 = |\phi_{T_{n,M}^*}(t) - \phi_{Y_M}(t)|$.
    Let $\mathcal{K}_+ = \{k \in \{1, \dots, M\}: \lambda_k > 0\}$ be the set of indices corresponding to strictly positive population eigenvalues. For any $k \in \mathcal{K}_+$, we have $\widehat{\lambda}_{k,n} \xrightarrow{p} \lambda_k > 0$.
    
    We first establish the asymptotic normality of the vector $\mathbf{Z}_{\mathcal{K}_+} = (Z_{k,n})_{k \in \mathcal{K}_+}$. Recall that $Z_{k,n} = \sum_{i=1}^n u_{k,i} V_i$, where $\sum_i u_{k,i}^2 = 1$ and $V_i$ are i.i.d. multipliers. To apply the Lindeberg--Feller CLT, it suffices to verify Lyapunov's condition. Given that multipliers have finite fourth moments, this is satisfied if $\max_{1 \le i \le n} |u_{k,i}| \xrightarrow{p} 0$.
    From the eigenvalue equation $\mathbf{H}_n \mathbf{u}_k = n \widehat{\lambda}_{k,n} \mathbf{u}_k$, we express the $i$-th eigenvector component as:
    \begin{equation*}
        u_{k,i} = \frac{1}{n \widehat{\lambda}_{k,n}} \sum_{j=1}^n H_{ij} u_{k,j}.
    \end{equation*}
    By the Cauchy-Schwarz inequality and the normalization $\|\mathbf{u}_k\|=1$:
    \begin{equation*}
        \max_{1\le i\le n}|u_{k,i}| \le \frac{1}{n \widehat{\lambda}_{k,n}} \max_{1\le i \le n}\left( \sum_{j=1}^n H_{ij}^2 \right)^{1/2} \|\mathbf{u}_k\| = \frac{1}{\sqrt{n} \widehat{\lambda}_{k,n}} \max_{1\le i\le n}\left( \frac{1}{n} \sum_{j=1}^n \left(f_p(S_i,S_j,\theta_0)\right)^2 \right)^{1/2}.
    \end{equation*}
     Since $k \in \mathcal{K}_+$, the denominator satisfies $\widehat{\lambda}_{k,n} \xrightarrow{p} \lambda_k > 0$. Consequently, by Lemma \ref{lemma:max_row_bound}, $\max_{1 \le i \le n} |u_{k,i}|=o_p(1)$.
    
    Thus, the Cramér--Wold device implies that conditional on the data, the sub-vector $\mathbf{Z}_{\mathcal{K}_+}$ converges in distribution to a standard multivariate normal vector $\mathbf{W}_{\mathcal{K}_+} \sim \mathcal{N}(\mathbf{0}, \mathbf{I}_{|\mathcal{K}_+|})$.
    
    Now consider the statistic $T_{n,M}^*$. We decompose the sum based on the index sets:
    \begin{equation*}
        T_{n,M}^* = \sum_{k \in \mathcal{K}_+} \widehat{\lambda}_{k,n} Z_{k,n}^2 + \sum_{k \in \mathcal{K}_0} \widehat{\lambda}_{k,n} Z_{k,n}^2.
    \end{equation*}
    For the first term, by the Continuous Mapping Theorem and the convergence of eigenvalues and eigenvectors derived above:
    \begin{equation*}
        \sum_{k \in \mathcal{K}_+} \widehat{\lambda}_{k,n} Z_{k,n}^2 \xrightarrow{d^*} \sum_{k \in \mathcal{K}_+} \lambda_k W_k^2.
    \end{equation*}
    For the second term, where $k \in \mathcal{K}_0$ (i.e., $\lambda_k=0$), we observe that $\widehat{\lambda}_{k,n} \xrightarrow{p} 0$. Furthermore, $\mathbb{E}^*[Z_{k,n}^2] = 1$ implies $Z_{k,n}^2 = O_p^*(1)$. Therefore, by Slutsky's Theorem,
    \begin{equation*}
        \widehat{\lambda}_{k,n} Z_{k,n}^2 \xrightarrow{p^*} 0 \cdot O_p^*(1) = 0.
    \end{equation*}
    Combining these, we have $T_{n,M}^* \xrightarrow{d^*} \sum_{k \in \mathcal{K}_+} \lambda_k W_k^2 + 0 = Y_M$. By Lévy's Continuity Theorem, this implies $\Delta_2 \xrightarrow{p} 0$ as $n \to \infty$.

Finally, $\Delta_3 = \mathbb{E}[\sum_{k=M+1}^\infty \lambda_k W_k^2] = \sum_{k=M+1}^\infty \lambda_k$ is deterministic and vanishes as $M \to \infty$.

To conclude, for any $\varepsilon > 0$, we first choose $M$ sufficiently large such that $|t|(\sum_{k=M+1}^\infty \lambda_k) < \varepsilon/2$. With $M$ fixed, as $n \to \infty$, both $|t|\Delta_1$ (converging to the tail sum) and $\Delta_2$ (converging to 0) behave such that the probability of the RHS of \eqref{eq:triangle} exceeding $\varepsilon$ tends to zero. Thus, $|\phi_{n}^*(t) - \phi_Y(t)| \xrightarrow{p} 0$, which implies $T_n^* \xrightarrow{d^*} Y$ in probability.

Under the fixed alternative $\mathbb{H}_1$, the bootstrap statistic is $T_n^*(\theta_1) = n^{-1} \boldsymbol{V}^{\prime} \mathbf{H}_n(\theta_1) \boldsymbol{V}$. Crucially, although the kernel $H$ is non-degenerate, the zero mean of the multipliers renders $T_n^*$ a degenerate quadratic form conditional on the sample. Assuming $\mathbb{E}[V_i^4] < \infty$, the conditional second moment expands to:
\begin{equation}
    \mathbb{E}^*[(T_n^*(\theta_1))^2] = \frac{1}{n^2} \left( (\mathbb{E}[V^4] - 3) \sum_{i=1}^n H_{ii}^2 + (\text{tr}(\mathbf{H}_n))^2 + 2 \text{tr}(\mathbf{H}_n^2) \right).
\end{equation}
The first term is $O_p(n^{-1})$ and vanishes asymptotically. By the Law of Large Numbers, the remaining trace terms satisfy $n^{-1} \text{tr}(\mathbf{H}_n) \xrightarrow{p} \mathbb{E}[H(Z, Z)]$ and $n^{-2} \text{tr}(\mathbf{H}_n^2) \xrightarrow{p} \mathbb{E}[H(Z, Z')^2]$. Consequently,
\begin{equation}
    \mathbb{E}^*[(T_n^*(\theta_1))^2] \xrightarrow{p} (\mathbb{E}[H(Z, Z)])^2 + 2 \mathbb{E}[H(Z, Z')^2].
\end{equation}
Given this finite probability limit, Chebyshev's inequality implies $T_n^*(\theta_1) = O_{p^*}(1)$, completing the proof.
\end{proof}

%\newpage
\section{Additional Methodological Discussions}
\label{sec:discussion}
\subsection{Discussion on Two-Step Efficient Estimation}
\label{app:efficiency}
The proposed KMD estimator $\widehat{\theta}_{n}$ is $\sqrt{n}$-consistent and asymptotically normal, effectively exploiting the continuum of unconditional moment restrictions. However, it generally does not attain the semiparametric efficiency bound of \cite{chamberlain1987asymptotic}. Achieving efficiency directly within the minimum distance framework would require characterizing a kernel function whose associated RKHS embedding induces an optimal weighting measure over the continuum of moment conditions. A potential solution to this problem is to construct the regularized inverse of the covariance operator, as developed in \cite{carrasco2000generalization}. However, as they explicitly note, this inversion is ill-posed and requires spectral decomposition and the selection of data-dependent tuning parameters (e.g., for Tikhonov regularization).

In this subsection, we adopt the two-step strategy of \cite{dominguez2004consistent}, which offers a computationally efficient approach while preserving the operational simplicity of the KMD framework. Let $J_n(\theta)$ denote the feasible efficient GMM objective function, constructed using consistent nonparametric estimators of conditional moments (e.g., via kernel or series methods as in \cite{newey1990efficient} and \cite{robinson1991best}). Given that $\widehat{\theta}_n$ is $\sqrt{n}$-consistent, it serves as a robust starting value that allows us to construct an asymptotically efficient estimator $\widehat{\theta}_{\text{eff}}$ via a single Newton-Raphson update:
\begin{equation}
\label{eq:NR-update}
    \widehat{\theta}_{\text{eff}} = \widehat{\theta}_n - \left[ \ddot{J}_n(\widehat{\theta}_n) \right]^{-1} \dot{J}_n(\widehat{\theta}_n),
\end{equation}
where $\dot{J}_n$ and $\ddot{J}_n$ denote the gradient and Hessian of the feasible objective, respectively. Under regularity conditions sufficient for the consistency of the nonparametric estimators (see \cite{newey1990efficient}, \cite{robinson1991best}), it follows from \cite{robinson1988stochastic} that $\widehat{\theta}_{\text{eff}} - \widetilde{\theta} = o_p(n^{-1/2})$, where $\widetilde{\theta}$ is the theoretically efficient GMM estimator.

\begin{remark}
    It is important to acknowledge that there is ``no free lunch'' in the pursuit of semiparametric efficiency. Specifically, the evaluation of the derivatives $\dot{J}_n$ and $\ddot{J}_n$ required for the update in \eqref{eq:NR-update} necessitates the consistent nonparametric estimation of conditional moments. This requirement introduces nontrivial practical challenges. First, standard nonparametric techniques typically assume that the conditioning covariates have compact support. Second, as the dimension of the covariates increases, these estimators suffer from the ``curse of dimensionality'', leading to slower convergence rates. Third, the implementation requires selecting smoothing parameters, a step that can be highly sensitive to statistical inference. In sharp contrast, the proposed KMD estimator $\widehat{\theta}_{n}$ is simple to implement and robust; it circumvents the direct estimation of conditional moments and imposes no compactness restrictions on the support of the covariates. Notably, our simulation results suggest that the finite-sample efficiency loss of the simpler $\widehat{\theta}_{n}$ is often marginal; see Section \ref{sec:simu} for a detailed comparison.
\end{remark}

\subsection{Discussion on $U$ or $V$ Statistic for Estimation and Testing}
\label{subsec:UorV}
\subsubsection{On Estimation}
\label{subsec:UorV-estimation}

The off-diagonal $U$-statistic criterion and the corresponding $V$-statistic criterion are asymptotically close, differing only by diagonal terms. The distinction, however, can matter for the finite-sample geometry of the sample objective. Off-diagonal criteria arise naturally in smooth minimum-distance and local-smoothing approaches, including \citet{lavergne2013smooth}. Our use of the $V$-statistic is motivated by a different consideration: it preserves the empirical RKHS squared-norm representation
$$
    \widehat Q_V(\theta)
    =
    \|\widehat\mu_n(\theta)\|_{\mathcal H_K}^2
    =
    \frac{1}{n^2}\rho(\theta)'K\rho(\theta).
$$
Since the Gram matrix $K$ is positive semidefinite, $\widehat Q_V(\theta)\ge 0$ for all $\theta$. By contrast, the off-diagonal criterion
$$
    \widehat Q_U(\theta)
    =
    \frac{1}{n(n-1)}
    \rho(\theta)'\{K-\operatorname{diag}(K)\}\rho(\theta)
$$
does not generally define a squared norm, because $K-\operatorname{diag}(K)$ need not be positive semidefinite.

Figure \ref{fig:uv-objective-nliv} illustrates this finite-sample distinction in the continuous nonlinear IV design used in the simulation section. To isolate the effect of deleting the diagonal terms, both criteria are computed with the same Gaussian median kernel and evaluated over the grid $\theta\in[-1.5,0.8]$, with the true value $\theta_0=0.5$. The $V$-statistic objective remains nonnegative and is minimized near the true parameter in this draw. The off-diagonal $U$-statistic objective, while asymptotically close to the $V$-version, can take negative values and attain lower values near the boundary of the plotted search interval. This diagnostic is not intended as a performance comparison or as a critique of the asymptotic theory of off-diagonal SMD estimators. It simply illustrates that deleting the diagonal terms can change the finite-sample optimization landscape. In nonlinear models, both objectives may be nonconvex in $\theta$, but the $V$-criterion retains the lower-bound and metric interpretation implied by its squared-norm structure.

\begin{figure}[htbp]
    \centering
    \includegraphics[width=0.90\textwidth]{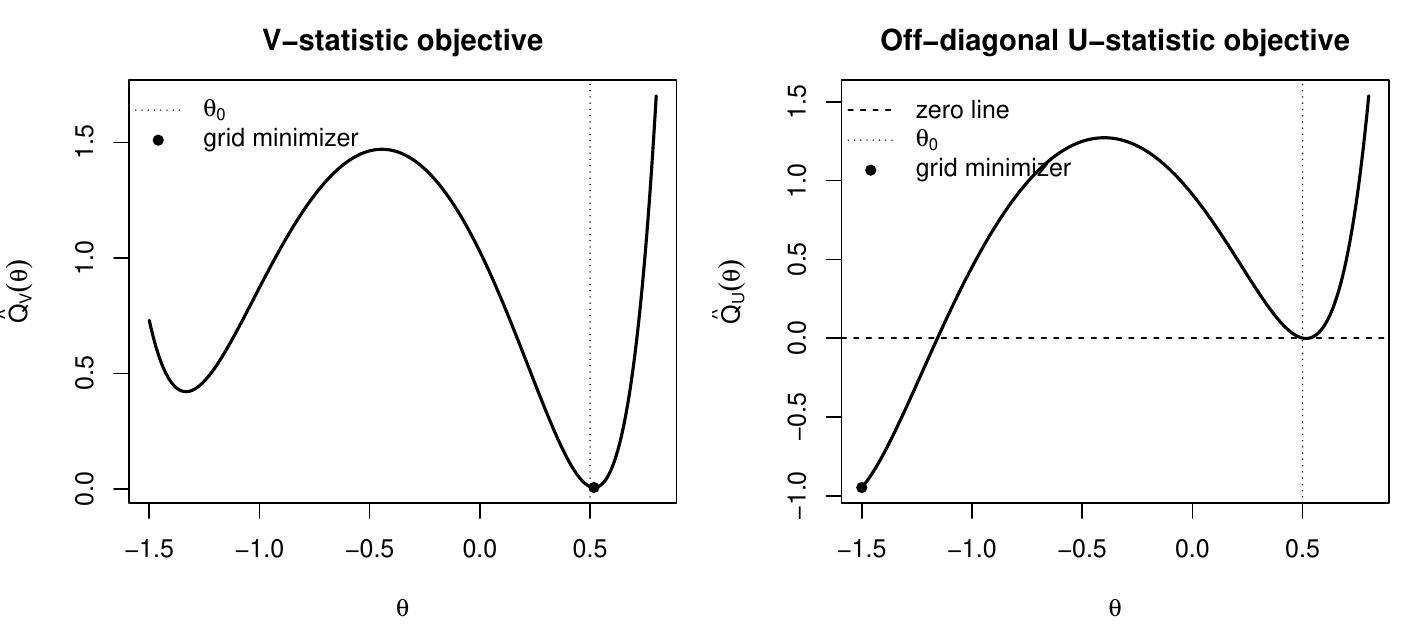}
    \caption{Finite-sample objective functions in the continuous NLIV design, evaluated over $\theta\in[-1.5,0.8]$. The left panel plots the $V$-statistic RKHS squared-norm objective. The right panel plots the off-diagonal $U$-statistic objective computed with the same kernel. The vertical dotted line marks the true value $\theta_0=0.5$.}
    \label{fig:uv-objective-nliv}
\end{figure}

The following stylized calculation explains why such behavior can occur. Consider two observations with $s_1=a$, $s_2=-a$, $a>0$, and kernel matrix
$$
K=
\begin{pmatrix}
1 & \kappa\\
\kappa & 1
\end{pmatrix},
\qquad 0<\kappa<1.
$$
In the correctly specified noiseless exponential model $y_i=\exp(1+\theta_0s_i)$, the $V$-statistic objective $4^{-1}\rho(\theta)'K\rho(\theta)$ is nonnegative and is uniquely minimized at $\theta_0$. By contrast, the off-diagonal objective equals
$$
Q_U(\theta)
=
\frac{\kappa}{2}\rho_1(\theta)\rho_2(\theta)
=
-\kappa e^2\{\cosh(a(\theta-\theta_0))-1\},
$$
which is zero at $\theta_0$ and strictly negative for every $\theta\neq\theta_0$. Hence, on a compact parameter space, the off-diagonal criterion may attain its finite-sample minimum away from $\theta_0$. This example does not contradict the asymptotic validity of off-diagonal criteria; rather, it isolates the finite-sample algebraic effect of removing the diagonal terms from a squared-norm criterion.

\subsubsection{On Testing}
\label{subsec:UorV-testing}

While the $V$-statistic formulation is essential for parameter estimation to ensure a valid, non-negative minimization objective, an alternative strategy for the testing stage involves the $U$-statistic, constructed by excluding the diagonal terms:
$$
\widetilde T_n\coloneq\frac{1}{n-1}\sum_{i\neq j}\rho(Z_i,\widehat{\theta}_n)^\prime K(X_i,X_j)\rho(Z_j,\widehat{\theta}_n).
$$
The primary theoretical appeal of this alternative lies in the potential relaxation of regularity conditions; specifically, it reduces the finite fourth-moment requirement on $\rho(Z, \theta_0)$ and $g(Z, \theta_0)$ to finite second moments by eliminating the need to control the diagonal components. However, we retain the unified $V$-statistic framework for the following reasons that outweigh this marginal gain.

As discussed in Remark \ref{rem:generalized_sargan}, the $V$-statistic admits a natural interpretation as a Generalized Sargan Test (or $J$-test) by representing the squared RKHS norm $\|\widehat{\mu}_n(\widehat{\theta}_n)\|_{\mathcal{H}_K}^2$, guaranteeing a non-negative measure of misspecification. In contrast, the $U$-type test statistic corresponds to an indefinite quadratic form that deviates from the distance-minimization framework. As shown below, the limiting distribution of $\widetilde{T}_n$ under $\mathbb{H}_0$ involves negative nuisance shifts:
$$
\widetilde{T}_n\xrightarrow{d}\sum_{k=1}^{\infty}\lambda_k\left(W_k^2-1\right)-\delta,
$$
where $\delta\coloneq\mathbb{E}[\rho(Z,\theta_0)^\prime G(X,\theta_0)\Delta(\theta_0)^{-1}G(X,\theta_0)^\prime\rho(Z,\theta_0)]$ captures the asymptotic bias arising from the diagonal components inherent to the $V$-statistic estimation. These shift terms allow $\widetilde{T}_n$ to take negative values, rendering its interpretation counter-intuitive as a measure of model fit. Analogously, under the local alternative $\mathbb{H}_{1n}$, the distribution exhibits an identical shift structure:
$$
\widetilde{T}_n\xrightarrow{d}\sum_{k=1}^{\infty}\lambda_k\left[\left(W_k+c_k\right)^2-1\right]-\delta,
$$
which deviates from the unified spectral representation of \cite{bierens1997asymptotic}.

Finally, it is worth noting that the validity of the preliminary estimator $\widehat{\theta}_n$ inherently relies on the $V$-statistic formulation, which requires the finiteness of $\mathbb{E}[\|g(Z, \theta_0)'\rho(Z, \theta_0)\|^2]$. Given this unavoidable requirement, we consider the marginal relaxation of regularity conditions offered by the $U$-type test statistic insufficient to justify sacrificing the geometric coherence and analytical interpretability of our unified framework.

\paragraph{Detailed Derivation for the Asymptotics of $U$-type Test Statistic}
\begin{proof}
We adopt the notation from the proofs of Lemmas \ref{lemma:theta_lr} and \ref{lemma:repre_H0}. Let the $U$-type test statistic $\widetilde T_n\coloneq n\widetilde Q_n(\widehat\theta_n)$, with
$$
\widetilde Q_n(\theta)\coloneq\frac{1}{n(n-1)}\sum_{i\neq j}\rho(Z_i,\theta)^\prime K(X_i,X_j)\rho(Z_j,\theta).
$$
Under $\mathbb{H}_0$, by the mean value expansion along with equation \ref{eq:mv-T}, we have
$$
\widetilde T_n=n\widehat{Q}_n(\theta_0)+\left(\sqrt{n}\nabla_\theta\widetilde{Q}_n(\theta_0)\right)'\left[\sqrt{n}\left(\widehat{\theta}_n-\theta_0\right)\right]+\frac{1}{2}\left[\sqrt{n}\left(\widehat{\theta}_n-\theta_0\right)\right]'\widetilde{H}_n(\bar{\theta}_n)\left[\sqrt{n}\left(\widehat{\theta}_n-\theta_0\right)\right].
$$
Noting that $U$-statistics and $V$-statistics share asymptotically equivalent gradients and Hessians, we have
  \begin{align*}
        \widetilde{H}_n(\bar{\theta}_n)&=2\Delta(\theta_0)+o_p(1),\\
        \sqrt{n}\nabla_\theta\widetilde{Q}_n(\theta_0)&=2S_n(\theta_0)+o_p(1),\\
        \sqrt{n}\left(\widehat{\theta}_n-\theta_0\right)&=-\Delta(\theta_0)^{-1}S_n(\theta_0)+o_p(1).
    \end{align*}
    Thus,
    \begin{align*}
       \widetilde T_n&=n\widetilde Q_n(\theta_0)-S_n(\theta_0)\Delta(\theta_0)^{-1}S_n(\theta_0)+o_p(1)\\
       &=\frac{1}{n-1}\sum_{i\neq j}\rho(Z_i,\theta_0)^\prime K(X_i,X_j)\rho(Z_j,\theta_0)\\
       &\quad -\frac{1}{n}\sum_i\sum_j\rho(Z_i,\theta_0)^\prime G(X_i,\theta_0)\Delta(\theta_0)^{-1}G(X_j,\theta_0)^\prime\rho(Z_j,\theta_0)+o_p(1)\\
       &=\frac{1}{n-1}\sum_{i\neq j}\rho(Z_i,\theta_0)^\prime\left[K(X_i,X_j)-G(X_i,\theta_0)\Delta(\theta_0)^{-1}G(X_j,\theta_0)^\prime\right]\rho(Z_j,\theta_0)\\
       &\quad +\frac{1}{n(n-1)}\sum_{i\neq j}\rho(Z_i,\theta_0)^\prime G(X_i,\theta_0)\Delta(\theta_0)^{-1}G(X_j,\theta_0)^\prime \rho(Z_j,\theta_0)\\
       &\quad-\frac{1}{n}\sum_{i} \rho(Z_i,\theta_0)^\prime G(X_i,\theta_0)\Delta(\theta_0)^{-1}G(X_i,\theta_0)^\prime \rho(Z_i,\theta_0)+o_p(1).
    \end{align*}
    Furthermore, under $\mathbb{H}_0$, 
    $$
    \frac{1}{n(n-1)}\sum_{i\neq j}\rho(Z_i,\theta_0)^\prime\left[K(X_i,X_j)-G(X_i,\theta_0)\Delta(\theta_0)^{-1}G(X_j,\theta_0)^\prime\right]\rho(Z_j,\theta_0)\xrightarrow{p}0,
    $$
    $$
    \frac{1}{n}\sum_{i} \rho(Z_i,\theta_0)^\prime G(X_i,\theta_0)\Delta(\theta_0)^{-1}G(X_i,\theta_0)^\prime \rho(Z_i,\theta_0)\xrightarrow{p}\mathbb{E}[\rho(Z,\theta_0)^\prime G(X,\theta_0)\Delta(\theta_0)^{-1}G(X,\theta_0)^\prime\rho(Z,\theta_0)]\coloneq\delta.
    $$
    Thus, under $\mathbb{H}_0$,
    $$
    \widetilde T_n=\frac{1}{n-1}\sum_{i\neq j}\rho(Z_i,\theta_0)^\prime K^p(X_i,X_j)\rho(Z_j,\theta_0)-\delta+o_p(1),
    $$
    where the first term is an $n$-scaled degenerate $U$-statistic. Invoking the Theorem in Section 5.5.2 of \citet{serfling1980approximation}, we obtain
$$
\widetilde T_n\xrightarrow{d}\sum_{k=1}^{\infty}\lambda_k\left(W_k^2-1\right)-\delta,
$$
where $\lambda_k$ and $W_k$ are defined as in Theorem \ref{thm:H0}.
\end{proof}

\section{Additional Simulation Results}
\label{sec:add-simu}

This section provides additional simulation details to verify the robustness of the proposed framework. We first present supplementary estimation results for the Box--Cox model with $\lambda_0 = 1.0$ in Section \ref{subsec:boxcox-lambda1}, followed by a detailed description of the Double Machine Learning (DML) implementation in Section \ref{subsec:dml-details}. Section \ref{subsec:coverage} validates the proposed inference theory by reporting the pointwise coverage probabilities of the asymptotic confidence intervals across different model specifications. Lastly, Section \ref{subsec:dl-test} provides comparative simulation results for the specification test of \cite{dominguez2015simple}.

\subsection{The Linear Regression Model}
\label{app:lr_simulation}

We set the standard Linear Regression Model:
\begin{equation*}
    Y_{i} = X_{i}^{\prime}\beta_{0} + \epsilon_{i}, \quad i=1,\dots,n,
\end{equation*}
where the true parameter vector is set to $\beta_{0}=(1,\dots,1,-1,\dots,-1)^{\prime}$, with the first $\lceil p/2\rceil$ elements equal to 1. The regressors $X_{i}$ are generated from a multivariate normal distribution $\mathcal{N}(0,\Sigma_{X})$ with a Toeplitz covariance structure, $\Sigma_{X,jk} = 0.5^{|j-k|}$, to induce moderate correlation among covariates. We examine two error specifications to assess robustness: (i) \textit{Homoskedasticity}, with $\epsilon_{i} \sim \mathcal{N}(0,1)$; and (ii) \textit{Heteroskedasticity}, with $\epsilon_{i} = \exp(0.25 X_{i,1}) \cdot \nu_{i}$, where $\nu_{i} \sim \mathcal{N}(0,1)$.

For estimation, we benchmark the KMD estimator against Ordinary Least Squares (OLS), Feasible Weighted Least Squares (FWLS), and the infeasible Oracle GMM estimator (GMM$^*$). GMM$^*$ employs the true conditional variance weights to attain the semiparametric efficiency bound, whereas FWLS relies on a feasible log-linear variance approximation.

\begin{table}[htbp]
  \centering
  \caption{Estimation Performance in Linear Regression (Bias and RMSE)}
  \label{tab:lr_estimation}
  
  \footnotesize
  \renewcommand{\arraystretch}{1.15}
  
  % 定义等宽居中列 Y
  \newcolumntype{Y}{>{\centering\arraybackslash}X}
  
  % 使用 tabularx，前两列 c，后八列 Y (Bias 4列 + RMSE 4列)
  \begin{tabularx}{\textwidth}{cc YYYY YYYY}
    \toprule
    & & \multicolumn{4}{c}{Bias} & \multicolumn{4}{c}{RMSE} \\
    \cmidrule(lr){3-6} \cmidrule(lr){7-10}
    $n$ & $p$ & OLS & FWLS & GMM$^*$ & KMD & OLS & FWLS & GMM$^*$ & KMD \\
    \midrule
    
    \multicolumn{10}{l}{\textit{Panel A: Homoskedastic Errors}} \\
    \addlinespace[0.2em]
    \multirow{3}{*}{$100$} 
      & $3$ & $0.005$ & $0.005$ & $0.005$ & $0.009$ & $0.121$ & $0.130$ & $0.121$ & $0.141$ \\
      & $5$ & $0.005$ & $0.009$ & $0.005$ & $0.007$ & $0.125$ & $0.138$ & $0.125$ & $0.157$ \\
      & $10$ & $0.013$ & $0.018$ & $0.013$ & $0.019$ & $0.135$ & $0.154$ & $0.135$ & $0.179$ \\
    \addlinespace[0.5em]
    \multirow{3}{*}{$200$} 
      & $3$ & $0.002$ & $0.003$ & $0.002$ & $0.003$ & $0.086$ & $0.091$ & $0.086$ & $0.097$ \\
      & $5$ & $0.003$ & $0.005$ & $0.003$ & $0.005$ & $0.088$ & $0.094$ & $0.088$ & $0.106$ \\
      & $10$ & $0.010$ & $0.009$ & $0.010$ & $0.013$ & $0.092$ & $0.101$ & $0.092$ & $0.121$ \\

    \midrule
    
    \multicolumn{10}{l}{\textit{Panel B: Heteroskedastic Errors}} \\
    \addlinespace[0.2em]
    \multirow{3}{*}{$100$} 
      & $3$ & $0.007$ & $0.004$ & $0.004$ & $0.010$ & $0.132$ & $0.122$ & $0.113$ & $0.149$ \\
      & $5$ & $0.006$ & $0.010$ & $0.006$ & $0.007$ & $0.135$ & $0.131$ & $0.117$ & $0.165$ \\
      & $10$ & $0.016$ & $0.015$ & $0.012$ & $0.021$ & $0.144$ & $0.150$ & $0.129$ & $0.189$ \\
    \addlinespace[0.5em]
    \multirow{3}{*}{$200$} 
      & $3$ & $0.002$ & $0.003$ & $0.003$ & $0.001$ & $0.094$ & $0.084$ & $0.079$ & $0.102$ \\
      & $5$ & $0.004$ & $0.004$ & $0.003$ & $0.006$ & $0.095$ & $0.088$ & $0.082$ & $0.111$ \\
      & $10$ & $0.011$ & $0.008$ & $0.009$ & $0.014$ & $0.099$ & $0.096$ & $0.087$ & $0.127$ \\
      
    \bottomrule
  \end{tabularx}
  \par
  \vspace{0.5em}
  \begin{minipage}{\textwidth}
    \footnotesize
    \textit{Note:} The table reports the mean bias norm ($||\widehat{\beta} - \beta_0||_2$) and average element-wise RMSE for linear regression estimators over $1000$ replications.
  \end{minipage}
\end{table}

Table \ref{tab:lr_estimation} summarizes the finite-sample estimation performance under linear regression settings. In terms of consistency, the KMD estimator exhibits negligible bias across all scenarios and matches parametric benchmarks. Regarding efficiency, the KMD estimator incurs a moderate efficiency loss relative to OLS (Panel A; e.g., approximately 15--20\% increase in RMSE for $n=100$). Under asymmetric heteroskedasticity (Panel B), while Oracle GMM and FWLS naturally outperform unweighted methods, the KMD estimator remains stable, with RMSEs comparable to OLS. These findings are unsurprising and acceptable, given that the standard KMD estimator relies on a global kernel metric and does not explicitly exploit the conditional variance structure for weighting. 

We next evaluate the specification test for the linear regression model, where the null hypothesis asserts that $\mathbb{E}[Y - X^{\prime}\beta_0 \mid X] = 0$ a.s. Under this specification, the score function $g(Z, \beta) = -X$ is independent of $\beta$, yielding a key computational advantage as the correction terms in the projected kernel $\widehat{K}_n^p$ depend solely on the observed regressors.

To evaluate power, we introduce a perturbation $f(X_i)$ under the alternative $\mathbb{H}_{1}: Y_i = X_i^{\prime}\beta_{0} + f(X_i) + \epsilon_i$. Setting the magnitude to $\delta=0.5$, we consider three distinct misspecification classes, where $X_{i,k}$ denotes the $k$-th component of the regressor vector $X_i$:
\begin{itemize}
    \setlength{\itemsep}{0pt}
    \setlength{\parskip}{0pt}
    \setlength{\parsep}{0pt}
    \item $\mathbb{H}_{1a}$ (Quadratic): $f(X_i) = \delta X_{i,1}^{2}$, representing smooth, global nonlinearity.
    \item $\mathbb{H}_{1b}$ (Discontinuous): $f(X_i) = \delta \cdot \mathbb{I}(|X_{i,1}| \le 0.5)$, representing a local structural break or ``bump'' often difficult for global polynomial tests to detect.
    \item $\mathbb{H}_{1c}$ (Interaction): $f(X_i) = \delta X_{i,1}X_{i,2}$, capturing omitted cross-variable dependencies.
\end{itemize}

\begin{table}[htbp]
  \centering
  \caption{Empirical Rejection Rates for Linear Regression Specification Tests ($\alpha=0.05$)}
  \label{tab:lr_test}
  
  \small 
  \renewcommand{\arraystretch}{1.15} 
  
  \newcolumntype{Y}{>{\centering\arraybackslash}X}
  
  \begin{tabularx}{\textwidth}{cc YYY YYY}
    \toprule
    & & \multicolumn{3}{c}{Homoskedastic Errors} & \multicolumn{3}{c}{Heteroskedastic Errors} \\
    \cmidrule(lr){3-5} \cmidrule(lr){6-8}
    $n$ & DGP & $p=3$ & $p=5$ & $p=10$ & $p=3$ & $p=5$ & $p=10$ \\
    \midrule
    
    % --- n=100 ---
    \multirow{4}{*}{$100$} 
      & Size ($\mathbb{H}_0$)             & $0.046$ & $0.050$ & $0.031$ & $0.040$ & $0.047$ & $0.030$ \\
      & Power ($\mathbb{H}_{1a}$: Quad.)  & $0.956$ & $0.949$ & $0.915$ & $0.946$ & $0.920$ & $0.894$ \\
      & Power ($\mathbb{H}_{1b}$: Disc.)  & $0.502$ & $0.455$ & $0.344$ & $0.476$ & $0.425$ & $0.315$ \\
      & Power ($\mathbb{H}_{1c}$: Inter.) & $0.556$ & $0.499$ & $0.400$ & $0.514$ & $0.451$ & $0.375$ \\
      
    \addlinespace[0.5em] 

    \multirow{4}{*}{$200$} 
      & Size ($\mathbb{H}_0$)             & $0.058$ & $0.047$ & $0.040$ & $0.061$ & $0.047$ & $0.047$ \\
      & Power ($\mathbb{H}_{1a}$: Quad.)  & $1.000$ & $1.000$ & $1.000$ & $1.000$ & $1.000$ & $1.000$ \\
      & Power ($\mathbb{H}_{1b}$: Disc.)  & $0.846$ & $0.790$ & $0.724$ & $0.823$ & $0.741$ & $0.685$ \\
      & Power ($\mathbb{H}_{1c}$: Inter.) & $0.946$ & $0.844$ & $0.812$ & $0.917$ & $0.796$ & $0.766$ \\

    \addlinespace[0.5em]
    
    % --- n=400 ---
    \multirow{4}{*}{$400$} 
      & Size ($\mathbb{H}_0$)             & $0.052$ & $0.052$ & $0.049$ & $0.053$ & $0.049$ & $0.049$ \\
      & Power ($\mathbb{H}_{1a}$: Quad.)  & $1.000$ & $1.000$ & $1.000$ & $1.000$ & $1.000$ & $1.000$ \\
      & Power ($\mathbb{H}_{1b}$: Disc.)  & $0.983$ & $0.976$ & $0.962$ & $0.978$ & $0.964$ & $0.940$ \\
      & Power ($\mathbb{H}_{1c}$: Inter.) & $1.000$ & $0.999$ & $0.997$ & $0.995$ & $0.987$ & $0.979$ \\
      
    \bottomrule
  \end{tabularx}
\end{table}

Table \ref{tab:lr_test} reports the empirical rejection rates for the specification test. Under the null hypothesis ($\mathbb{H}_0$), the proposed test maintains accurate size control across all designs. The rejection rates align closely with the nominal $5\%$ level and are robust to both homoskedastic and heteroskedastic errors. While slight undersizing is observed in high-dimensional small samples (e.g., $n=100, p=10$), this distortion diminishes as the sample size increases to $n=200$.

Regarding power performance, the test displays consistent power against all three classes of misspecification, with rejection rates increasing monotonically toward unity as the sample size grows. For the smooth quadratic alternative ($\mathbb{H}_{1a}$), the test is particularly powerful, achieving near-unity rejection rates even at $n=100$. For the discontinuous ($\mathbb{H}_{1b}$) and interaction ($\mathbb{H}_{1c}$) alternatives, the power is substantial but naturally lower than that for the global quadratic alternative, reflecting the challenge of detecting local irregularities or multivariate dependencies. Furthermore, the results illustrate the expected trade-off with respect to dimensionality: for a fixed sample size, power decreases as the dimension $p$ increases from $3$ to $10$. Nevertheless, the test remains omnibus, with power exceeding $0.94$ in all scenarios when $n=400$, confirming the consistency of the KMD test against a broad range of alternatives. 

\subsection{Additional Simulation Results for Box--Cox Estimation}
\label{subsec:boxcox-lambda1}

Table \ref{tab:boxcox_estimation_3} reports the estimation performance under the linear specification ($\lambda_0=1.0$). While estimation errors increase across all methods relative to the $\lambda_0=0$ and $\lambda_0=0.5$ cases, likely due to the weaker identification of the transformation parameter as the model approaches linearity, the KMD estimator consistently maintains superior performance in terms of both Bias and RMSE.
\begin{table}[htbp]
  \centering
  \caption{Estimation Performance of Box--Cox Model Parameters ($\lambda=1.0$)}
  \label{tab:boxcox_estimation_3}
  
  \footnotesize 
  \renewcommand{\arraystretch}{1.15}
  
  % --- 表格定义 ---
  % 1. \setlength{\tabcolsep}{0pt} 让 extrolsep 完全控制内部间距
  % 2. 在列定义首尾添加 @{\hspace{1em}} 为左右两侧留出呼吸空间
  \setlength{\tabcolsep}{0pt} 
  \begin{tabular*}{\textwidth}{
    @{\hspace{1em}\extracolsep{\fill}} % 左侧留白 + 自动填充
    c c c 
    *{3}{S[table-format=-1.3]} % Bias 三列
    *{3}{S[table-format=1.3]}  % RMSE 三列
    @{\hspace{1em}}            % 右侧留白
  }
    \toprule
    & & & \multicolumn{3}{c}{Bias} & \multicolumn{3}{c}{RMSE} \\
    \cmidrule(lr){4-6} \cmidrule(lr){7-9}
    
    % S列表头需用 {} 包裹
    {$n$} & {$p$} & {Param.} & {NL2S} & {Shin-MD} & {KMD} & {NL2S} & {Shin-MD} & {KMD} \\
    \midrule

    \multicolumn{9}{l}{\textit{Panel A: Homoskedastic Errors}} \\
    \addlinespace[0.2em]
    \multirow{9}{*}{$100$} & \multirow{3}{*}{$3$} 
         & $\alpha$  & -0.045 &  0.054 & -0.058 & 0.313 & 0.315 & 0.185 \\
       & & $\beta$   &  0.042 &  0.455 &  0.016 & 0.174 & 0.666 & 0.135 \\
       & & $\lambda$ & -0.028 & -0.297 & -0.050 & 0.232 & 0.650 & 0.130 \\
    \addlinespace[0.2em]
                           & \multirow{3}{*}{$5$} 
         & $\alpha$  & -0.130 & -0.057 & -0.102 & 0.297 & 0.341 & 0.191 \\
       & & $\beta$   &  0.038 &  0.266 &  0.053 & 0.145 & 0.432 & 0.133 \\
       & & $\lambda$ & -0.078 & -0.194 & -0.067 & 0.171 & 0.446 & 0.104 \\
    \addlinespace[0.2em]
                           & \multirow{3}{*}{$10$} 
         & $\alpha$  & -0.256 & -0.288 & -0.187 & 0.356 & 0.373 & 0.248 \\
       & & $\beta$   &  0.250 &  0.272 &  0.198 & 0.167 & 0.221 & 0.151 \\
       & & $\lambda$ & -0.113 & -0.126 & -0.085 & 0.151 & 0.153 & 0.102 \\
    \addlinespace[0.5em]
    \multirow{9}{*}{$200$} & \multirow{3}{*}{$3$} 
         & $\alpha$  & -0.007 &  0.112 & -0.024 & 0.200 & 0.425 & 0.133 \\
       & & $\beta$   &  0.029 &  0.453 &  0.004 & 0.106 & 0.740 & 0.093 \\
       & & $\lambda$ & -0.006 & -0.283 & -0.025 & 0.164 & 0.659 & 0.099 \\
    \addlinespace[0.2em]
                           & \multirow{3}{*}{$5$} 
         & $\alpha$  & -0.083 &  0.012 & -0.058 & 0.241 & 0.366 & 0.141 \\
       & & $\beta$   &  0.024 &  0.400 &  0.032 & 0.106 & 0.514 & 0.096 \\
       & & $\lambda$ & -0.051 & -0.240 & -0.039 & 0.139 & 0.530 & 0.075 \\
    \addlinespace[0.2em]
                           & \multirow{3}{*}{$10$} 
         & $\alpha$  & -0.201 & -0.245 & -0.132 & 0.301 & 0.343 & 0.183 \\
       & & $\beta$   &  0.197 &  0.238 &  0.138 & 0.128 & 0.196 & 0.107 \\
       & & $\lambda$ & -0.088 & -0.115 & -0.059 & 0.127 & 0.141 & 0.073 \\

    \addlinespace[0.5em]
    
    \multicolumn{9}{l}{\textit{Panel B: Heteroskedastic Errors}} \\
    \addlinespace[0.2em]
    \multirow{9}{*}{$100$} & \multirow{3}{*}{$3$} 
         & $\alpha$  & -0.033 &  0.041 & -0.078 & 0.312 & 0.331 & 0.199 \\
       & & $\beta$   &  0.048 &  0.478 &  0.010 & 0.168 & 0.706 & 0.138 \\
       & & $\lambda$ & -0.017 & -0.316 & -0.065 & 0.238 & 0.656 & 0.141 \\
    \addlinespace[0.2em]
                           & \multirow{3}{*}{$5$} 
         & $\alpha$  & -0.162 & -0.072 & -0.126 & 0.344 & 0.332 & 0.218 \\
       & & $\beta$   &  0.038 &  0.301 &  0.056 & 0.149 & 0.451 & 0.140 \\
       & & $\lambda$ & -0.096 & -0.231 & -0.083 & 0.189 & 0.492 & 0.120 \\
    \addlinespace[0.2em]
                           & \multirow{3}{*}{$10$} 
         & $\alpha$  & -0.286 & -0.306 & -0.213 & 0.390 & 0.382 & 0.278 \\
       & & $\beta$   &  0.276 &  0.294 &  0.224 & 0.177 & 0.213 & 0.161 \\
       & & $\lambda$ & -0.126 & -0.136 & -0.098 & 0.163 & 0.154 & 0.115 \\
    \addlinespace[0.5em]
    \multirow{9}{*}{$200$} & \multirow{3}{*}{$3$} 
         & $\alpha$  & -0.004 &  0.073 & -0.044 & 0.229 & 0.355 & 0.149 \\
       & & $\beta$   &  0.029 &  0.480 &  0.012 & 0.112 & 0.721 & 0.095 \\
       & & $\lambda$ &  0.004 & -0.296 & -0.035 & 0.180 & 0.640 & 0.107 \\
    \addlinespace[0.2em]
                           & \multirow{3}{*}{$5$} 
         & $\alpha$  & -0.087 & -0.019 & -0.076 & 0.242 & 0.326 & 0.151 \\
       & & $\beta$   &  0.025 &  0.308 &  0.039 & 0.104 & 0.458 & 0.096 \\
       & & $\lambda$ & -0.054 & -0.200 & -0.053 & 0.144 & 0.457 & 0.086 \\
    \addlinespace[0.2em]
                           & \multirow{3}{*}{$10$} 
         & $\alpha$  & -0.234 & -0.260 & -0.151 & 0.331 & 0.338 & 0.203 \\
       & & $\beta$   &  0.230 &  0.235 &  0.161 & 0.136 & 0.176 & 0.113 \\
       & & $\lambda$ & -0.103 & -0.114 & -0.070 & 0.140 & 0.133 & 0.084 \\
    \bottomrule
  \end{tabular*}
  
  \par
  \vspace{0.5em}
  \begin{minipage}{\textwidth}
    \footnotesize
    \textit{Note:} The bias for $\alpha$ and $\lambda$ is the raw bias, while the bias for $\beta$ is the $L_2$ norm of the bias vector.
  \end{minipage}
\end{table}
\FloatBarrier

\subsection{Additional Simulation Results for IV Estimation}
\label{subsec:iv-estimation}

Table \ref{tab:iv_estimation_hetero} reports the estimation performance under heteroskedastic errors. The results are quite similar to those under homoskedasticity: the relative performance of the estimators remains largely unchanged, and KMD continues to perform competitively across the three IV designs.
\begin{table}[htbp]
  \centering
  \vspace{-2em}
  \caption{Estimation Performance of IV Estimators (Bias and RMSE)}
  \label{tab:iv_estimation_hetero}
  
  \footnotesize 
  \renewcommand{\arraystretch}{1.15} 
  
  % --- 表格定义 ---
  % 10列：3个标签列 + 7个数据列
  % S[table-format=-2.3]: 预留两位整数位，适应 Panel B 可能出现的较大数值
  \setlength{\tabcolsep}{0pt}
  \begin{tabular*}{\textwidth}{
    @{\hspace{1em}\extracolsep{\fill}}
    c c l                  
    *{7}{S[table-format=-2.3]} 
    @{\hspace{1em}}
  }
    \toprule
    & & & \multicolumn{7}{c}{Estimators} \\
    \cmidrule(lr){4-10} 
    
    % 表头
    {$n$} & {$p$} & {Metric$^{\dagger}$} & {2SLS$^*$} & {GMM$^*$} & {Sieve-GMM} & {SMD-h1} & {SMD-hdim} & {DML-OS} & {KMD} \\
    \midrule
    
    % =======================================================
    % Panel A: Linear Structural Equation (已填充数据)
    % =======================================================
    \multicolumn{10}{l}{\textit{Panel A: Linear Structural Equation}} \\
    \addlinespace[0.2em]
   \multirow{6}{*}{$100$}
  & \multirow{2}{*}{$3$} & Bias & 0.024 & 0.103 & 0.389 & -0.251 & -0.220 & -0.171 & -0.037 \\
  &                      & RMSE & 3.529 & 3.379 & 3.553 & 3.557 & 3.725 & 3.493 & 3.466 \\
  \addlinespace[0.3em]
  & \multirow{2}{*}{$5$} & Bias & -0.064 & 0.066 & 0.614 & -0.249 & -0.329 & -0.234 & 0.014 \\
  &                      & RMSE & 3.305 & 3.254 & 3.649 & 3.601 & 4.074 & 3.344 & 3.379 \\
  \addlinespace[0.3em]
  & \multirow{2}{*}{$10$} & Bias & 0.075 & 0.195 & 1.333 & -0.200 & -0.677 & -0.105 & 0.148 \\
  &                      & RMSE & 3.471 & 3.324 & 4.424 & 4.309 & 6.705 & 3.415 & 3.522 \\
\addlinespace[0.5em]
\multirow{6}{*}{$200$}
  & \multirow{2}{*}{$3$} & Bias & 0.047 & 0.061 & 0.215 & -0.037 & -0.031 & -0.055 & 0.045 \\
  &                      & RMSE & 2.413 & 2.290 & 2.360 & 2.505 & 2.554 & 2.431 & 2.478 \\
  \addlinespace[0.3em]
  & \multirow{2}{*}{$5$} & Bias & 0.051 & 0.065 & 0.328 & -0.022 & -0.008 & -0.039 & 0.072 \\
  &                      & RMSE & 2.422 & 2.311 & 2.540 & 2.603 & 2.801 & 2.419 & 2.503 \\
  \addlinespace[0.3em]
  & \multirow{2}{*}{$10$} & Bias & 0.114 & 0.087 & 0.722 & -0.076 & -0.290 & 0.016 & 0.083 \\
  &                      & RMSE & 2.486 & 2.406 & 2.821 & 2.881 & 4.274 & 2.472 & 2.532 \\

    \midrule
    
    % =======================================================
\multicolumn{10}{l}{\textit{Panel B: Nonlinear (Exponential) Structural Equation}} \\
\addlinespace[0.2em]
\multirow{6}{*}{$100$}
  & \multirow{2}{*}{$3$} & Bias & 9.565 & -0.079 & 0.547 & -0.605 & -0.740 & -0.624 & 0.063 \\
      &                      & RMSE & 10.387 & 1.706 & 1.923 & 2.485 & 2.722 & 2.434 & 2.035 \\
      \addlinespace[0.3em]
      & \multirow{2}{*}{$5$} & Bias & 9.676 & -0.002 & 1.061 & -0.574 & -1.345 & -0.520 & 0.138 \\
      &                      & RMSE & 10.454 & 1.720 & 2.071 & 2.877 & 9.282 & 2.497 & 2.209 \\
      \addlinespace[0.3em]
      & \multirow{2}{*}{$10$} & Bias & 9.551 & -0.040 & 1.824 & -2.495 & -21.486 & -0.478 & 0.108 \\
      &                      & RMSE & 10.340 & 1.676 & 2.549 & 14.134 & 71.551 & 2.397 & 1.985 \\
\addlinespace[0.5em]
\multirow{6}{*}{$200$}
  & \multirow{2}{*}{$3$} & Bias & 9.669 & -0.092 & 0.221 & -0.349 & -0.387 & -0.383 & -0.032 \\
      &                      & RMSE & 10.053 & 1.165 & 1.283 & 1.612 & 1.601 & 1.573 & 1.477 \\
      \addlinespace[0.3em]
      & \multirow{2}{*}{$5$} & Bias & 9.876 & 0.072 & 0.662 & -0.184 & -0.248 & -0.151 & 0.138 \\
      &                      & RMSE & 10.255 & 1.159 & 1.425 & 1.599 & 1.939 & 1.513 & 1.424 \\
      \addlinespace[0.3em]
      & \multirow{2}{*}{$10$} & Bias & 9.668 & -0.043 & 1.069 & -0.427 & -4.688 & -0.284 & 0.039 \\
      &                      & RMSE & 10.058 & 1.191 & 1.662 & 2.268 & 25.804 & 1.655 & 1.481 \\
  \midrule
      \multicolumn{10}{l}{\textit{Panel C: Nonlinear (Exponential) Structural Equation with Mixed Covariates}} \\
\addlinespace[0.2em]
\multirow{6}{*}{$100$}
 & \multirow{2}{*}{$3$} & Bias & 7.163 & 0.083 & 0.369 & -0.655 & -1.030 & -0.673 & 0.159 \\
      &                      & RMSE & 7.869 & 1.587 & 1.874 & 2.477 & 7.613 & 4.321 & 1.966 \\
      \addlinespace[0.3em]
      & \multirow{2}{*}{$5$} & Bias & 9.339 & 0.194 & 0.851 & -1.075 & -2.370 & -1.116 & 0.438 \\
      &                      & RMSE & 10.442 & 2.091 & 2.339 & 5.672 & 13.971 & 7.582 & 2.358 \\
      \addlinespace[0.3em]
      & \multirow{2}{*}{$10$} & Bias & 8.165 & 0.172 & 1.634 & -5.669 & -29.150 & -1.400 & 1.228 \\
      &                      & RMSE & 9.921 & 1.898 & 2.563 & 26.265 & 90.265 & 9.247 & 2.093 \\
\addlinespace[0.5em]
\multirow{6}{*}{$200$}
  & \multirow{2}{*}{$3$} & Bias & 7.309 & 0.025 & 0.131 & -0.291 & -0.326 & -0.283 & 0.050 \\
      &                      & RMSE & 7.682 & 1.060 & 1.286 & 1.574 & 1.538 & 1.444 & 1.440 \\
      \addlinespace[0.3em]
      & \multirow{2}{*}{$5$} & Bias & 9.477 & 0.127 & 0.519 & -0.281 & -0.384 & -0.217 & 0.296 \\
      &                      & RMSE & 10.048 & 1.338 & 1.508 & 1.861 & 2.239 & 1.746 & 1.555 \\
      \addlinespace[0.3em]
      & \multirow{2}{*}{$10$} & Bias & 8.437 & 0.057 & 0.870 & -0.575 & -5.822 & -0.345 & 0.726 \\
      &                      & RMSE & 9.397 & 1.272 & 1.552 & 2.854 & 28.098 & 1.871 & 1.526 \\
\bottomrule
  \end{tabular*}
  
  \par
  \vspace{0.5em}
  \begin{minipage}{\textwidth}
    \footnotesize
    \textit{Note:} $^{\dagger}$ All values are multiplied by $100$ for readability. 
  \end{minipage}
\end{table}
\FloatBarrier
\subsection{Details of the DML-OS Benchmark}
\label{subsec:dml-details}

This section describes the implementation of the DML-OS benchmark used in the simulation study. The benchmark is motivated by the orthogonal score and cross-fitting ideas of \citet{chernozhukov2018double} but is adapted to the IV designs considered in this paper. In both designs, nuisance functions are learned strictly out-of-fold using $K=5$ cross-fitting with random forests, and the structural parameter is then estimated from the pooled cross-fitted score.

\subsubsection{Design I: Linear Structural Equation}

\paragraph{Model specification.}
The first design is
\begin{align}
    Y &= \theta_0 X + \epsilon, \qquad \E[\epsilon \mid Z] = 0, \\
    X &= f_0(Z) + u, \qquad \E[u \mid Z] = 0,
\end{align}
where $X$ is endogenous, and $Z$ is a possibly high-dimensional vector of instruments.

\paragraph{Score construction.}
For this design, we use the score
$$
\psi(W;\theta,\eta)=(Y-\theta X)\eta(Z),
$$
with nuisance function $\eta_0(Z)=\E[X\mid Z]$. Orthogonality follows from
$$
\partial_{\eta}\E_P[\psi(W;\theta_0,\eta_0)][\eta-\eta_0]
=
\E_P[\epsilon\{\eta(Z)-\eta_0(Z)\}]
=
0,
$$
which holds by iterated expectations and the condition $\E[\epsilon\mid Z]=0$.

\paragraph{Estimation algorithm.}
The implementation proceeds as follows.
\begin{enumerate}
    \item The sample of size $N$ is randomly partitioned into $K=5$ folds, denoted by $\{I_k\}_{k=1}^K$.
    
    \item For each fold $k$, the nuisance function $\eta_0(Z)=\E[X\mid Z]$ is estimated using only the auxiliary sample $I_k^c=\{1,\dots,N\}\setminus I_k$. Specifically, we regress $X$ on $Z$ using a random forest with $1000$ trees, minimum node size $5$, and $mtry=\max\{\lfloor p/3\rfloor,1\}$.
    
    \item The fitted random forest is then evaluated on the held-out fold $I_k$ to obtain cross-fitted predictions $\widehat{\eta}_i$ for $i\in I_k$. Repeating this step over all folds yields a full set of cross-fitted nuisance estimates $\{\widehat{\eta}_i\}_{i=1}^N$.
    
    \item The structural parameter is estimated from the pooled score equation
    $$
    \frac{1}{N}\sum_{i=1}^N (Y_i-\theta X_i)\widehat{\eta}_i=0.
    $$
    In this linear case, the solution is available in closed form:
    $$
    \widehat{\theta}
    =
    \left(\sum_{i=1}^N \widehat{\eta}_i X_i\right)^{-1}
    \sum_{i=1}^N \widehat{\eta}_i Y_i.
    $$
\end{enumerate}

\subsubsection{Design II: Nonlinear Structural Equation}

\paragraph{Model specification.}
The second design is
$$
Y=\exp(1+\theta_0 X)+\epsilon, \qquad \E[\epsilon\mid Z]=0,
$$
so that
$$
m(W,\theta)=Y-\exp(1+\theta X),
\qquad
\E[m(W,\theta_0)\mid Z]=0.
$$

\paragraph{Score construction.}
We use the score
$$
\psi(W;\theta,\mu)=m(W,\theta)\mu(Z),
$$
where $\mu(\cdot)$ is a nuisance function learned from the data. Motivated by the score-gradient construction, we approximate $\mu_0(Z)$ by the conditional expectation of $X\exp(1+\theta_0X)$, up to an irrelevant sign normalization. Since this object depends on the unknown parameter value, we estimate it using a fold-specific preliminary estimator.

\paragraph{Estimation algorithm.}
The implementation is fully cross-fitted and proceeds as follows.
\begin{enumerate}
    \item As in Design I, the sample is randomly partitioned into $K=5$ folds $\{I_k\}_{k=1}^K$.
    
    \item For each fold $k$, we use only the auxiliary sample $I_k^c$ to compute a preliminary estimator $\bar{\theta}^{(-k)}$. Our primary choice is a sieve-GMM estimator based on polynomial sieve moments. If this step fails numerically, we use a nonlinear least squares estimator computed on the same auxiliary sample; if that also fails, we revert to the initial value.
    
    \item Using the same auxiliary sample $I_k^c$, we construct the fold-specific pseudo-outcome
    $$
    D_i^{*,(-k)} = X_i \exp\!\bigl(1+\bar{\theta}^{(-k)}X_i\bigr),
    \qquad i\in I_k^c.
    $$
    We then regress $D_i^{*,(-k)}$ on $Z_i$ over $I_k^c$ using a random forest with $1000$ trees, minimum node size $5$, and $mtry=\max\{\lfloor p/3\rfloor,1\}$.
    
    \item The fitted random forest is evaluated on the held-out fold $I_k$ to obtain cross-fitted nuisance predictions $\widehat{\mu}_i$ for $i\in I_k$. Repeating this step across all folds yields a full set of cross-fitted predictions $\{\widehat{\mu}_i\}_{i=1}^N$.
    
    \item The structural parameter is then estimated from the pooled score equation
    $$
    \frac{1}{N}\sum_{i=1}^N \Bigl\{Y_i-\exp(1+\theta X_i)\Bigr\}\widehat{\mu}_i = 0.
    $$
    Numerically, we first search for a root over a fixed interval. If a sign change is detected, we solve the equation using one-dimensional root-finding. Otherwise, we use the minimizer of the squared pooled score over the same interval as a numerical fallback.
\end{enumerate}

\paragraph{Implementation note.}
In both designs, nuisance functions are estimated strictly out-of-fold, so that the final score is formed entirely from cross-fitted predictions. The benchmark is included as a DML-style orthogonal score comparison, motivated by \citet{chernozhukov2018double}. We do not attempt to verify the high-level rate conditions of that paper for the present simulation designs, and the random forest tuning parameters are chosen for numerical stability rather than to establish formal rate guarantees.

\subsection{Simulation Results for Finite Sample Coverage of KMD Estimation}
\label{subsec:coverage}
In this section, we evaluate the finite sample inference performance of the KMD estimator across all model specifications discussed in the main text: \textbf{Linear Regression (LR)}, the \textbf{Box--Cox Transformation Model}, \textbf{Linear IV (LIV)}, and \textbf{Nonlinear IV (NLIV)}.

For each specification, we compute the pointwise coverage probabilities (CP) of the 95\% asymptotic confidence intervals. These intervals are constructed using the sandwich variance estimators derived from the asymptotic normality results in Theorem \ref{thm:theta_an} and Proposition \ref{prop:variance_consistency}. The following tables present the empirical coverage rates under various sample sizes ($n$), dimensions ($p$), and error structures (Homoskedastic and Heteroskedastic). The results demonstrate that the empirical coverage rates converge to the nominal 95\% level as the sample size increases, confirming the validity of the KMD inference framework.

\begin{table}[htbp]
\centering
\caption{Pointwise Coverage Percentages for Linear Regression Model}
\label{tab:cov-lr}
% 1. 进一步减小行高 (从 1.05 -> 0.95)
\renewcommand{\arraystretch}{1.05}
% 2. 减小列间距 (默认约 6pt -> 3pt)，这会让表格变窄，从而允许 resizebox 缩放时字体保留得稍微大一点，或者更紧凑
%\setlength{\tabcolsep}{3pt}
% 3. 设置基础字体为极小
\scriptsize
\resizebox{\textwidth}{!}{
\begin{tabular}{ccccccccccccc} % 修正了列数，确保为 13 列 (N, P, Error + 10个Theta)
\toprule
$n$ & $p$ & Error & $\theta_{1}$ & $\theta_{2}$ & $\theta_{3}$ & $\theta_{4}$ & $\theta_{5}$ & $\theta_{6}$ & $\theta_{7}$ & $\theta_{8}$ & $\theta_{9}$ & $\theta_{10}$ \\
\midrule
\multirow{6}{*}{100} & \multirow{2}{*}{3} & Hetero & 0.943 & 0.940 & 0.919 &  &  &  &  &  &  &  \\
 &  & Homo & 0.942 & 0.942 & 0.924 &  &  &  &  &  &  &  \\
\cmidrule{2-13}
 & \multirow{2}{*}{5} & Hetero & 0.950 & 0.935 & 0.938 & 0.944 & 0.932 &  &  &  &  &  \\
 &  & Homo & 0.944 & 0.940 & 0.939 & 0.948 & 0.935 &  &  &  &  &  \\
\cmidrule{2-13}
 & \multirow{2}{*}{10} & Hetero & 0.945 & 0.936 & 0.942 & 0.936 & 0.934 & 0.928 & 0.944 & 0.937 & 0.945 & 0.941 \\
 &  & Homo & 0.939 & 0.930 & 0.945 & 0.939 & 0.935 & 0.928 & 0.941 & 0.940 & 0.941 & 0.946 \\
\midrule
\multirow{6}{*}{200} & \multirow{2}{*}{3} & Hetero & 0.937 & 0.934 & 0.953 &  &  &  &  &  &  &  \\
 &  & Homo & 0.936 & 0.938 & 0.956 &  &  &  &  &  &  &  \\
\cmidrule{2-13}
 & \multirow{2}{*}{5} & Hetero & 0.945 & 0.951 & 0.939 & 0.942 & 0.942 &  &  &  &  &  \\
 &  & Homo & 0.947 & 0.943 & 0.941 & 0.944 & 0.936 &  &  &  &  &  \\
\cmidrule{2-13}
 & \multirow{2}{*}{10} & Hetero & 0.945 & 0.941 & 0.944 & 0.961 & 0.961 & 0.952 & 0.932 & 0.948 & 0.954 & 0.953 \\
 &  & Homo & 0.940 & 0.949 & 0.943 & 0.964 & 0.962 & 0.955 & 0.934 & 0.955 & 0.954 & 0.949 \\
\midrule
\multirow{6}{*}{400} & \multirow{2}{*}{3} & Hetero & 0.952 & 0.955 & 0.943 &  &  &  &  &  &  &  \\
 &  & Homo & 0.959 & 0.946 & 0.937 &  &  &  &  &  &  &  \\
\cmidrule{2-13}
 & \multirow{2}{*}{5} & Hetero & 0.939 & 0.955 & 0.950 & 0.942 & 0.952 &  &  &  &  &  \\
 &  & Homo & 0.932 & 0.951 & 0.948 & 0.941 & 0.951 &  &  &  &  &  \\
\cmidrule{2-13}
 & \multirow{2}{*}{10} & Hetero & 0.950 & 0.956 & 0.936 & 0.942 & 0.944 & 0.946 & 0.951 & 0.952 & 0.950 & 0.935 \\
 &  & Homo & 0.955 & 0.952 & 0.937 & 0.939 & 0.947 & 0.956 & 0.949 & 0.951 & 0.951 & 0.936 \\
\bottomrule
\end{tabular}
}
\end{table}

\begin{table}[htbp]
\centering
\caption{Pointwise Coverage Percentages for Box--Cox Model ($\lambda_0=0$)}
\label{tab:cov-bc_lambda0}
\renewcommand{\arraystretch}{1.05}
\resizebox{\textwidth}{!}{
\begin{tabular}{ccccccccccccccc}
\toprule
$n$ & $p$ & Error & $\alpha$ & $\beta_{1}$ & $\beta_{2}$ & $\beta_{3}$ & $\beta_{4}$ & $\beta_{5}$ & $\beta_{6}$ & $\beta_{7}$ & $\beta_{8}$ & $\beta_{9}$ & $\beta_{10}$ & $\lambda$ \\
\midrule
\multirow{6}{*}{100} & \multirow{2}{*}{3} & Hetero & 0.938 & 0.937 & 0.942 & 0.934 &  &  &  &  &  &  &  & 0.915 \\
 &  & Homo & 0.956 & 0.943 & 0.947 & 0.948 &  &  &  &  &  &  &  & 0.965 \\
\cmidrule{2-15}
 & \multirow{2}{*}{5} & Hetero & 0.946 & 0.923 & 0.939 & 0.924 & 0.935 & 0.944 &  &  &  &  &  & 0.914 \\
 &  & Homo & 0.942 & 0.923 & 0.933 & 0.931 & 0.923 & 0.943 &  &  &  &  &  & 0.945 \\
\cmidrule{2-15}
 & \multirow{2}{*}{10} & Hetero & 0.907 & 0.910 & 0.930 & 0.931 & 0.922 & 0.934 & 0.917 & 0.937 & 0.934 & 0.927 & 0.925 & 0.892 \\
 &  & Homo & 0.937 & 0.930 & 0.923 & 0.922 & 0.938 & 0.919 & 0.933 & 0.927 & 0.929 & 0.922 & 0.905 & 0.926 \\
\midrule
\multirow{6}{*}{200} & \multirow{2}{*}{3} & Hetero & 0.931 & 0.947 & 0.943 & 0.945 &  &  &  &  &  &  &  & 0.940 \\
 &  & Homo & 0.962 & 0.939 & 0.944 & 0.948 &  &  &  &  &  &  &  & 0.967 \\
\cmidrule{2-15}
 & \multirow{2}{*}{5} & Hetero & 0.942 & 0.940 & 0.943 & 0.950 & 0.949 & 0.933 &  &  &  &  &  & 0.931 \\
 &  & Homo & 0.962 & 0.955 & 0.945 & 0.942 & 0.956 & 0.954 &  &  &  &  &  & 0.954 \\
\cmidrule{2-15}
 & \multirow{2}{*}{10} & Hetero & 0.939 & 0.940 & 0.942 & 0.938 & 0.952 & 0.940 & 0.951 & 0.933 & 0.946 & 0.941 & 0.939 & 0.939 \\
 &  & Homo & 0.945 & 0.937 & 0.933 & 0.940 & 0.930 & 0.935 & 0.950 & 0.943 & 0.947 & 0.935 & 0.943 & 0.942 \\
\midrule
\multirow{6}{*}{400} & \multirow{2}{*}{3} & Hetero & 0.946 & 0.952 & 0.941 & 0.949 &  &  &  &  &  &  &  & 0.935 \\
 &  & Homo & 0.952 & 0.934 & 0.927 & 0.940 &  &  &  &  &  &  &  & 0.966 \\
\cmidrule{2-15}
 & \multirow{2}{*}{5} & Hetero & 0.935 & 0.950 & 0.941 & 0.952 & 0.943 & 0.936 &  &  &  &  &  & 0.930 \\
 &  & Homo & 0.950 & 0.939 & 0.943 & 0.950 & 0.950 & 0.944 &  &  &  &  &  & 0.949 \\
\cmidrule{2-15}
 & \multirow{2}{*}{10} & Hetero & 0.938 & 0.942 & 0.940 & 0.945 & 0.941 & 0.936 & 0.953 & 0.948 & 0.949 & 0.959 & 0.943 & 0.944 \\
 &  & Homo & 0.950 & 0.935 & 0.943 & 0.946 & 0.949 & 0.943 & 0.946 & 0.952 & 0.956 & 0.952 & 0.943 & 0.943 \\
\bottomrule
\end{tabular}
}
\end{table}

\begin{table}[htbp]
\centering
\caption{Pointwise Coverage Percentages for Box--Cox Model ($\lambda_0=0.5$)}
\label{tab:cov-bc_lambda05}
\renewcommand{\arraystretch}{1.05}
\resizebox{\textwidth}{!}{
\begin{tabular}{ccccccccccccccc}
\toprule
$N$ & $P$ & Error & $\alpha$ & $\beta_{1}$ & $\beta_{2}$ & $\beta_{3}$ & $\beta_{4}$ & $\beta_{5}$ & $\beta_{6}$ & $\beta_{7}$ & $\beta_{8}$ & $\beta_{9}$ & $\beta_{10}$ & $\lambda$ \\
\midrule
\multirow{6}{*}{100} & \multirow{2}{*}{3} & Hetero & 0.952 & 0.944 & 0.943 & 0.940 &  &  &  &  &  &  &  & 0.931 \\
 &  & Homo & 0.969 & 0.946 & 0.952 & 0.943 &  &  &  &  &  &  &  & 0.973 \\
\cmidrule{2-15}
 & \multirow{2}{*}{5} & Hetero & 0.960 & 0.926 & 0.941 & 0.927 & 0.938 & 0.952 &  &  &  &  &  & 0.930 \\
 &  & Homo & 0.949 & 0.926 & 0.937 & 0.932 & 0.924 & 0.949 &  &  &  &  &  & 0.954 \\
\cmidrule{2-15}
 & \multirow{2}{*}{10} & Hetero & 0.909 & 0.912 & 0.932 & 0.934 & 0.922 & 0.934 & 0.917 & 0.938 & 0.935 & 0.928 & 0.924 & 0.864 \\
 &  & Homo & 0.931 & 0.936 & 0.932 & 0.921 & 0.938 & 0.913 & 0.939 & 0.928 & 0.932 & 0.921 & 0.912 & 0.914 \\
\midrule
\multirow{6}{*}{200} & \multirow{2}{*}{3} & Hetero & 0.950 & 0.946 & 0.944 & 0.948 &  &  &  &  &  &  &  & 0.949 \\
 &  & Homo & 0.969 & 0.941 & 0.942 & 0.955 &  &  &  &  &  &  &  & 0.967 \\
\cmidrule{2-15}
 & \multirow{2}{*}{5} & Hetero & 0.947 & 0.936 & 0.947 & 0.951 & 0.944 & 0.935 &  &  &  &  &  & 0.937 \\
 &  & Homo & 0.962 & 0.956 & 0.945 & 0.942 & 0.958 & 0.956 &  &  &  &  &  & 0.951 \\
\cmidrule{2-15}
 & \multirow{2}{*}{10} & Hetero & 0.935 & 0.940 & 0.936 & 0.941 & 0.953 & 0.936 & 0.951 & 0.934 & 0.942 & 0.940 & 0.941 & 0.894 \\
 &  & Homo & 0.939 & 0.938 & 0.933 & 0.941 & 0.934 & 0.930 & 0.949 & 0.944 & 0.950 & 0.938 & 0.941 & 0.930 \\
\midrule
\multirow{6}{*}{400} & \multirow{2}{*}{3} & Hetero & 0.959 & 0.948 & 0.938 & 0.950 &  &  &  &  &  &  &  & 0.942 \\
 &  & Homo & 0.953 & 0.939 & 0.933 & 0.940 &  &  &  &  &  &  &  & 0.965 \\
\cmidrule{2-15}
 & \multirow{2}{*}{5} & Hetero & 0.937 & 0.949 & 0.943 & 0.949 & 0.943 & 0.934 &  &  &  &  &  & 0.936 \\
 &  & Homo & 0.954 & 0.944 & 0.943 & 0.953 & 0.948 & 0.944 &  &  &  &  &  & 0.947 \\
\cmidrule{2-15}
 & \multirow{2}{*}{10} & Hetero & 0.931 & 0.943 & 0.941 & 0.943 & 0.944 & 0.941 & 0.951 & 0.953 & 0.947 & 0.960 & 0.941 & 0.937 \\
 &  & Homo & 0.952 & 0.933 & 0.942 & 0.943 & 0.946 & 0.942 & 0.942 & 0.951 & 0.955 & 0.949 & 0.943 & 0.941 \\
\bottomrule
\end{tabular}
}
\end{table}

\begin{table}[htbp]
\centering
\caption{Pointwise Coverage Probabilities for Linear, Nonlinear, and Mixed IV Models}
\label{tab:cov-iv-comparison}
\small
\renewcommand{\arraystretch}{1.08}
\setlength{\tabcolsep}{5.5pt}
\begin{tabular}{ccccc ccc}
\toprule
\multirow{2}{*}{$n$} 
& \multirow{2}{*}{Dim.} 
& \multicolumn{3}{c}{Homoskedastic Errors} 
& \multicolumn{3}{c}{Heteroskedastic Errors} \\
\cmidrule(lr){3-5} \cmidrule(lr){6-8}
& 
& Linear IV 
& Nonlinear IV 
& Mixed NLIV
& Linear IV 
& Nonlinear IV 
& Mixed NLIV \\
\midrule
\multirow{3}{*}{100} 
& 3  & 0.950 & 0.938 & 0.937 & 0.954 & 0.943 & 0.946 \\
& 5  & 0.953 & 0.931 & 0.929 & 0.951 & 0.933 & 0.937 \\
& 10 & 0.953 & 0.909 & 0.920 & 0.948 & 0.929 & 0.931 \\
\midrule
\multirow{3}{*}{200} 
& 3  & 0.956 & 0.940 & 0.937 & 0.934 & 0.954 & 0.950 \\
& 5  & 0.948 & 0.940 & 0.934 & 0.948 & 0.940 & 0.937 \\
& 10 & 0.947 & 0.929 & 0.936 & 0.956 & 0.941 & 0.934 \\
\midrule
\multirow{3}{*}{400} 
& 3  & 0.955 & 0.956 & 0.960 & 0.963 & 0.935 & 0.948 \\
& 5  & 0.942 & 0.949 & 0.952 & 0.947 & 0.957 & 0.960 \\
& 10 & 0.946 & 0.951 & 0.953 & 0.952 & 0.950 & 0.949 \\
\bottomrule
\end{tabular}
\end{table}
\FloatBarrier

\subsection{Simulation Results for the Test of \cite{dominguez2015simple}}
\label{subsec:dl-test}
To benchmark the finite-sample performance of the proposed KMD test, particularly its robustness to the dimensionality of the conditioning variables, this section presents a comparative analysis with the omnibus specification test of \cite{dominguez2015simple}. We implement the DL test following their exact protocol, utilizing the Consistent Method of Moments (CMM) estimator and the projection-based wild bootstrap procedure. While the DL test is consistent against fixed alternatives, its construction using coordinate-wise indicator functions makes it inherently susceptible to data sparsity as the dimension of the covariate space increases, a limitation that the reproducing-kernel-based KMD aims to mitigate.
\begin{table}[htbp]
  \centering
  \caption{Empirical Rejection Rates for Linear Regression Specification Tests ($\alpha=0.05$)}
  \label{tab:DL_lr_test}
  
  \small % 使用小号字体
  \renewcommand{\arraystretch}{1.15} 
  
  % 定义等宽居中列 Y
  \newcolumntype{Y}{>{\centering\arraybackslash}X}
  
  % 使用 tabularx，前两列 c (自然宽度居中)，后六列 Y (等宽居中)
  \begin{tabularx}{\textwidth}{cc YYY YYY}
    \toprule
    & & \multicolumn{3}{c}{Homoskedastic Errors} & \multicolumn{3}{c}{Heteroskedastic Errors} \\
    \cmidrule(lr){3-5} \cmidrule(lr){6-8}
    $n$ & DGP Model & $p=3$ & $p=5$ & $p=10$ & $p=3$ & $p=5$ & $p=10$ \\
    \midrule
    
    % --- n=100 ---
    \multirow{4}{*}{$100$} 
      & Size ($\mathbb{H}_0$)             & $0.050$ & $0.033$ & $0.014$ & $0.042$ & $0.031$ & $0.014$ \\
      & Power ($\mathbb{H}_{1a}$: Quad.)  & $0.747$ & $0.483$ & $0.012$ & $0.743$ & $0.489$ & $0.012$ \\
      & Power ($\mathbb{H}_{1b}$: Disc.)  & $0.488$ & $0.268$ & $0.015$ & $0.484$ & $0.267$ & $0.018$ \\
      & Power ($\mathbb{H}_{1c}$: Inter.) & $0.157$ & $0.101$ & $0.006$ & $0.154$ & $0.093$ & $0.008$ \\
      
    \addlinespace[0.5em] % 样本量之间增加间距
    
    % --- n=200 ---
    \multirow{4}{*}{$200$} 
      & Size ($\mathbb{H}_0$)             & $0.039$ & $0.043$ & $0.025$ & $0.037$ & $0.044$ & $0.021$ \\
      & Power ($\mathbb{H}_{1a}$: Quad.)  & $0.990$ & $0.926$ & $0.061$ & $0.988$ & $0.933$ & $0.062$ \\
      & Power ($\mathbb{H}_{1b}$: Disc.)  & $0.852$ & $0.681$ & $0.047$ & $0.853$ & $0.700$ & $0.043$ \\
      & Power ($\mathbb{H}_{1c}$: Inter.) & $0.401$ & $0.328$ & $0.030$ & $0.384$ & $0.348$ & $0.030$ \\

    \addlinespace[0.5em]
    
    % --- n=400 ---
    \multirow{4}{*}{$400$} 
      & Size ($\mathbb{H}_0$)             & $0.039$ & $0.044$ & $0.028$ & $0.041$ & $0.051$ & $0.025$ \\
      & Power ($\mathbb{H}_{1a}$: Quad.)  & $1.000$ & $1.000$ & $0.460$ & $1.000$ & $1.000$ & $0.478$ \\
      & Power ($\mathbb{H}_{1b}$: Disc.)  & $0.990$ & $0.955$ & $0.249$ & $0.989$ & $0.965$ & $0.287$ \\
      & Power ($\mathbb{H}_{1c}$: Inter.) & $0.836$ & $0.740$ & $0.135$ & $0.836$ & $0.770$ & $0.160$ \\
    \bottomrule
  \end{tabularx}
\end{table}

\begin{table}[htbp]
  \centering
  \caption{Empirical Rejection Rates for Box--Cox Specification Tests ($\alpha=0.05$)}
  \label{tab:DL_boxcox_test}
  
  \small % 使用小号字体
  \renewcommand{\arraystretch}{1.15} 
  
  % 定义等宽居中列 Y
  \newcolumntype{Y}{>{\centering\arraybackslash}X}
  
  % 使用 tabularx，前两列 c (自然宽度居中)，后六列 Y (等宽居中)
  \begin{tabularx}{\textwidth}{cc YYY YYY}
    \toprule
    & & \multicolumn{3}{c}{Homoskedastic Errors} & \multicolumn{3}{c}{Heteroskedastic Errors} \\
    \cmidrule(lr){3-5} \cmidrule(lr){6-8}
    $n$ & DGP Model & $p=3$ & $p=5$ & $p=10$ & $p=3$ & $p=5$ & $p=10$ \\
    \midrule
    
    % --- n=100 ---
    \multirow{4}{*}{$100$} 
      & Size ($\mathbb{H}_0$)             & $0.044$ & $0.026$ & $0.016$ & $0.042$ & $0.041$ & $0.005$ \\
      & Power ($\mathbb{H}_{1a}$: Quad.)  & $0.339$ & $0.179$ & $0.015$ & $0.388$ & $0.225$ & $0.019$ \\
      & Power ($\mathbb{H}_{1b}$: Disc.)  & $0.945$ & $0.495$ & $0.033$ & $0.948$ & $0.526$ & $0.030$ \\
      & Power ($\mathbb{H}_{1c}$: Inter.) & $0.241$ & $0.062$ & $0.008$ & $0.270$ & $0.088$ & $0.008$ \\
      
    \addlinespace[0.5em]
    
    % --- n=200 ---
    \multirow{4}{*}{$200$} 
      & Size ($\mathbb{H}_0$)             & $0.039$ & $0.029$ & $0.021$ & $0.044$ & $0.037$ & $0.016$ \\
      & Power ($\mathbb{H}_{1a}$: Quad.)  & $0.680$ & $0.569$ & $0.071$ & $0.763$ & $0.684$ & $0.078$ \\
      & Power ($\mathbb{H}_{1b}$: Disc.)  & $1.000$ & $0.969$ & $0.169$ & $1.000$ & $0.977$ & $0.168$ \\
      & Power ($\mathbb{H}_{1c}$: Inter.) & $0.561$ & $0.243$ & $0.027$ & $0.668$ & $0.300$ & $0.025$ \\

    \addlinespace[0.5em]
    
    % --- n=400 ---
    \multirow{4}{*}{$400$} 
      & Size ($\mathbb{H}_0$)             & $0.037$ & $0.045$ & $0.038$ & $0.045$ & $0.036$ & $0.028$ \\
      & Power ($\mathbb{H}_{1a}$: Quad.)  & $0.941$ & $0.959$ & $0.405$ & $0.986$ & $0.988$ & $0.463$ \\
      & Power ($\mathbb{H}_{1b}$: Disc.)  & $1.000$ & $1.000$ & $0.726$ & $1.000$ & $1.000$ & $0.738$ \\
      & Power ($\mathbb{H}_{1c}$: Inter.) & $0.909$ & $0.671$ & $0.131$ & $0.979$ & $0.793$ & $0.188$ \\
      
    \bottomrule
  \end{tabularx}
\end{table}

\begin{table}[htbp]
  \centering
  \caption{Empirical Rejection Rates for IV Specification Tests ($\alpha=0.05$)}
  \label{tab:DL_iv_test}
  
  \small
  \renewcommand{\arraystretch}{1.15}
  \newcolumntype{Y}{>{\centering\arraybackslash}X}
  
  \begin{tabularx}{\textwidth}{cc YYY YYY}
    \toprule
    & & \multicolumn{3}{c}{Homoskedastic Errors} 
      & \multicolumn{3}{c}{Heteroskedastic Errors} \\
    \cmidrule(lr){3-5} \cmidrule(lr){6-8}
    $n$ & DGP 
      & $d=3$ & $d=5$ & $d=10$
      & $d=3$ & $d=5$ & $d=10$ \\
    \midrule
    
    \multicolumn{8}{l}{\textit{Panel A: Linear structural model with continuous instruments}} \\

    \multirow{2}{*}{$100$}
      & Size ($\mathbb H_0$)  
        & $0.027$ & $0.025$ & $0.000$
        & $0.038$ & $0.019$ & $0.001$ \\
      & Power ($\mathbb H_1$) 
        & $0.201$ & $0.068$ & $0.001$
        & $0.236$ & $0.069$ & $0.002$ \\
    \addlinespace[0.3em]
    
    \multirow{2}{*}{$200$}
      & Size ($\mathbb H_0$)  
        & $0.043$ & $0.021$ & $0.002$
        & $0.052$ & $0.027$ & $0.003$ \\
      & Power ($\mathbb H_1$) 
        & $0.448$ & $0.159$ & $0.007$
        & $0.519$ & $0.169$ & $0.001$ \\
    \addlinespace[0.3em]
    
    \multirow{2}{*}{$400$}
      & Size ($\mathbb H_0$)  
        & $0.050$ & $0.035$ & $0.009$
        & $0.039$ & $0.028$ & $0.007$ \\
      & Power ($\mathbb H_1$) 
        & $0.785$ & $0.434$ & $0.040$
        & $0.869$ & $0.483$ & $0.025$ \\
    
    \addlinespace[0.7em]
    \multicolumn{8}{l}{\textit{Panel B: Nonlinear structural model with continuous instruments}} \\
        \multirow{2}{*}{$100$}
      & Size ($\mathbb H_0$)  
        & $0.034$ & $0.018$ & $0.003$
        & $0.029$ & $0.035$ & $0.003$ \\
      & Power ($\mathbb H_1$) 
        & $0.583$ & $0.108$ & $0.001$
        & $0.628$ & $0.143$ & $0.002$ \\
    \addlinespace[0.3em]
    
    \multirow{2}{*}{$200$}
      & Size ($\mathbb H_0$)  
        & $0.045$ & $0.024$ & $0.009$
        & $0.043$ & $0.019$ & $0.013$ \\
      & Power ($\mathbb H_1$) 
        & $0.959$ & $0.513$ & $0.006$
        & $0.965$ & $0.557$ & $0.001$ \\
    \addlinespace[0.3em]
    
    \multirow{2}{*}{$400$}
      & Size ($\mathbb H_0$)  
        & $0.049$ & $0.030$ & $0.041$
        & $0.047$ & $0.021$ & $0.040$ \\
      & Power ($\mathbb H_1$) 
        & $1.000$ & $0.965$ & $0.024$
        & $1.000$ & $0.973$ & $0.014$ \\
    
    \addlinespace[0.7em]
    \multicolumn{8}{l}{\textit{Panel C: Mixed nonlinear structural model with continuous and discrete instruments}} \\
    \multirow{2}{*}{$100$}
      & Size ($\mathbb H_0$)  
        & $0.037$ & $0.040$ & $0.018$
        & $0.031$ & $0.027$ & $0.026$ \\
      & Power ($\mathbb H_1$) 
       & $0.357$ & $0.232$ & $0.051$
        & $0.521$ & $0.318$ & $0.087$ \\
    \addlinespace[0.3em]
    
    \multirow{2}{*}{$200$}
      & Size ($\mathbb H_0$)  
        & $0.056$ & $0.048$ & $0.026$
        & $0.043$ & $0.040$ & $0.036$ \\
      & Power ($\mathbb H_1$) 
        & $0.724$ & $0.606$ & $0.133$
        & $0.900$ & $0.764$ & $0.211$ \\
    \addlinespace[0.3em]
    
    \multirow{2}{*}{$400$}
      & Size ($\mathbb H_0$)  
        & $0.046$ & $0.054$ & $0.024$
        & $0.063$ & $0.038$ & $0.035$ \\
      & Power ($\mathbb H_1$) 
        & $0.975$ & $0.950$ & $0.417$
        & $0.998$ & $0.989$ & $0.609$ \\
    
    \bottomrule
  \end{tabularx}
\end{table}
\FloatBarrier

\section{Additional Empirical Results}
\label{sec:add-empirical}

In the primary empirical analysis presented in Section \ref{sec:real}, the estimated equivalence scale parameter $\hat{\theta}$ was found to be statistically insignificant from zero. To ensure that our main conclusions regarding the functional form of the Engel curves are not driven by the uncertainty associated with this parameter, we conduct a robustness check by imposing the restriction $\theta = 0$. Under this restriction, the effective total expenditure simplifies to $u_i = \ln x_i$.

Regarding specification tests, the rejection of the linear form remains valid by construction, as $\theta$ is unidentified in that setting. For the quadratic specification, the restricted KMD test yields a $p$-value of $0.783$, indicating that the model fits the data well even without the equivalence scale parameter.

Tables \ref{tab:quaids_estimates_restricted} and \ref{tab:wald_tests_restricted}, together with Figure \ref{fig:quaid_results_restricted}, present the results for the restricted model. The structural estimates, Wald test results, and Engel curve shapes are virtually indistinguishable from those in the unrestricted case reported in the main text. These findings confirm that our empirical analysis is robust to excluding the equivalence scale parameter.

\begin{table}[htbp]
  \centering
  \caption{Parameter Estimates for the Restricted Quadratic Specification ($\theta=0$)}
  \label{tab:quaids_estimates_restricted}
  
  \small
  \renewcommand{\arraystretch}{1.25} % 适当增加行高
  
  % --- 表格设置 ---
  \setlength{\tabcolsep}{0pt} 
  
  \begin{tabular*}{\textwidth}{
    @{\extracolsep{\fill}} 
    l 
    c 
    c 
    c 
    c 
    }
    \toprule
    Commodity & $\alpha_l$ & $\beta_l$ & $\lambda_l$ & $\gamma_l$ \\
    \midrule
    Food      & $-0.144$ $(0.632)$ & $\phantom{-}0.193$ $(0.228)$ & $-0.025$ $(0.021)$ & $\phantom{-}0.051$ $(0.004)$ \\
    Catering  & $-0.264$ $(0.441)$ & $\phantom{-}0.120$ $(0.160)$ & $-0.010$ $(0.015)$ & $-0.005$ $(0.003)$ \\
    Alcohol   & $-0.555$ $(0.460)$ & $\phantom{-}0.245$ $(0.166)$ & $-0.024$ $(0.015)$ & $-0.023$ $(0.004)$ \\
    Fuel      & $\phantom{-}0.983$ $(0.262)$ & $-0.303$ $(0.095)$ & $\phantom{-}0.024$ $(0.009)$ & $\phantom{-}0.011$ $(0.002)$ \\
    Motor     & $-2.056$ $(0.732)$ & $\phantom{-}0.843$ $(0.265)$ & $-0.080$ $(0.024)$ & $-0.019$ $(0.005)$ \\
    Fares     & $\phantom{-}0.809$ $(0.378)$ & $-0.303$ $(0.138)$ & $\phantom{-}0.029$ $(0.013)$ & $-0.007$ $(0.003)$ \\
    Leisure   & $\phantom{-}1.174$ $(1.028)$ & $-0.527$ $(0.378)$ & $\phantom{-}0.062$ $(0.035)$ & $-0.017$ $(0.007)$ \\
    \bottomrule
  \end{tabular*}
  
  \par
  \vspace{0.5em}
  \begin{minipage}{\textwidth}
    \footnotesize
    \textit{Notes:} The table reports the commodity-specific coefficients estimated under the restriction $\theta=0$. Standard errors are in parentheses. The parameters correspond to the specification $w_{il} = \alpha_l + \beta_l u_i + \lambda_l u_i^2 + \gamma_l k_i + \epsilon_{il}$, where $u_i = \ln x_i$.
  \end{minipage}
\end{table}

\begin{table}[htbp]
  \centering
  \caption{Wald Tests for Joint Significance of Expenditure Terms under $\theta=0$ ($\mathbb{H}_0: \beta_l = \lambda_l = 0$)}
  \label{tab:wald_tests_restricted}
  \small
  \renewcommand{\arraystretch}{1.2}
  \setlength{\tabcolsep}{0pt}
  
  \begin{tabular*}{\textwidth}{@{\extracolsep{\fill}}l ccccccc}
    \toprule
    & Food & {Catering} & {Alcohol} & {Fuel} & {Motor} & {Fares} & {Leisure} \\
    \midrule
    $p$-value & $0.000^{***}$ & $0.449$ & $0.002^{**}$ & $0.000^{***}$ & $0.000^{***}$ & $0.002^{**}$ & $0.000^{***}$ \\
    \bottomrule
  \end{tabular*}
  \par
  \vspace{0.5em}
  \begin{minipage}{\textwidth}
    \footnotesize
    \textit{Note:} Significance levels: $^{*} p<0.05$, $^{**} p<0.01$, $^{***} p<0.001$.
  \end{minipage}
\end{table}

\begin{figure*}[htbp]
    \centering
    \includegraphics[width=1.0\textwidth]{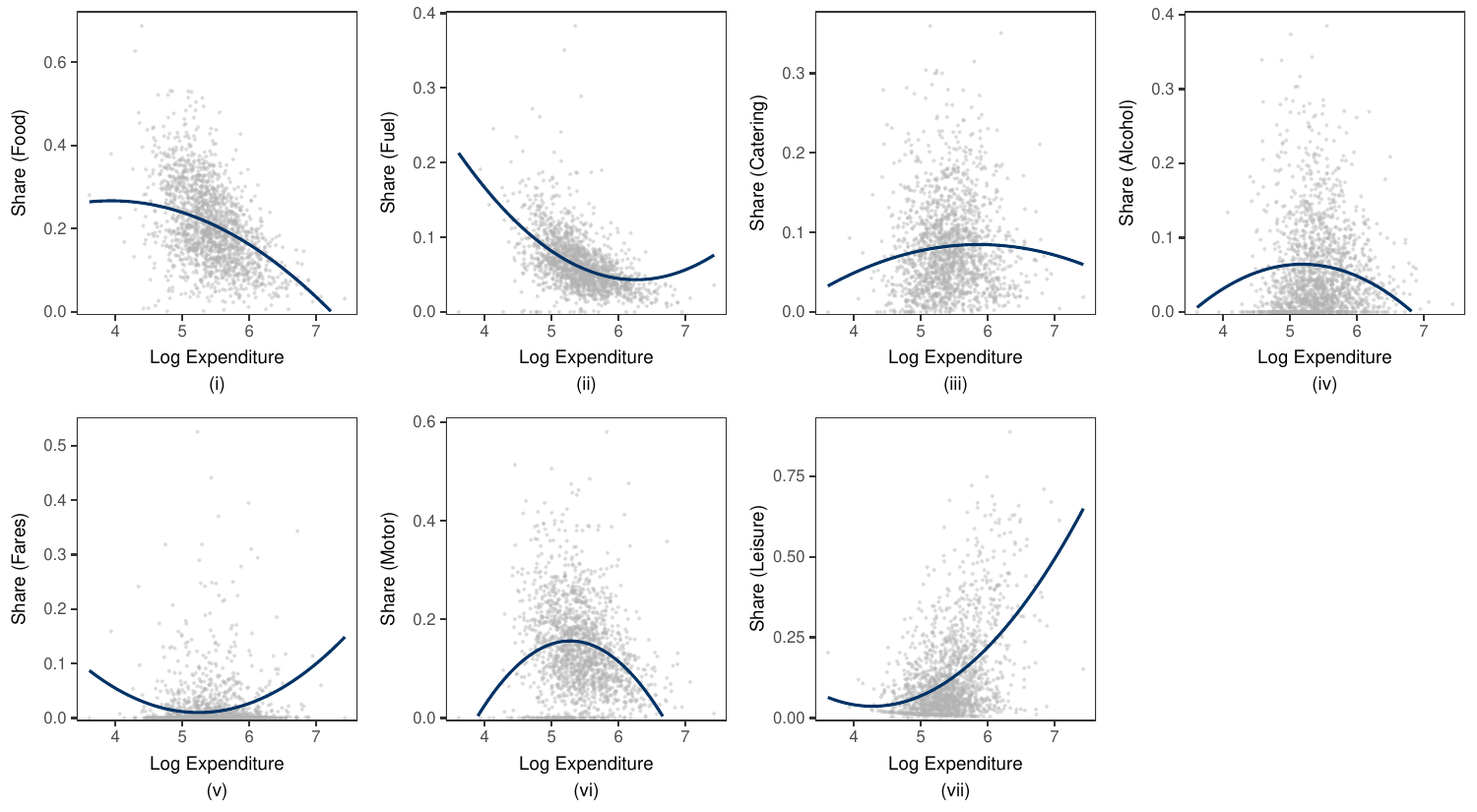}
    \caption{Estimated QUAIDS Engel curves for seven commodity groups under the restriction $\theta=0$. Grey dots represent observed household shares, while solid blue-black lines indicate the fitted quadratic specification. Sub-labels (i) through (vii) refer to food, fuel, catering, alcohol, fares, motor, and leisure, respectively.}
    \label{fig:quaid_results_restricted}
\end{figure*}

\clearpage

\putbib 
\end{bibunit}


\begin{thebibliography}{35}
\providecommand{\natexlab}[1]{#1}
\providecommand{\url}[1]{\texttt{#1}}
\expandafter\ifx\csname urlstyle\endcsname\relax
  \providecommand{\doi}[1]{doi: #1}\else
  \providecommand{\doi}{doi: \begingroup \urlstyle{rm}\Url}\fi

\bibitem[Amemiya and Powell(1981)]{amemiya1981comparison}
Takeshi Amemiya and James~L Powell.
\newblock A comparison of the box-cox maximum likelihood estimator and the non-linear two-stage least squares estimator.
\newblock \emph{Journal of Econometrics}, 17\penalty0 (3):\penalty0 351--381, 1981.

\bibitem[Banks et~al.(1997)Banks, Blundell, and Lewbel]{banks1997quadratic}
James Banks, Richard Blundell, and Arthur Lewbel.
\newblock Quadratic engel curves and consumer demand.
\newblock \emph{Review of Economics and Statistics}, 79\penalty0 (4):\penalty0 527--539, 1997.

\bibitem[Bickel and Doksum(1981)]{bickel1981analysis}
Peter~J Bickel and Kjell~A Doksum.
\newblock An analysis of transformations revisited.
\newblock \emph{Journal of the American Statistical Association}, 76\penalty0 (374):\penalty0 296--311, 1981.

\bibitem[Bierens(1982)]{bierens1982consistent}
Herman~J Bierens.
\newblock Consistent model specification tests.
\newblock \emph{Journal of Econometrics}, 20\penalty0 (1):\penalty0 105--134, 1982.

\bibitem[Bierens(1990)]{bierens1990consistent}
Herman~J Bierens.
\newblock A consistent conditional moment test of functional form.
\newblock \emph{Econometrica}, 58\penalty0 (6):\penalty0 1443--1458, 1990.

\bibitem[Bierens and Ploberger(1997)]{bierens1997asymptotic}
Herman~J Bierens and Werner Ploberger.
\newblock Asymptotic theory of integrated conditional moment tests.
\newblock \emph{Econometrica}, 65\penalty0 (5):\penalty0 1129--1151, 1997.

\bibitem[Blundell et~al.(2007)Blundell, Chen, and Kristensen]{blundell2007semi}
Richard Blundell, Xiaohong Chen, and Dennis Kristensen.
\newblock Semi-nonparametric iv estimation of shape-invariant engel curves.
\newblock \emph{Econometrica}, 75\penalty0 (6):\penalty0 1613--1669, 2007.

\bibitem[Carrasco and Florens(2000)]{carrasco2000generalization}
Marine Carrasco and Jean-Pierre Florens.
\newblock Generalization of gmm to a continuum of moment conditions.
\newblock \emph{Econometric Theory}, 16\penalty0 (6):\penalty0 797--834, 2000.

\bibitem[Chernozhukov et~al.(2018)Chernozhukov, Chetverikov, Demirer, Duflo, Hansen, Newey, and Robins]{chernozhukov2018double}
Victor Chernozhukov, Denis Chetverikov, Mert Demirer, Esther Duflo, Christian Hansen, Whitney Newey, and James Robins.
\newblock Double/debiased machine learning for treatment and structural parameters.
\newblock \emph{The Econometrics Journal}, 21\penalty0 (1):\penalty0 C1--C68, 2018.

\bibitem[Delgado et~al.(2006)Delgado, Dom{\'\i}nguez, and Lavergne]{delgado2006consistent}
Miguel~A Delgado, Manuel~A Dom{\'\i}nguez, and Pascal Lavergne.
\newblock Consistent tests of conditional moment restrictions.
\newblock \emph{Annales d'{\'E}conomie et de Statistique}, 1\penalty0 (81):\penalty0 33--67, 2006.

\bibitem[Dom{\'\i}nguez and Lobato(2004)]{dominguez2004consistent}
Manuel~A Dom{\'\i}nguez and Ignacio~N Lobato.
\newblock Consistent estimation of models defined by conditional moment restrictions.
\newblock \emph{Econometrica}, 72\penalty0 (5):\penalty0 1601--1615, 2004.

\bibitem[Dom{\'\i}nguez and Lobato(2015)]{dominguez2015simple}
Manuel~A Dom{\'\i}nguez and Ignacio~N Lobato.
\newblock A simple omnibus overidentification specification test for time series econometric models.
\newblock \emph{Econometric Theory}, 31\penalty0 (4):\penalty0 891--910, 2015.

\bibitem[Donald et~al.(2003)Donald, Imbens, and Newey]{donald2003empirical}
Stephen~G Donald, Guido~W Imbens, and Whitney~K Newey.
\newblock Empirical likelihood estimation and consistent tests with conditional moment restrictions.
\newblock \emph{Journal of Econometrics}, 117\penalty0 (1):\penalty0 55--93, 2003.

\bibitem[Donald et~al.(2009)Donald, Imbens, and Newey]{donald2009choosing}
Stephen~G Donald, Guido~W Imbens, and Whitney~K Newey.
\newblock Choosing instrumental variables in conditional moment restriction models.
\newblock \emph{Journal of Econometrics}, 152\penalty0 (1):\penalty0 28--36, 2009.

\bibitem[Escanciano(2006)]{escanciano2006consistent}
Juan~Carlos Escanciano.
\newblock A consistent diagnostic test for regression models using projections.
\newblock \emph{Econometric Theory}, 22\penalty0 (6):\penalty0 1030--1051, 2006.

\bibitem[Escanciano(2024)]{escanciano2024gaussian}
Juan~Carlos Escanciano.
\newblock A gaussian process approach to model checks.
\newblock \emph{The Annals of Statistics}, 52\penalty0 (5):\penalty0 2456--2481, 2024.

\bibitem[Gretton et~al.(2006)Gretton, Borgwardt, Rasch, Sch{\"o}lkopf, and Smola]{gretton2006kernel}
Arthur Gretton, Karsten~M. Borgwardt, Malte Rasch, Bernhard Sch{\"o}lkopf, and Alex~J. Smola.
\newblock A kernel method for the two-sample-problem.
\newblock In \emph{Advances in Neural Information Processing Systems}, volume~19, pages 513--520, 2006.

\bibitem[Gretton et~al.(2012)Gretton, Borgwardt, Rasch, Sch{\"o}lkopf, and Smola]{gretton2012kernel}
Arthur Gretton, Karsten~M Borgwardt, Malte~J Rasch, Bernhard Sch{\"o}lkopf, and Alexander Smola.
\newblock A kernel two-sample test.
\newblock \emph{Journal of Machine Learning Research}, 13\penalty0 (1):\penalty0 723--773, 2012.

\bibitem[Hansen(1982)]{hansen1982large}
Lars~Peter Hansen.
\newblock Large sample properties of generalized method of moments estimators.
\newblock \emph{Econometrica}, 50\penalty0 (4):\penalty0 1029--1054, 1982.

\bibitem[Kitamura et~al.(2004)Kitamura, Tripathi, and Ahn]{kitamura2004empirical}
Yuichi Kitamura, Gautam Tripathi, and Hyungtaik Ahn.
\newblock Empirical likelihood-based inference in conditional moment restriction models.
\newblock \emph{Econometrica}, 72\penalty0 (6):\penalty0 1667--1714, 2004.

\bibitem[Lavergne and Patilea(2013)]{lavergne2013smooth}
Pascal Lavergne and Valentin Patilea.
\newblock Smooth minimum distance estimation and testing with conditional estimating equations: uniform in bandwidth theory.
\newblock \emph{Journal of Econometrics}, 177\penalty0 (1):\penalty0 47--59, 2013.

\bibitem[Leser(1963)]{leser1963forms}
Conrad Emanuel~Victor Leser.
\newblock Forms of engel functions.
\newblock \emph{Econometrica}, 31\penalty0 (4):\penalty0 694--703, 1963.

\bibitem[Mammen(1993)]{mammen1993bootstrap}
Enno Mammen.
\newblock Bootstrap and wild bootstrap for high dimensional linear models.
\newblock \emph{The Annals of Statistics}, 21\penalty0 (1):\penalty0 255--285, 1993.

\bibitem[Muandet et~al.(2017)Muandet, Fukumizu, Sriperumbudur, Sch{\"o}lkopf, et~al.]{muandet2017kernel}
Krikamol Muandet, Kenji Fukumizu, Bharath Sriperumbudur, Bernhard Sch{\"o}lkopf, et~al.
\newblock Kernel mean embedding of distributions: A review and beyond.
\newblock \emph{Foundations and Trends{\textregistered} in Machine Learning}, 10\penalty0 (1-2):\penalty0 1--141, 2017.

\bibitem[Muandet et~al.(2020)Muandet, Jitkrittum, and K{\"u}bler]{muandet2020kernel}
Krikamol Muandet, Wittawat Jitkrittum, and Jonas K{\"u}bler.
\newblock Kernel conditional moment test via maximum moment restriction.
\newblock In \emph{Conference on Uncertainty in Artificial Intelligence}, volume 124, pages 41--50. PMLR, 2020.

\bibitem[Newey(1985)]{newey1985generalized}
Whitney~K Newey.
\newblock Generalized method of moments specification testing.
\newblock \emph{Journal of Econometrics}, 29\penalty0 (3):\penalty0 229--256, 1985.

\bibitem[Newey(1990)]{newey1990efficient}
Whitney~K Newey.
\newblock Efficient instrumental variables estimation of nonlinear models.
\newblock \emph{Econometrica}, 58\penalty0 (4):\penalty0 809--837, 1990.

\bibitem[Robinson(1991)]{robinson1991best}
Peter~M Robinson.
\newblock Best nonlinear three-stage least squares estimation of certain econometric models.
\newblock \emph{Econometrica}, 59\penalty0 (3):\penalty0 755--786, 1991.

\bibitem[Sancetta(2022)]{sancetta2022testing}
Alessio Sancetta.
\newblock Testing subspace restrictions in the presence of high dimensional nuisance parameters.
\newblock \emph{Electronic Journal of Statistics}, 16\penalty0 (2):\penalty0 5277--5320, 2022.

\bibitem[Sargan(1958)]{sargan1958estimation}
John~D Sargan.
\newblock The estimation of economic relationships using instrumental variables.
\newblock \emph{Econometrica}, pages 393--415, 1958.

\bibitem[Shin(2008)]{shin2008semiparametric}
Youngki Shin.
\newblock Semiparametric estimation of the box--cox transformation model.
\newblock \emph{The Econometrics Journal}, 11\penalty0 (3):\penalty0 517--537, 2008.

\bibitem[Sriperumbudur et~al.(2011)Sriperumbudur, Fukumizu, and Lanckriet]{sriperumbudur2011universality}
Bharath~K Sriperumbudur, Kenji Fukumizu, and Gert~RG Lanckriet.
\newblock Universality, characteristic kernels and rkhs embedding of measures.
\newblock \emph{Journal of Machine Learning Research}, 12\penalty0 (7):\penalty0 2389--2410, 2011.

\bibitem[Wooldridge(2010)]{wooldridge2010econometric}
Jeffrey~M Wooldridge.
\newblock \emph{Econometric analysis of cross section and panel data}.
\newblock MIT press, 2010.

\bibitem[Working(1943)]{working1943statistical}
Holbrook Working.
\newblock Statistical laws of family expenditure.
\newblock \emph{Journal of the American Statistical Association}, 38\penalty0 (221):\penalty0 43--56, 1943.

\bibitem[Zhang et~al.(2023)Zhang, Imaizumi, Sch{\"o}lkopf, and Muandet]{zhang2023instrumental}
Rui Zhang, Masaaki Imaizumi, Bernhard Sch{\"o}lkopf, and Krikamol Muandet.
\newblock Instrumental variable regression via kernel maximum moment loss.
\newblock \emph{Journal of Causal Inference}, 11\penalty0 (1):\penalty0 20220073, 2023.

\end{thebibliography}


\begin{thebibliography}{16}
\providecommand{\natexlab}[1]{#1}
\providecommand{\url}[1]{\texttt{#1}}
\expandafter\ifx\csname urlstyle\endcsname\relax
  \providecommand{\doi}[1]{doi: #1}\else
  \providecommand{\doi}{doi: \begingroup \urlstyle{rm}\Url}\fi

\bibitem[Bierens and Ploberger(1997)]{bierens1997asymptotic}
Herman~J Bierens and Werner Ploberger.
\newblock Asymptotic theory of integrated conditional moment tests.
\newblock \emph{Econometrica}, 65\penalty0 (5):\penalty0 1129--1151, 1997.

\bibitem[Billingsley(1999)]{billingsley1999convergence}
Patrick Billingsley.
\newblock \emph{Convergence of probability measures}.
\newblock John Wiley \& Sons, 1999.

\bibitem[Carrasco and Florens(2000)]{carrasco2000generalization}
Marine Carrasco and Jean-Pierre Florens.
\newblock Generalization of gmm to a continuum of moment conditions.
\newblock \emph{Econometric Theory}, 16\penalty0 (6):\penalty0 797--834, 2000.

\bibitem[Chamberlain(1987)]{chamberlain1987asymptotic}
Gary Chamberlain.
\newblock Asymptotic efficiency in estimation with conditional moment restrictions.
\newblock \emph{Journal of Econometrics}, 34\penalty0 (3):\penalty0 305--334, 1987.

\bibitem[Chernozhukov et~al.(2018)Chernozhukov, Chetverikov, Demirer, Duflo, Hansen, Newey, and Robins]{chernozhukov2018double}
Victor Chernozhukov, Denis Chetverikov, Mert Demirer, Esther Duflo, Christian Hansen, Whitney Newey, and James Robins.
\newblock Double/debiased machine learning for treatment and structural parameters.
\newblock \emph{The Econometrics Journal}, 21\penalty0 (1):\penalty0 C1--C68, 2018.

\bibitem[Dom{\'\i}nguez and Lobato(2004)]{dominguez2004consistent}
Manuel~A Dom{\'\i}nguez and Ignacio~N Lobato.
\newblock Consistent estimation of models defined by conditional moment restrictions.
\newblock \emph{Econometrica}, 72\penalty0 (5):\penalty0 1601--1615, 2004.

\bibitem[Dom{\'\i}nguez and Lobato(2015)]{dominguez2015simple}
Manuel~A Dom{\'\i}nguez and Ignacio~N Lobato.
\newblock A simple omnibus overidentification specification test for time series econometric models.
\newblock \emph{Econometric Theory}, 31\penalty0 (4):\penalty0 891--910, 2015.

\bibitem[Hsing and Eubank(2015)]{hsing2015theoretical}
Tailen Hsing and Randall~L Eubank.
\newblock \emph{Theoretical foundations of functional data analysis, with an introduction to linear operators}, volume 997.
\newblock Wiley Online Library, 2015.

\bibitem[Lavergne and Patilea(2013)]{lavergne2013smooth}
Pascal Lavergne and Valentin Patilea.
\newblock Smooth minimum distance estimation and testing with conditional estimating equations: uniform in bandwidth theory.
\newblock \emph{Journal of Econometrics}, 177\penalty0 (1):\penalty0 47--59, 2013.

\bibitem[Newey(1990)]{newey1990efficient}
Whitney~K Newey.
\newblock Efficient instrumental variables estimation of nonlinear models.
\newblock \emph{Econometrica}, 58\penalty0 (4):\penalty0 809--837, 1990.

\bibitem[Newey and McFadden(1994)]{newey1994large}
Whitney~K Newey and Daniel McFadden.
\newblock Large sample estimation and hypothesis testing.
\newblock \emph{Handbook of Econometrics}, 4:\penalty0 2111--2245, 1994.

\bibitem[Robinson(1988)]{robinson1988stochastic}
Peter~M Robinson.
\newblock The stochastic difference between econometric statistics.
\newblock \emph{Econometrica}, 56\penalty0 (3):\penalty0 531--548, 1988.

\bibitem[Robinson(1991)]{robinson1991best}
Peter~M Robinson.
\newblock Best nonlinear three-stage least squares estimation of certain econometric models.
\newblock \emph{Econometrica}, 59\penalty0 (3):\penalty0 755--786, 1991.

\bibitem[Serfling(1980)]{serfling1980approximation}
Robert~J Serfling.
\newblock \emph{Approximation theorems of mathematical statistics}.
\newblock John Wiley \& Sons, 1980.

\bibitem[Shao(1999)]{shao1999mathematical}
Jun Shao.
\newblock \emph{Mathematical Statistics}.
\newblock Springer, 1999.

\bibitem[Van~der Vaart(2000)]{van2000asymptotic}
Aad~W Van~der Vaart.
\newblock \emph{Asymptotic statistics}, volume~3.
\newblock Cambridge university press, 2000.

\end{thebibliography}
\end{document}